\documentclass[11pt]{article}

\usepackage{graphicx}
\usepackage{fullpage}
\usepackage[colorlinks=true]{hyperref}
\usepackage{xspace,amsmath,amssymb,amsthm}
\usepackage{algorithmicx,algorithm}
\usepackage{cite}
\usepackage{comment}
\usepackage[noend]{algpseudocode}
\usepackage{multirow}
\usepackage{bm}

\algnewcommand\algorithmicinput{\textbf{INPUT:}}
\algnewcommand\INPUT{\item[\algorithmicinput]}
\algnewcommand\algorithmicoutput{\textbf{OUTPUT:}}
\algnewcommand\OUTPUT{\item[\algorithmicoutput]}

\newcommand{\problem}[1]{\textsc{#1}\xspace}

\newcommand{\NMC}{\problem{Node Multiway Cut}}
\newcommand{\SFVS}{\problem{Subset Feedback Vertex Set}}

\newcommand{\INTCOVER}{\problem{0/1/all Deletion}}

\newcommand{\set}[1]{\{#1\}}

\newcommand{\F}{\mathcal{F}}
\newcommand{\A}{\mathcal{A}}
\newcommand{\T}{\mathcal{T}}
\newcommand{\prev}{\mathrm{prev}}
\newcommand{\tail}{\mathrm{tail}}
\newcommand{\bbZ}{\mathbb{Z}}
\newcommand{\imp}{\mathrm{imp}}
\newcommand{\all}{\mathbf{all}}

\newcommand{\yynote}[1]{\textcolor{magenta}{(Yutaro: #1)}}
\newcommand{\ynote}[1]{\textcolor{red}{(Yuichi: #1)}}

\newcommand{\parameterizedproblem}[4]{
\vspace{5pt}
\noindent\fbox{\begin{minipage}{0.95\linewidth}
\noindent  \textsc{#1} \hfill \textbf{Parameter:} #3\\
\textbf{Input:} #2\\
\textbf{Question:} #4
\end{minipage}}
\vspace{5pt}
}

\newtheorem{theorem}{Theorem}
\newtheorem{lemma}{Lemma}
\newtheorem{corollary}{Corollary}
\newtheorem{claim}{Claim}

\theoremstyle{definition}
\newtheorem{definition}{Definition}

\title{0/1/all CSPs, Half-Integral $A$-path Packing,\\ and Linear-Time FPT Algorithms}

\author{Yoichi Iwata\thanks{Supported by JSPS KAKENHI Grant Number JP17K12643.}\\
  National Institute of Informatics\\
  \texttt{yiwata@nii.ac.jp}
  \and
  Yutaro Yamaguchi\thanks{Supported by JSPS KAKENHI Grant Number JP16H06931 and JST ACT-I Grant Number JPMJPR16UR.}\\
  Osaka University\\
  \texttt{yutaro\_yamaguchi@ist.osaka-u.ac.jp}
  \and
  Yuichi Yoshida\thanks{Supported by JST ERATO Grant Number JPMJER1305 and JSPS KAKENHI Grant Number JP17H04676.}\\
  National Institute of Informatics\\
  \texttt{yyoshida@nii.ac.jp}
}

\date{}

\begin{document}
\maketitle

\begin{abstract}
  A recent trend in the design of FPT algorithms is exploiting the half-integrality of LP relaxations.
  In other words, starting with a half-integral optimal solution to an LP relaxation, we assign integral values to variables one-by-one by branch and bound.
  This technique is general and the resulting time complexity has a low dependency on the parameter.
  However, the time complexity often becomes a large polynomial in the input size because we need to compute half-integral optimal LP solutions.

  In this paper, we address this issue by providing an $O(km)$-time algorithm for solving the LPs arising from various FPT problems, where $k$ is the optimal value and $m$ is the number of edges/constraints.
  Our algorithm is based on interesting connections among 0/1/all constraints, which has been studied in the field of constraints satisfaction,
  $A$-path packing, which has been studied in the field of combinatorial optimization, and the LPs used in FPT algorithms.
  With the aid of this algorithm, we obtain improved FPT algorithms for various problems, including \textsc{Group Feedback Vertex Set}, \textsc{Subset Feedback Vertex Set}, \textsc{Node Multiway Cut}, \textsc{Node Unique Label Cover}, and \textsc{Non-monochromatic Cycle Transversal}.
  The obtained running time for each of these problems is linear in the input size and has the current smallest dependency on the parameter.
  In particular, these algorithms are the first linear-time FPT algorithms for problems including \textsc{Group Feedback Vertex Set} and \textsc{Non-monochromatic Cycle Transversal}.


\end{abstract}

\thispagestyle{empty}
\setcounter{page}{0}
\newpage

\tableofcontents

\thispagestyle{empty}
\setcounter{page}{0}

\newpage


\section{Introduction}

\subsection{FPT Algorithms using Half-Integral LP Relaxations}

Parameterized complexity is the subject of studying the complexity of parameterized problems.
A parameterized problem with a parameter $k$ is \emph{fixed parameter tractable (FPT)} if we can solve the problem in $f(k)\mathrm{poly}(n)$ time, where $n$ is the input size.
Various parameterized problems are known to be FPT.
See~\cite{Cygan:2015bf,Downey:2012vk} and references therein for a comprehensive list of FPT problems.

One of the motivations of studying parameterized complexity is to understand tractable subclasses of (NP-)hard problems; hence, the primary interest has been which parameterized problems admit FPT algorithms.
However, from a practical point of view, the running time with respect to the input size must also be small.
Indeed, linear-time FPT algorithms (i.e., FPT algorithms whose running times are linear in the input size) are proposed for several problems including
\textsc{Treewidth}~\cite{Bodlaender:2006kk}, \textsc{Almost 2-SAT}~\cite{Iwata:2014tq,RamanujanS14}, \textsc{Feedback Vertex Set (FVS)}~\cite{BarYehuda:2000iv}, \textsc{Subset FVS}~\cite{Lokshtanov:2015jt}, \textsc{Directed FVS}~\cite{Lokshtanov:2016tu}, and \textsc{Node Unique Label Cover}~\cite{Lokshtanov:2016ue}.
These works focused on reducing the running time with respect to the input size, and hence dependency on the parameter is often suboptimal.
For example, \textsc{Subset FVS} admits an FPT algorithm running in $O^*(4^k)$\footnote{
$O^*(\cdot)$ hides a polynomial dependency on the input size. When focusing on reducing the $f(k)$ part, the $\mathrm{poly}(n)$ part is often ignored using this notation.} time~\cite{Iwata:2016ja}
whereas the current best linear-time FPT algorithm has time complexity $O(25.6^k m)$~\cite{Lokshtanov:2015jt}, where $m$ is the number of edges in the input graph.

The half-integrality of the LP relaxations has recently been used to design FPT algorithms for a broad range of problems~\cite{Guillemot11a,Cygan:2013jv,LokshtanovNRRS14,Iwata:2016ja,Wahlstrom17,Iwata17}.
To see the idea, let us consider a minimization problem whose goal is to find a solution of size $k$, and suppose that it admits a half-integral LP relaxation\footnote{
Most of the LPs used in the FPT algorithms are not natural LP relaxations of the original problems, but are LP relaxations of \emph{rooted} problems.
For example, the rooted version of \textsc{FVS} is a problem of finding a minimum vertex subset $S$, such that the graph obtained by removing $S$ contains no cycles reachable from a prescribed vertex $s$.
Note that the existence of a half-integral LP relaxation to the rooted problem does not imply a 2-approximability of the original problem.
}, that is, an LP relaxation with an optimal LP solution that only uses values in $\{0,\frac{1}{2},1\}$.
The algorithm is based on the standard branch-and-bound framework.
First, we compute a half-integral (optimal) LP solution.
We can stop if all the variables have integral values or the sum of the values in the LP solution exceeds $k$.
Otherwise, we fix variables with values $1$.
We then pick an arbitrary variable with value $\frac{1}{2}$, and branch into the case that its value is fixed to $0$ and the case that its value is fixed to $1$.
This approach has several big advantages.
It can be applicable to various problems just by changing the LP, and has a small time complexity with respect to parameter $k$.
For many problems including \textsc{Almost 2-SAT}~\cite{LokshtanovNRRS14}, \textsc{Node Multiway Cut}~\cite{Cygan:2013jv}, and \textsc{Group FVS}~\cite{Iwata:2016ja}, the current smallest dependency on the parameter is indeed achieved by this approach.

A drawback of the abovementioned approach is that it is not trivial how to efficiently compute half-integral LP solutions.
Iwata, Wahlstr\"om, and Yoshida~\cite{Iwata:2016ja} viewed half-integrality as a discrete relaxation.
They showed that one can compute half-integral LP solutions for \textsc{Almost 2-SAT} and \textsc{(Edge) Unique Label Cover} in time linear in the input size by reducing them to the $s$-$t$ cut problem.
Moreover, they showed that we can compute LP solutions with an extremal condition, which we call the \emph{farthest condition} herein (see Section~\ref{sec:intcover} for details), in the same running time.
Using farthest solutions, the LP lower bound strictly increases for each branching.
Consequently, they obtained linear-time FPT algorithms for these problems.

The $s$-$t$ cut approach unfortunately does not work well for other problems such as \textsc{Group FVS}, \textsc{Subset FVS}, \textsc{Node Multiway Cut}, and \textsc{Node Unique Label Cover}
because it is essentially applicable only to edge-deletion problems\footnote{Note that edge-deletion problems are easily reducible to the corresponding vertex-deletion problems in most cases by subdividing the edges and creating $k$ copies of the original vertices.}
and because the auxiliary network size is generally not linear in the input size (e.g., for \textsc{Subset Feedback Edge Set}, the size of the network becomes $2^{O(m)}$).
We need to resort to solving linear programs for these problems, and hence the resulting FPT algorithms have large dependencies on the input size.
Among these problems, linear-time FPT algorithms have been obtained for \textsc{Node Multiway Cut}~\cite{ChenLL09,Iwata:2014tq}, \textsc{Subset FVS}~\cite{Lokshtanov:2015jt}, and \textsc{Node Unique Label Cover}~\cite{Lokshtanov:2016ue} by problem-specific arguments without using LP relaxations.
However, they have larger dependencies on $k$ (Table~\ref{tab:results}).

The main contribution of this study is the development of an algorithm that computes farthest half-integral solutions to those LPs in time linear in the input size.
Using the connection of computing half-integral LP solutions and FPT algorithms, we obtain linear-time FPT algorithms for various problems, which are summarized in Table~\ref{tab:results}.
In particular, for \textsc{Subset FVS}, \textsc{Node Multiway Cut}, and \textsc{Node Unique Label Cover}, we substantially improve the dependency on the parameter.
For the other problems including \textsc{Group FVS} and \textsc{Non-monochromatic Cycle Transversal}, we obtain the first linear-time FPT algorithms.
We note that, for every problem in the table, the $f(k)$ part in the running time of our algorithm matches or improves the smallest known.
All of these results are obtained by the same approach, i.e., the branch-and-bound framework combined with the efficient computation of half-integral LPs, which demonstrates its generality.

\begin{table}[t!]
  \centering
  \caption{Summary of our FPT results.
  Here, $d$ denotes the maximum domain size, $\Sigma$ is the alphabet set; $m$ denotes the number of edges/constraints in the input; and $T_\Gamma$ denotes the time complexity for performing group operations.
  All the algorithms are deterministic except for the $O(25.6^k m)$-time algorithm for \textsc{Subset FVS}.
  See Section~\ref{sec:fpt:formulation} for the problem definitions.}\label{tab:results}
  \begin{tabular}{|c|c|c|c|}
  \hline
  Problem & Smallest $f(k)$ & Existing linear-time FPT & Our result \\
  \hline
  \textsc{0/1/all Deletion} & $O^*(d^{2k})$~\cite{Iwata:2016ja} & --- & $O(d^{2k} k m)$ \\
  \hline
  \textsc{Node Unique Label Cover} & $O^*(|\Sigma|^{2k})$~\cite{Iwata:2016ja} & $|\Sigma|^{O(k|\Sigma|)}m$~\cite{Lokshtanov:2016ue}&$O(|\Sigma|^{2k}k m)$\\
  \hline
  \textsc{Two-fan Deletion} & $O^*(9^k)$~\cite{Iwata:2016ja} & --- & $O(4^k k m)$ \\
  \hline
  \textsc{Monochromatically} & \multirow{2}{*}{---} & \multirow{2}{*}{---} & \multirow{2}{*}{$O(4^k k m)$} \\
  \textsc{Orientable Deletion} & & & \\
  \hline
  \textsc{Subset Pseudoforest} & \multirow{2}{*}{---} & \multirow{2}{*}{---} & \multirow{2}{*}{$O(4^k k m)$} \\
  \textsc{Deletion} & & & \\
  \hline
  \textsc{Node Multiway Cut} & $O^*(2^k)$~\cite{Cygan:2013jv} & $O(4^km)$~\cite{ChenLL09,Iwata:2014tq}&$O(2^k k m)$ \\
  \hline
  \textsc{Group FVS} & $O^*(4^k T_\Gamma)$~\cite{Iwata:2016ja} & --- & $O(4^k k m T_{\Gamma})$ \\
  \hline
  \multirow{2}{*}{\textsc{Subset FVS}} &  \multirow{2}{*}{$O^*(4^k)$~\cite{Iwata:2016ja}} & $O(25.6^k m)$ (randomized)~\cite{Lokshtanov:2015jt} & \multirow{2}{*}{$O(4^k k m)$} \\
  & & $2^{O(k\log k)}m$ (deterministic)~\cite{Lokshtanov:2015jt} & \\
  \hline
  \textsc{Non-monochromatic} & \multirow{2}{*}{$O^*(4^k)$~\cite{Wahlstrom17}} & \multirow{2}{*}{---} & \multirow{2}{*}{$O(4^k k m)$}  \\
  \textsc{Cycle Transversal} & & & \\
  \hline
  \end{tabular}
\end{table}

\subsection{\INTCOVER and $A$-path Packing}\label{sec:01all_A-path}
We consider the following problem, called \INTCOVER, to establish a unified framework and provide linear-time FPT algorithms for the abovementioned problems using half-integral LP relaxations.
Let $V$ be a set of variables.
Each variable $v\in V$ has an individual domain $D(v)$.
A function $\varphi$ on $V$ that maps each variable $v\in V$ to a value $\varphi(v)\in D(v)$ is called an \emph{assignment} for $V$.
We consider the following two types of binary constraints on $(u,v)\in V\times V$, called \emph{0/1/all constraints}~\cite{Cooper:1994gt}\footnote{Precisely speaking, the permutation and two-fan constraints together with empty and complete constraints are obtained by enforcing arc-consistency on the 0/1/all constraints introduced in~\cite{Cooper:1994gt}.}.
\begin{enumerate}
  \item Permutation constraint $\pi(\varphi(u))=\varphi(v)$, where $\pi:D(u)\to D(v)$ is a bijection.
  \item Two-fan constraint $(\varphi(u)=a)\vee (\varphi(v)=b)$, where $a\in D(u)$ and $b\in D(v)$.
\end{enumerate}

Let $C$ be a set of 0/1/all constraints on $V$.
We assume that $C$ contains at most one constraint for each pair of distinct variables $u, v \in V$.
A constraint on $(u, v)$ in $C$ is denoted by $C_{uv}$.
For a subset $U\subseteq V$, we denote the set $\{C_{uv}\in C\mid u,v\in U\}$ by $C[U]$.

\parameterizedproblem{\INTCOVER}
{A set of variables $V$ each of which has a domain of size at most $d$, a set $C$ of 0/1/all constraints each of which is given as a table of size $O(d)$, a partial assignment $\varphi_A$ for a subset $A\subseteq V$, and an integer $k$.}
{$k,d$}
{Is there a pair of a set $X \subseteq V$ of at most $k$ variables and a partial assignment $\varphi$ for $V\setminus X$ such that
(1) $\varphi(v)=\varphi_A(v)$ holds for every $v\in A\setminus X$ and (2) $\varphi$ satisfies all of $C[V\setminus X]$?}
\\The set $X$ in the question is called a \emph{deletion set}.
The task of the optimization version of this problem is to compute the size of a minimum deletion set.
Various FPT problems can be expressed as \INTCOVER (by using a large $d$).
Note that for several problems, we need exponential-size domains, and hence a linear-time FPT algorithm for \INTCOVER does not directly imply linear-time FPT algorithms for such problems.
We show that, by giving constraints not as a table but as an oracle and by using a specialized branching strategy, we can obtain linear-time FPT algorithms even for such problems in a unified way.
See Section~\ref{sec:fpt:formulation} for details.


The \emph{primal graph} for $C$ is a simple undirected graph $G=(V,E)$
such that an edge $uv \in E$ exists if and only if a constraint $C_{uv}$ on $(u, v)$ exists.
An important property of the 0/1/all constraints is that, when fixing the value of a variable $u \in V$ to $p \in D(u)$, the set of values of $v \in V$ satisfying the constraint $C_{uv}$ is either $D(v)$ (when the constraint is a two-fan with $a=p$) or a singleton $\{q\}$ (when the constraint is a permutation with $q =\pi(p)$ or a two-fan with $a\neq p$ and $b=q$).
The latter-type implication is called a {\em unit propagation}.

When we are given a partial assignment $\varphi_A$ for a subset $A \subseteq V$,
unit propagations occur along walks in the primal graph $G$ starting at the vertices in $A$.
If the unit propagations along two different walks starting at (possibly the same) vertices in $A$
lead to a contradiction (i.e., implicate distinct singletons for the same variable), then at least one variable on the two walks must be contained in the deletion set.
The concatenation of such two walks is a walk between the vertices in $A$ (called an {\em $A$-walk}) that is said to be {\em conflicting} (see Section~\ref{sec:intcover} for a formal definition).

This observation provides a lower bound on the minimum size of a deletion set as follows.
Suppose that a deletion set $X \subseteq V$ with $|X| = k$ exists,
and let $\F = \F_{C,\varphi_A}$ be the set of all conflicting $A$-walks in the primal graph $G$
for $C$ with respect to $\varphi_A$.
Then, the remaining graph $G - X$ cannot have a walk in $\F$ (i.e., $X$ is a cover (or a hitting set) of $\F$).
Hence, the minimum size of such a cover is at most $k$.

We now consider an LP relaxation of finding a minimum cover of the set $\F$ of all conflicting $A$-walks in $G$,
called the {\em $\F$-covering problem}:
we are required to find a function $x : V \to {\mathbb R}_{\geq 0}$
minimizing the total value $|x| := \sum_{v \in V}x(v)$ under the constraint
that $\sum_{v \in W}x(v) \geq 1$ for every $W \in \F$,
where ``$\sum_{v \in W}$'' means the summation over the occurrences of vertices $v$ in $W$ considering the multiplicity (e.g., $x(v)$ is summed twice if $W$ intersects $v$ twice).
The dual of this LP is written down as follows, as the {\em $\F$-packing problem}:
we are required to find a function $y: \F \to \mathbb{R}_{\geq 0}$
maximizing the total value $|y| := \sum_{W \in \F}y(W)$ subject to $\sum_{W \in \F \colon v \in W}y(W) \leq 1$ for every $v \in V$, where we also consider the multiplicity of the occurrences of vertices in a walk in the summation ``$\sum_{W \in \F \colon v \in W}$" (e.g., $y(W)$ is summed twice if $W$ intersects $v$ twice).

We then propose an $O(kmT)$-time algorithm for finding
a pair of a {\em half-integral} $\F$-cover $x$ and a {\em half-integral} $\F$-packing $y$
with $|x| = |y| \leq \frac{k}{2}$ (if exists, and otherwise correctly concluding it),
where $m = |C|$ and $T$ is the running time of a certain oracle for simulating unit propagations
(see Section~\ref{sec:oracle} for the detail).
Note that, by LP duality, these $x$ and $y$ are both optimal solutions
to the $\F$-covering and $\F$-packing problems, respectively.

Combining our algorithm with the result in \cite{Iwata:2016ja} (see also Theorem~\ref{thm:persistency} in Section~\ref{sec:intcover}), we can obtain an FPT algorithm for \INTCOVER and, hence, FPT algorithms for other various problems.
Section~\ref{sec:linear} provides the detailed discussion for each specific problem.

\subsection{Related Work on Half-Integral $A$-path Packing}\label{sec:comparison}
When $C$ contains only the permutation constraints, the integral version of the $\F$-packing problem has been studied under the name of {\em non-returning $A$-path packing}\cite{pap2007packing,pap2008packing,yamaguchi2016packing}.
This is the current most general case of tractable (integral) path packing problems.
Our half-integral result suggests a conjecture that the integral packing of conflicting paths will also be tractable.
We also believe that ideas behind our half-integral algorithm will be useful for obtaining a faster algorithm for integral path packing problems.

For a further special case (of non-returning $A$-path packing), called \emph{internally disjoint\footnote{They can share terminals, but not inner vertex or edge.}$A$-path packing},
the integral version of the dual covering problem coincides with \textsc{Node Multiway Cut},
and its LP relaxation is used in the branch-and-bound FPT algorithm~\cite{Cygan:2013jv} and a 2-approximation algorithm~\cite{garg2004multiway} for this problem.

Several previous works in the field of combinatorial optimization can be found as regards the half-integral version of internally disjoint $A$-path packing.
Garg~et~al.~\cite{garg2004multiway} and Pap~\cite{pap2007some,pap2008strongly} found that both LPs always enjoy half-integral optimal solutions
even if each non-terminal vertex has an individual integral capacity instead of $1$.
Hirai~\cite{hirai2015dual} and Pap~\cite{pap2007some,pap2008strongly} developed algorithms
for finding such half-integral optimal solutions,
which both run in strongly polynomial time
(i.e., the number of elementary operations performed through each algorithm
does not depend on the capacity values).
One makes use of a sophisticated algorithm for the submodular flow problem~\cite{fujishige1992new} whereas the other relies on the ellipsoid method to solve LPs whose coefficient matrices only have $0, \pm1$ entries~\cite{frank1987application}.
Specializing on the uncapacitated case, Babenko~\cite{babenko2010fast} provided an $O(knm)$-time algorithm
for finding a maximum half-integral packing, where $n$ and $m$ are the numbers of vertices and edges, respectively, and $k$ denotes the optimal value, which is at most $O(n)$.
Our algorithm for the $\F$-packing/covering problems improves the previous best running time even against this (internally disjoint) special case.

\subsection{Proof Sketch}

Basically, we iteratively augment a half-integral $\F$-packing $y$,
and construct a half-integral $\F$-cover $x$ of the same size when no augmentation is possible.
Since the $\F$-covering and $\F$-packing problems are the dual LPs to each other,
this implies that $x$ and $y$ are optimal solutions to these problems.
In the case of the maximum $s$-$t$ flow, when we failed to find an augmenting path, we can construct a minimum $s$-$t$ cut by taking the set of edges on the flow that separates the vertices visited by the failed search from the unvisited vertices.
Similarly, we can construct a minimum half-integral $\F$-cover by taking the set of vertices on the half-integral $\F$-packing that separates visited vertices from the unvisited vertices.

We describe the idea behind our algorithm for $\F$-packing by showing a relation to computing half-integral (non-bipartite) matchings,
which are often called \emph{$2$-matchings} in the field of combinatorial optimization (see~\cite[Chapter 30]{schrijver2002combinatorial} for the basics).
Although we can easily obtain a maximum half-integral matching by a reduction to maximum bipartite matching~\cite{NemhauserT75}, we propose a different approach herein.
The idea behind this approach can be used for the half-integral $\F$-packing.
We focus on special half-integral matchings that consist of vertex-disjoint edges of weight $1$ and odd cycles\footnote{The number of edges in the cycle is odd.} with each edge having weight $\frac{1}{2}$.
A maximum half-integral matching with this special structure is known to always exist~\cite{Balinski65}.
In each step, we search for an alternating path (a path that alternately uses edges of weight $0$ and~$1$ and never uses the edges of weight $\frac{1}{2}$)
from the vertices not used in the current matching. 
Augmentation can be categorized into three types (see Figure~\ref{fig:matching}).

\begin{figure}[t]
  \centering
  \includegraphics[scale=0.6]{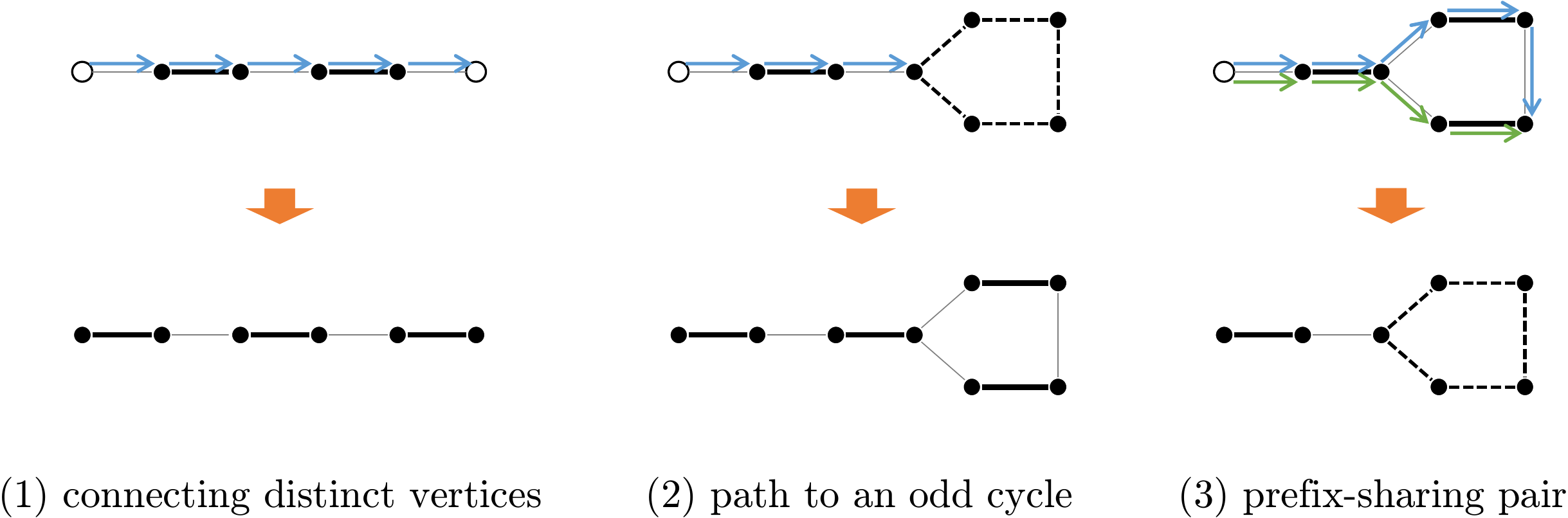}
  \caption{Augmentations for half-integral matching. Thick solid lines denote edges of weight $1$, dashed lines denote edges of weight $\frac{1}{2}$, and thin solid lines denote edges of weight $0$}
  \label{fig:matching}
\end{figure}

The first case is when we found an alternating path $P$ connecting two distinct vertices not used in the current matching.
The current matching uses all the even edges in $P$.
We can augment the matching size by $1$ by taking all the odd edges and discarding all the even edges.

The second case is when we found an alternating path $P$ of length $2a+1$ ending at a vertex on an odd cycle $C$ of length $2b+1$.
The current matching uses all the even edges in $P$ with weight $1$ and all the edges in $C$ with weight $\frac{1}{2}$.
Thus, the size of the current matching induced by $P$ and $C$ is $a+b+\frac{1}{2}$.
We can easily augment this matching to an integral matching of size $a+b+1$ by alternately taking the edges in $P$ and $C$.

The third case is when we found a pair of prefix-sharing alternating paths $P$ of odd length and $Q$ of even length ending at the same vertex.
This case corresponds to a \emph{blossom} in Edmonds' algorithm for the maximum (integral) matching problem~\cite{Edmonds:1965up}.
In Edmonds' algorithm, the alternation is applied to the common prefix and then the vertices on the cycles induced by $P$ and $Q$ are shrunk.
In our approach, we can augment the matching size by $\frac{1}{2}$ by applying the alternation to the common prefix, then by transforming the cycle induced by $P$ and $Q$ to an odd cycle with each edge having weight $\frac{1}{2}$.

We use a similar approach in our algorithm for $\F$-packing/covering.
We focus on a special type of half-integral $\F$-packings that consist of disjoint paths in $\F$ of weight $1$, called \emph{integral paths}, and \emph{wheels}, 
which are the sums of an odd number of walks in $\F$ of weight $\frac{1}{2}$ and correspond to odd cycles for matching.
See Section~\ref{sec:hi:basic} for the precise definition.
Although a maximum half-integral packing with a similar special structure for a very special case of $\F$-packing, internally disjoint $A$-path packing, is known to always exist~\cite{pap2006constructive},
the existence of such a special solution was previously unknown for any other cases.
The correctness of our algorithm provides a constructive proof of the existence.

In each step, we search for an alternating path from the vertices in $A$ not used in the current packing.
As opposed to the case of computing half-integral matchings, the alternating paths may use edges contained in wheels.
However, for simplicity, we ignore such a case in the explanation that follows.
Roughly speaking, an alternating path is a sequence of paths $P_1,\ldots,P_\ell$, where $P_i$ is a path internally disjoint from any integral path and wheel for odd $i$ and is fully contained in an integral path for even $i$.
In the alternation operation, we replace the integral path containing $P_2$ with the path obtained by concatenating $P_1$ and the prefix of the integral path containing $P_2$.
We then replace the integral path containing $P_4$ with \ldots, and so on.
The definition of alternating paths is rather complicated in ensuring that each introduced path is in $\F$.
However, it essentially plays the same role as that for the half-integral matching case.

Augmentation can be categorized into three types.
Each of which corresponds to the one for the half-integral matching case.
When we find an alternating path connecting two distinct vertices in $A$, we can augment the $\F$-packing by $1$ by applying the alternation.
When we find an alternating path ending at a vertex on a wheel that is the sum of $2a+1$ walks in $\F$ of weight $\frac{1}{2}$, we can augment the $\F$-packing by $\frac{1}{2}$ by decomposing the wheel into $a$ integral paths and introducing a new integral path.
We call these two types of alternating path as an \emph{augmenting path}.
When we find a pair of prefix-sharing alternating paths $P$ and $Q$ ending at the  same vertex (with some additional conditions), we can augment the $\F$-packing by $\frac{1}{2}$ by applying the alternation to the common prefix and then by introducing a new wheel.
We call this type of a pair of alternating paths as an \emph{augmenting pair}.

We need to test the membership in $\F$ in constant time to achieve the linear time complexity.
To this end, we exploit an observation that the algorithm only tests membership against some special walks. 

To obtain linear-time FPT algorithms, we need to compute a farthest minimum half-integral $\F$-cover, which is an $\F$-cover satisfying some extremal condition, in linear time.
In the existing work using the $s$-$t$ cut approach~\cite{Iwata:2016ja}, this is achieved by exploiting a structural property of all the minimum $s$-$t$ cuts~\cite{PQ80}.
We use a different approach because we do not have the corresponding structural property for all the minimum half-integral $\F$-covers.
Naively, we can find a farthest minimum half-integral $\F$-cover by at most $n$ computations of the minimum half-integral $\F$-cover.
We show that we can regard the whole sequence of computations as at most $2k$ computations of augmenting path/pair search by properly deciding the order of these computations using the maximum half-integral $\F$-packing.
Thus, it runs in linear time in total.

Another obstacle for obtaining linear-time FPT algorithms is the existence of two-fan constraints.
Because \textsc{2-SAT} can be expressed as \INTCOVER with $k=0$ and $d=2$, any linear-time FPT algorithm for \INTCOVER must be able to solve \textsc{2-SAT} in linear time.
The standard linear-time algorithm for \textsc{2-SAT} uses the strongly connected component decomposition of the implication graph,
and the existing linear-time FPT algorithms for \textsc{Almost 2-SAT}, a parameterized version of \textsc{2-SAT}, also use the strongly connected component decomposition of an auxiliary network~\cite{Iwata:2014tq,RamanujanS14,Iwata:2016ja}.
We cannot use this approach because the size of the auxiliary network becomes super-linear for our problems.
In our algorithm, we do not use the strongly connected component decomposition, but instead use a \emph{parallel unit-propagation}, which is an alternative linear-time algorithm for \textsc{2-SAT}.

\subsection{Comparison to Babenko's algorithm}
Because our algorithm for the $\F$-packing/covering problems improves the previous best running time even against the internally-disjoint special case,
we compare our algorithm with Babenko's algorithm to clarify the reason that we obtain such an improvement.
While both the algorithms iteratively augment a packing, the approaches are completely different.
The main difference is the existence of augmenting pairs and the algorithm for computing augmenting paths/pairs.

In our algorithm, we focus on packings with a special structure.
Both the definition of alternating paths and the algorithm for searching augmenting paths/pairs strongly rely on this structure.
While the existence of a maximum half-integral packing with the special structure was already known for internally disjoint $A$-path packing~\cite{pap2006constructive},
Babenko's algorithm does not directly exploit the structure, but uses a much weaker structure.
This is because his augmentation strategy does not preserve the special structure because it does not consider a notion corresponding to augmenting pairs of our algorithm.

In our algorithm, we directly compute an augmenting path/pair in $O(m)$ time and hence we can compute a maximum packing of size $k$ in $O(km)$ time.
In contrast, in Babenko's algorithm, an auxiliary network and its $s$-$t$ flow $f$ are constructed from the current packing, then an $f$-augmenting path for the standard maximum flow problem is computed.
From the obtained $f$-augmenting path, we can either augment the current packing or we can find a set of vertices that can be safely contracted to some vertex in $A$.
Because contractions occur at most $O(n)$ time per augmentation, the running time becomes $O(knm)$.

\subsection{Organization}
We introduce the notions used throughout the paper in Section~\ref{sec:def}.
Section~\ref{sec:hi} shows a fast algorithm that computes a maximum half-integral $\F$-packing and transforms it into a minimum half-integral $\F$-cover.
Section~\ref{sec:farthest} presents a fast algorithm for computing a farthest minimum half-integral $\F$-cover.
We provide linear-time FPT algorithms in Section~\ref{sec:linear} using this algorithm.


\section{Definitions}\label{sec:def}

\subsection{Basic Notations}

The \emph{multiplicity function} $\mathbf{1}_S:U \to \bbZ_{\geq 0}$ for a multiset $S$ on the ground set $U$ is defined such that $\mathbf{1}_S(a)$ is the number of times that $a \in U$ appears in $S$.
For two multisets $A$ and $B$ on the same ground set $U$, we denote by $A\setminus B$ the multiset such that $\mathbf{1}_{A\setminus B}(a)=\max\set{\mathbf{1}_A(a)-\mathbf{1}_B(a),0}$ holds for any element $a \in U$.
For a function $f:U\to\mathbb{R}$ and a multiset $S$ on the ground set $U$, we define
$f(S):=\sum_{a\in U}\mathbf{1}_S(a)f(a)$.
For a value $i\in\mathbb{R}$, we define $f^{-1}(i):=\{a\in U\mid f(a)=i\}$.

All the graphs in this study are undirected.
However, we sometimes need to take care of the direction of edges.
For an undirected graph $G=(V,E)$, we use the symbol $\hat{E}$ when we take care of the direction of the edges, i.e.,
$uv=vu$ for $uv\in E$ but $uv\neq vu$ for $uv\in \hat{E}$.
For simplicity, we assume that the graphs are simple; if a graph contains multiple
edges or self-loops, we can easily obtain an equivalent simple graph by subdividing the edges.
For vertex $v\in V$, we denote the set of incident edges by $\delta(v)$.
For a subset $U\subseteq V$, we denote the induced subgraph by $G[U]=(U,E[U])$.

For an undirected graph $G=(V,E)$, we define a \emph{walk} in $G$ as an ordered list $W=(v_0,\ldots,v_\ell)$ of vertices such that
$v_{i-1}v_{i}\in E$ for all $i = 1, \dots, \ell$.
The integer $\ell$ is called the \emph{length} of the walk.
We denote the first and last vertices of $W$ by $s(W)=v_0$ and by $t(W)=v_\ell$, respectively,
and we say that $W$ \emph{starts from} $s(W)$ and \emph{ends at} $t(W)$.
We denote the multisets of vertices, inner vertices, and (undirected) edges appeared in $W$ by $V(W)=\{v_0,\ldots,v_\ell\}$, by $V_{\rm in}(W) = \{v_1, \ldots, v_{\ell-1}\} = V(W) \setminus \{s(W), t(W)\}$,
and by $E(W)=\{v_0v_1,\ldots,v_{\ell-1}v_\ell\}$, respectively.
For an edge $e=uv\in \hat{E}$, we simply use the same symbol $e$ to denote the walk $(u,v)$.
A walk $W$ is called a \emph{(simple) path} if $\mathbf{1}_{V(W)}(v) \leq 1$ for every $v \in V$.
A walk $W$ is called a \emph{closed walk} if $s(W)=t(W)$.
A closed walk $W$ is called a \emph{(simple) cycle} if $\ell \geq 3$ and $\mathbf{1}_{V(W) \setminus \{s(W)\}}(v) \leq 1$ for every $v \in V$
(which implies $\mathbf{1}_{V(W)}(s(W)) = 2$).
We may regard a walk $W$ as a subgraph by ignoring the direction and the multiplicity.
We say that a walk $W$ is \emph{internally disjoint from} a subgraph $G'=(V',E')$ if none of the inner vertices of $W$ are in $V'$ and none of the edges of $W$ are in $E'$.
For a walk $W=(v_0,\ldots,v_\ell)$, we define the \emph{reversed walk} as $W^{-1}=(v_\ell,\ldots,v_0)$.
For a walk $W_1=(v_0,\ldots,v_{\ell'})$ and a walk $W_2=(v_{\ell'},\ldots,v_\ell)$ (where $0 \leq \ell' \leq \ell$),
we define the \emph{concatenation} of the two walks as $W_1\circ W_2=(v_0,\ldots,v_\ell)$.
The notation $W_1\circ W_2$ implicitly implies $t(W_1) = s(W_2)$. 

\subsection{\INTCOVER and Half-Integral Relaxation}\label{sec:intcover}
We first recall \INTCOVER defined in Section~\ref{sec:01all_A-path}.
We are given a set $V$ of variables $v$ with individual domains $D(v)$,
a set $C$ of 0/1/all constraints (permutation and two-fan constraints)
that can be represented by a simple undirected graph $G$, called the primal graph,
and a partial assignment $\varphi_A$ for a subset $A \subseteq V$.
The task is to determine whether a deletion set $X \subseteq V$
of at most $k$ variables exists, for which there exists an assignment $\varphi$ for $V \setminus X$
such that (1) $\varphi(v) = \varphi_A(v)$ holds for every $v \in A \setminus X$,
and (2) $\varphi$ satisfies every constraint $C_{uv} \in C[V \setminus X]$.

As described in Section~\ref{sec:01all_A-path},
an important property of the 0/1/all constraints is that,
when fixing the value of a variable $u \in V$ to $p \in D(u)$,
the set of values of $v \in V$ satisfying the constraint $C_{uv}$ is either $D(v)$ or a singleton $\{q\}$.
We define $C_{uv}(p):=\all$ in the former case and define $C_{uv}(p):=q$ in the latter case.
We extend this definition to walks in the primal graph as follows.
For a walk $(s)$ of length zero and an element $p\in D(s)$, we define $C_{(s)}(p):=p$.
For a walk $W=W'\circ e$ starting from $s\in V$ and an element $p\in D(s)$, we define $C_W(p):=\all$ when $C_{W'}(p)=\all$ and $C_W(p):=C_e(C_{W'}(p))$ when $C_{W'}(p)\neq\all$.
Suppose that an assignment $\varphi$ satisfying all the constraints exists.
Then, either $C_W(\varphi(s(W)))=\all$ or $\varphi(t(W))=C_W(\varphi(s(W)))$ holds for any walk $W$ in the primal graph.

Let $\varphi_A$ be a partial assignment for a subset $A\subseteq V$.
For a walk $W$ with $s(W)\in A$, we define $\imp_{\varphi_A}(W):=C_W(\varphi_A(s(W)))$, which represents the set of  assignments for $t(W)$ induced by $W$.
A walk $W$ is called a \emph{$\varphi_A$-implicational walk} if $s(W)\in A$ and $\imp_{\varphi_A}(W)\neq\all$.
A $\varphi_A$-implicational walk $W$ is called \emph{$\varphi_A$-conflicting} if $t(W)\in A$ and $\imp_{\varphi_A}(W)\neq\varphi_A(t(W))$ hold.
We omit the prefix/subscript $\varphi_A$ if it is clear from the context.
We use the following lemma.
\begin{lemma}\label{lem:def:imp}
For any two $\varphi_A$-implicational walks $P$ and $Q$ ending at the same vertex, $P\circ Q^{-1}$ is $\varphi_A$-conflicting if and only if $\imp_{\varphi_A}(P)\neq \imp_{\varphi_A}(Q)$.
\end{lemma}
\begin{proof}
We prove the lemma by induction on the length of $Q$.
When the length of $Q$ is zero, $P\circ Q^{-1}=P$ is conflicting if and only if $\imp(P)\neq \varphi_A(t(P))=\imp(Q)$.
When $Q=Q'\circ vu$, we consider three cases.

(Case 1) If $P\circ uv$ is not implicational, $C_{uv}$ must be a two-fan $(\varphi(u)=a) \vee (\varphi(v)=b)$ and $\imp(P)=a$.
Because $Q$ is implicational, we have $\imp(Q)=a$. Therefore, we have $\imp(P)=\imp(Q)$.
Because any prefix of a conflicting walk is implicational, $P\circ Q^{-1}$ is not conflicting.
Thus, the lemma holds for this case.

(Case 2) If $P\circ uv$ is implicational and $C_{uv}$ is a permutation $\pi(\varphi(u))=\varphi(v)$,
we have $\imp(P\circ uv)=\pi(\imp(P))$ and $\imp(Q')=\pi(\imp(Q))$.
Therefore, $\imp(P\circ uv)=\imp(Q')$ if and only if $\imp(P)=\imp(Q)$ holds.
Thus, from the induction hypothesis, the lemma holds for this case.

(Case 3) If $P\circ uv$ is implicational and $C_{uv}$ is a two-fan $(\varphi(u)=a)\vee (\varphi(v)=b)$,
we have $\imp(P)\neq a$, $\imp(Q)=a$, $\imp(P\circ uv)=b$, and $\imp(Q')\neq b$.
Therefore, from the induction hypothesis, $P\circ Q^{-1}$ is conflicting.
Thus, the lemma holds for this case.
\end{proof}
For two walks $P$ and $Q$ ending at the same vertex, we write $P\not\equiv Q$ if $P\circ Q^{-1}$ is conflicting and $P\equiv Q$ if $P\circ Q^{-1}$ is not conflicting, but both $P$ and $Q$ are implicational%
\footnote{When at least one of $P$ or $Q$ is not implicational, neither $P\equiv Q$ nor $P\not\equiv Q$ holds.}.
The abovementioned lemma implies that $(\equiv)$ is an equivalence relation.
This is a key property in our algorithms.
Appendix~\ref{sec:axiomatic} presents an alternative axiomatic definition of implicational/conflicting walks and show that
any path systems admitting this key property can be expressed as the set of implicational/conflicting walks with 0/1/all constraints.
We obtain the following corollaries from Lemma~\ref{lem:def:imp}.
\begin{corollary}\label{cor:def:rev}
If a walk $W$ is conflicting, then $W^{-1}$ is also conflicting.
\end{corollary}
\begin{corollary}\label{cor:def:three}
For three implicational walks $P$, $Q$, and $R$ ending at the same vertex, if $P\circ Q^{-1}$ is conflicting, then at least one of $P\circ R^{-1}$ and $R\circ Q^{-1}$ is conflicting.
\end{corollary}

We now introduce a half-integral relaxation\footnote{Originally we should define this as an LP relaxation, but we discuss it by assuming its half integrality, which is proved via our algorithm and LP duality.} of \INTCOVER.
Let $\F_{C,\varphi_A}$ denote the set of all $\varphi_A$-conflicting walks whose internal vertices do not intersect with $A$%
\footnote{This constraint is only for simplicity. From Corollary~\ref{cor:def:three}, if $P\circ Q^{-1}$ is conflicting for some walks $P$ and $Q$ ending in $A$, at least one of $P$ and $Q$ is conflicting.
Therefore, a maximum packing of conflicting walks that uses none of such walks always exists.}.
When $C$ and $\varphi_A$ are clear from the context, we simply write $\F$ to refer to $\F_{C,\varphi_A}$.
Note that from Corollary~\ref{cor:def:rev}, we can ignore the direction of walks in $\F$.
A function $x: V\to\{0,\frac{1}{2},1\}$ is called a \emph{half-integral $\F$-cover} if $x(V(W))\geq 1$ for every $W\in \F$.
The size of $x$ is defined as $|x|=x(V)$.
A function $y:\F\to\{0,\frac{1}{2},1\}$ is called a \emph{half-integral $\F$-packing} if for every vertex $v\in V$,
it holds that $\sum_{W\in \F} \mathbf{1}_{V(W)}(v)y(W)\leq 1$.
The size of $y$ is defined as $|y|=y(\F)$.
From the LP-duality, we have $|x|\geq |y|$ for any pair of half-integral $\F$-cover $x$ and half-integral $\F$-packing $y$.
Any deletion set $X$ for \INTCOVER must intersect every $\varphi_A$-conflicting walk; hence $\mathbf{1}_X$ is an integral $\F$-cover%
\footnote{Note that the converse may not hold.
For example, we have $\F=\emptyset$ when $A=\emptyset$. Therefore, $X=\emptyset$ is the minimum integral $\F$-cover.
In contrast, $X=\emptyset$ may not be a deletion set.
Thus, the half-integral relaxation does not always lead to a 2-approximation algorithm.}.
Therefore, the size of the minimum half-integral $\F$-cover provides a lower bound on the size of the minimum deletion set.
For a half-integral $\F$-cover $x$, let $R(x)$ denote the set of vertices $t$ such that an implicational walk $W$ with $x(V(W))=0$ ending at $t$ exists.

We can prove the following property called \emph{persistency} by a careful consideration of the results in~\cite{Iwata:2016ja} (see Appendix~\ref{sec:persistency} for a detailed discussion).

\begin{theorem}\label{thm:persistency}
Let $C$ be a set of 0/1/all constraints on a variable set $V$ and $\varphi_A$ be a partial assignment for a subset $A\subseteq V$.
For any minimum half-integral $\F_{C,\varphi_A}$-cover $x$, there exists a minimum deletion set $X$ containing every vertex $u$ with $x(u)=1$ but avoiding every vertex in $R(x)$.
\end{theorem}

We say that a minimum half-integral $\F$-cover $x'$ \emph{dominates} a minimum half-integral $\F$-cover $x$ if $R(x)\subsetneq R(x')$ holds.
A minimum half-integral $\F$-cover $x$ is called \emph{farthest} if there exists no minimum half-integral $\F$-cover dominating $x$.
Suppose that we have an $O(T)$-time algorithm for computing a farthest minimum half-integral $\F$-cover.
From Theorem~\ref{thm:persistency}, it is not difficult to obtain an $O(d^{2k}T)$-time FPT algorithm for \INTCOVER.
Moreover, the base $d$ of the exponent can be improved to a constant for several special cases.
We give a detailed discussion in Section~\ref{sec:fpt:algorithm}.

\subsection{Single-Branching Pair and Incremental-Test Oracle}\label{sec:oracle}
As shown in Section~\ref{sec:fpt:formulation}, various important NP-hard problems can be expressed as a special case of \INTCOVER.
However, the domain size $d$ is often $\omega(1)$ (or even $\exp(m)$ for several cases, where $m$ is the number of edges/constraints).
Hence, if every constraint is given as the table of size $d$, the total size of these tables already becomes super-linear.
For obtaining linear-time FPT algorithms, we use oracles instead of the explicit expression of the constraints to efficiently check whether a given walk is implicational/conflicting or not.

A pair $(P,Q)$ of implicational walks ending at the same vertex is called \emph{single-branching} if either $P\circ Q^{-1}$ forms a simple path or
they can be written as $P=R\circ P'$ and $Q=R\circ Q'$ for a (possibly zero-length) path $R$ and two walks $P'$ and $Q'$ for which $P'\circ Q'^{-1}$ forms a simple cycle that is internally disjoint from $R$.
In our algorithm, we need to test whether a given walk is implicational and whether $P\circ Q^{-1}$ is conflicting for a given single-branching pair $(P,Q)$.
To efficiently answer these queries, we use a tuple $(U, \mathcal{I}, \mathcal{A}, \mathcal{T})$, called an \emph{incremental-test oracle}, of a set $U$ and functions $\mathcal{I}$, $\mathcal{A}$, and $\mathcal{T}$ satisfying the following.
\begin{itemize}
  \item \emph{Init} $\mathcal{I}: A\to U$.
  \item \emph{Append} $\mathcal{A}: U\times\hat{E}\to U\cup\{\all\}$.
  		Let $\mathcal{A}^*$ be a function such that
  		\[
  			\mathcal{A}^*((s))=\begin{cases}
  			\mathcal{I}(s)&(s\in A)\\
  			\all&(s\not\in A)
  			\end{cases}\quad\mathrm{and}\quad
  			\mathcal{A}^*(W\circ e)=\begin{cases}
  			\mathcal{A}(\mathcal{A}^*(W),e)&(\mathcal{A}^*(W)\neq\all)\\
  			\all&(\mathcal{A}^*(W)=\all)
  			\end{cases}.
  		\]
  		Then, for any walk $W$, $\mathcal{A}^*(W)\neq\all$ if and only if $W$ is implicational.
  \item \emph{Test} $\mathcal{T}: U\times U\to\{\mathbf{true},\mathbf{false}\}$.
  		For any single-branching pair $(P, Q)$, $\mathcal{T}(\mathcal{A}^*(P),\mathcal{A}^*(Q))=\mathbf{true}$ if and only if $P\circ Q^{-1}$ is conflicting.
\end{itemize}

The running time of the incremental-test oracle is defined as the maximum running time of the three functions.
In general, we can naively implement the incremental-test oracle by setting $U:=\bigcup_{u\in V} D(u)$, $\mathcal{I}(s):=\varphi_A(s)$, $\mathcal{A}(x,e)=C_e(x)$, and $\mathcal{T}(x,y):=\mathbf{true}$ iff $x\neq y$.
We can obtain a constant-time oracle by this naive implementation for several cases including \NMC.
Meanwhile, the naive implementation takes $O(m)$ time for several other cases including \SFVS.
We will see in Section~\ref{sec:fpt:formulation} that we can implement a constant-time oracle for these cases by exploiting the constraint that the inputs to the test function are restrited to single-branching pairs.


\section{Half-Integral Packing and Covering}\label{sec:hi}
We prove the following theorem in this section.
\begin{theorem}\label{thm:hi}
Let $C$ be a set of 0/1/all constraints on variables $V$ and $\varphi_A$ be a partial assignment for a subset $A\subseteq V$.
Given the primal graph of $C$, the set $A$, an incremental-test oracle for $(C,\varphi_A)$, and an integer $k$,
we can compute a pair of minimum half-integral $\F_{C,\varphi_A}$-cover $x$ and maximum half-integral $\F_{C,\varphi_A}$-packing $y$ with $|x|=|y|\leq\frac{k}{2}$
or correctly conclude that the size of the minimum half-integral $\F_{C,\varphi_A}$-cover is at least $\frac{k+1}{2}$ in $O(kmT)$ time,
where $m$ is the number of constraints and $T$ is the running time of the incremental-test oracle.
\end{theorem}

Our algorithm is based on a simple augmentation strategy summarized as follows.
Starting with $y(W) = 0$ for every $W \in \F$,
we repeatedly update a half-integral $\F$-packing $y$
that always consists of only two types of conflicting walks defined in Section~\ref{sec:hi:basic}.
We search an \emph{augmenting path/pair} (see Section~\ref{sec:hi:aug}) in each iteration
using Algorithm~\ref{alg:search} described in Section~\ref{sec:hi:search}.
If one is found, we can improve the current $\F$-packing $y$
in linear time by Lemmas~\ref{lem:augment_path} and~\ref{lem:augment_pair} (Augmentation);
otherwise, as shown in Section~\ref{sec:hi:cover},
we can naturally construct a half-integral $\F$-cover of size $|y|$,
which guarantees the optimality of $y$ with the aid of the LP-duality.
Since each augmentation increases $|y|$ by at least $1 \over 2$,
the number of iterations is bounded by $k+1$.
Algorithm~\ref{alg:search} can be implemented in linear time by Lemma~\ref{lem:search},
which concludes Theorem~\ref{thm:hi}.

\subsection{Preliminaries}
\subsubsection{Basic $\F$-Packing}\label{sec:hi:basic}
In what follows, we focus on $\F$-packings that consist of only two types of conflicting walks.
One is a simple path $I\in\F$ of weight $1$, which is called an \emph{integral path}.
The other is a \emph{wheel} defined as follows
and consists of an odd number of conflicting walks of weight $1 \over 2$.

\begin{definition}\label{def:wheel}
A pair of a simple cycle $C = H_1\circ \ldots \circ H_d$ of weight $\frac{1}{2}$
and (possibly zero-length) paths $\{S_1,\ldots,S_d\}$ of weight $1$ is called a \emph{wheel}
if it satisfies the following conditions (Figure~\ref{fig:wheel}).
\begin{enumerate}
  \item $d$ is an odd positive integer.
  \item For any $i$, $S_i$ is a path from $A$ to $s(H_i)=t(H_{i-1})$ (where $H_0 = H_d$)
    that is internally disjoint from $C$.
  \item For any distinct $i$ and $j$, $S_i$ and $S_j$ share no vertices.
  \item For any $i$, $S_i\circ H_i\circ S_{i+1}^{-1}\in\F$, where $S_{d+1} = S_1$.
\end{enumerate}
The integer $d$ is called the \emph{degree} of the wheel, the cycle $C$ is called the \emph{half-integral cycle} of the wheel, and the paths $\{S_1,\ldots,S_d\}$ are called the \emph{spokes} of the wheel.

\begin{figure}[t]
  \centering
  \includegraphics[scale=1]{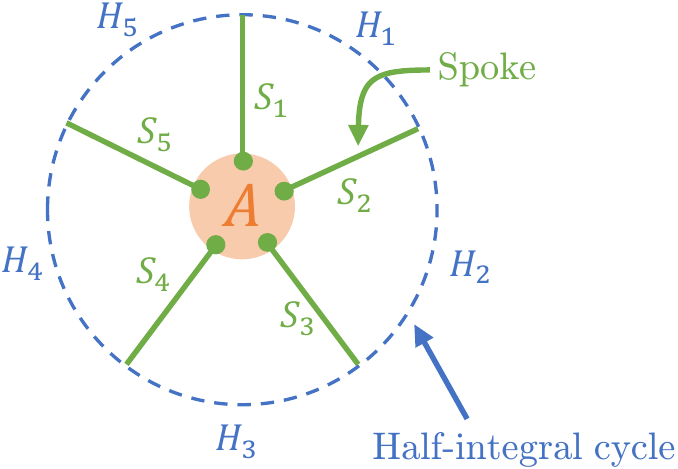}
  \caption{Wheel of degree 5.}
  \label{fig:wheel}
\end{figure}
\end{definition}

Note that a wheel of degree 1 is a closed walk $S_1\circ C\circ S_1^{-1}\in\F$ of weight $\frac{1}{2}$,
and a wheel of degree $d\geq 3$ is a sum of $d$ simple paths $\{S_1\circ H_1\circ S_2^{-1}, \ldots, S_d\circ H_d\circ S_1^{-1}\}\subseteq\F$ of weight $\frac{1}{2}$.

A half-integral $\F$-packing $y$ is called a \emph{basic $\F$-packing} if it is a sum of integral paths and wheels such that each vertex is contained in at most one of the integral paths and the wheels.
In our algorithm, a basic $\F$-packing $y$ is dealt with as a weighted graph\footnote{%
The weight is naturally defined for each vertex $v$ and edge $e$ by $\sum_{W \in \F}\mathbf{1}_{V(W)}(v)y(W)$ and $\sum_{W \in \F}\mathbf{1}_{E(W)}(e)y(W)$, respectively.}
so that we can efficiently update integral paths and wheels in $y$.
We denote by $V(y)$ and $E(y)$ the sets of vertices and (undirected) edges, respectively, that are contained in
some integral path or wheel in $y$ (i.e., of positive weights),
and particularly by $V_1(y)$ and $E_1(y)$,
the sets of those contained in some integral path or spoke in $y$ (i.e., are of weight $1$).
A walk is called internally disjoint from $y$ if it is internally disjoint from the subgraph $(V(y),E(y))$.

For a basic $\F$-packing $y$,
we define two functions $F_y$ (Forward) and $B_y$ (Backward) as follows (Figure~\ref{fig:FPB}).
Let $P$ be a positive-length path contained in an integral path $I$ in $y$.
From Corollary~\ref{cor:def:rev}, we can assume that $I$ has the same direction as $P$.
We then define paths $F_y(P)$ and $B_y(P)$ such that $I=F_y(P)\circ P\circ B_y(P)^{-1}$ holds.
For a vertex $v$ contained in a spoke $S$ in $y$, we denote by $F_y(v)$ the path from $t(S)$ to $v$ along $S$
and by $B_y(v)$ the path from $s(S)$ to $v$ along $S$ (i.e., $F_y(v)\circ B_y(v)^{-1}=S^{-1}$).
For a path $P$ contained in a spoke $S$ in $y$ in the opposite direction to $S$,
we define $F_y(P):=F_y(s(P))$ and $B_y(P):=B_y(t(P))$ (i.e., $F_y(P)\circ P\circ B_y(P)^{-1}=S^{-1}$).
We omit the subscript $y$ if it is clear from the context.

\begin{figure}[t]
  \centering
  \includegraphics[scale=0.7]{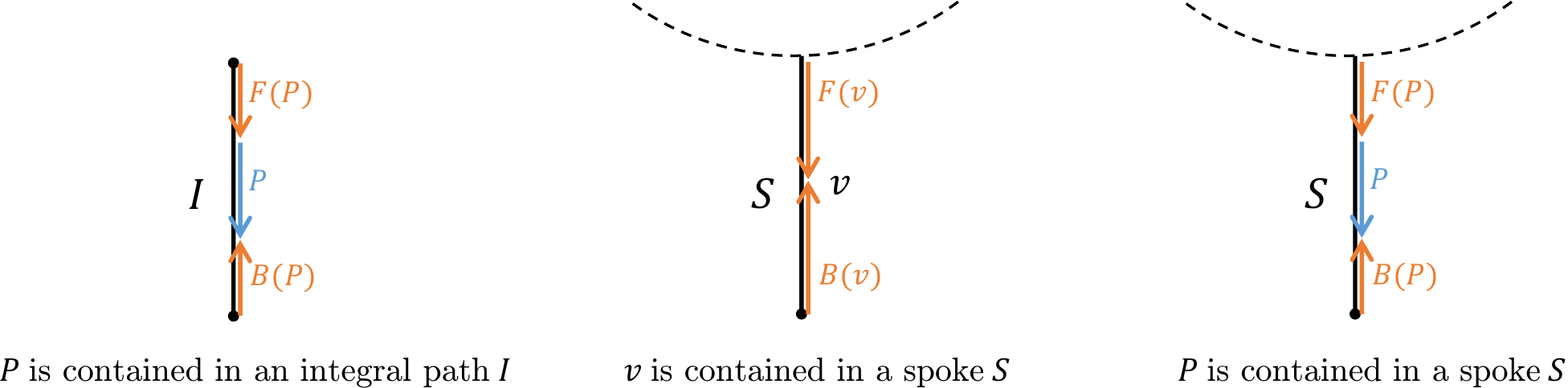}
  \caption{Definition of $F_y$ and $B_y$.}
  \label{fig:FPB}
\end{figure}

\subsubsection{Augmenting Path/Pair}\label{sec:hi:aug}
We first define an \emph{alternating path} to define our \emph{augmenting path/pair}.

\begin{definition}\label{def:alternating}
For a basic $\F$-packing $y$,
a concatenation of paths $P=P_1\circ\cdots\circ P_p$ is called a \emph{$y$-alternating path} if it satisfies all the following conditions.
\begin{enumerate}
  \item The edges in $E(P)$ are distinct (i.e., $\mathbf{1}_{E(P)}(e) \leq 1$ for every $e \in E$).\label{def:alternating:distinct}
  \item Every vertex in $P$ that is not contained in any integral path or spoke in $y$ appears in $P$ at most once%
  \footnote{From the other properties of the $y$-alternating paths, a vertex contained in an integral path or a spoke can appear twice in $P$. Hence, $P$ may
  not be a path in the precise sense.} (i.e., $\mathbf{1}_{V(P)}(v) \leq 1$ for every $v \in V \setminus V_1(y)$).\label{def:alternating:once}
  \item $s(P)\in A\setminus V(y)$.
  \item Each $P_i$ is a path of positive length satisfying the following conditions.\label{def:alternating:odd_even}
  	\begin{enumerate}
  	  \item For any odd $i$, $P_i$ is internally disjoint from $y$, and no internal vertex of $P_i$ is in $A$.
  	  \item For any even $i$, $P_i$ is contained in an integral path or a spoke in $y$. In the latter case, $P_i$ has the opposite direction to the spoke.
  	\end{enumerate}
  \item Let us define $B(P_0):=(s(P))$. The following conditions are satisfied for any $i$\footnote{%
  Technically, the conditions ``none of the $P_j$'s are contained in \ldots'' mean that $P$ does not admit a \emph{shortcut}
  (e.g., if there is some $P_j$ with $j>i$ contained in $B(P_i)$, we can obtain another $y$-alternating path $P_1\circ\cdots\circ P_{i-1}\circ W\circ P_{j+1}\circ\cdots P_p$,
  where $W$ is the path from $s(P_i)$ to $t(P_j)$ along the integral path or spoke).
  As in the case of matroid intersection, we need a shortcut-less alternating path.
  } (see Figure~\ref{fig:alternating}).\label{def:alternating:cycle_spoke}
  	\begin{enumerate}
  	  \item If $i$ is odd, $B(P_{i-1})\circ P_i$ is implicational.
  	  \item If $i$ is even and $P_i$ is contained in an integral path, then $B(P_{i-2})\circ P_{i-1}\not\equiv F(P_i)$, and none of the $P_j$'s are contained in $B(P_i)$ for $j>i$.
  	  		Moreover, if $B(P_{i-2})\circ P_{i-1}\not\equiv B(P_i)\circ P_i^{-1}$, none of the $P_j$'s are contained in $F(P_i)$ for $j>i$.\label{def:alternating:cycle}
  	  \item If $i$ is even and $P_i$ is contained in a spoke, then $B(P_{i-2})\circ P_{i-1}\equiv B(P_i)\circ P_i^{-1}$, and none of the $P_j$'s are contained in $B(P_i)$ for $j>i$.\label{def:alternating:spoke}
  	\end{enumerate}
  	\begin{figure}[t]
      \centering
      \includegraphics[scale=0.8]{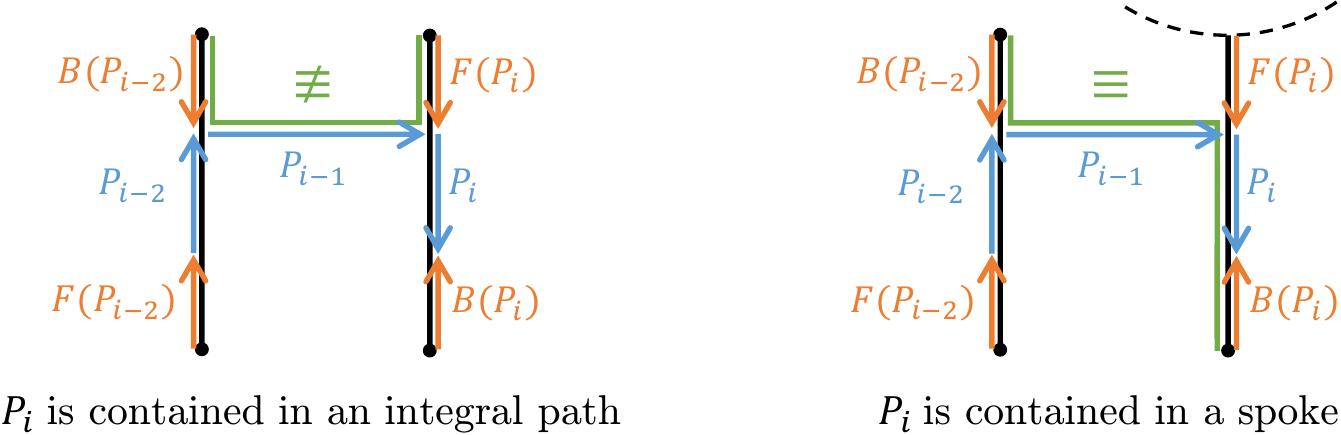}
      \caption{Conditions for alternating paths.}
      \label{fig:alternating}
	\end{figure}
\end{enumerate}
Each $P_i$ is called a \emph{segment} of $P$.
If $P$ consists of only a single segment, it is called \emph{single-segment}.
A zero-length path $P=(s)$ with $s\in A\setminus V(y)$ is considered as a $y$-alternating path with zero segments.
\end{definition}

We define $T_y(P)$ (Tail) for a $y$-alternating path $P=P_1\circ\cdots\circ P_p$ as follows:
$T_y(P):=B_y(P_{p-1})\circ P_p$ if $p$ is odd, and $T_y(P):=B_y(P_p)$ if $p$ is even.
Note that $T_y(P)$ is always implicational.
We omit the subscript $y$ if it is clear from the context.

\begin{definition}\label{def:augment_path}
A $y$-alternating path $P=P_1\circ \ldots \circ P_p$ is called a \emph{$y$-augmenting path} if $p$ is odd and one of the following conditions is satisfied.
\begin{enumerate}
  \item $t(P)\in A\setminus V(y)$, and $T(P)$ is conflicting.\label{def:augment_path:0}
  \item $t(P)$ is contained in a half-integral cycle, but in no spokes (i.e., $t(P) \in V(y) \setminus V_1(y)$).\label{def:augment_path:1}
  \item $t(P)$ is contained in a spoke $S$, and the following two conditions are satisfied.\label{def:augment_path:2}
  		\begin{enumerate}
  		  \item $T(P)\not\equiv B(t(P))$.\label{def:augment_path:21}
  		  \item For any $P_j$ contained in the spoke $S$, $t(P)$ is contained in $F(P_j)$.\label{def:augment_path:22}
  		\end{enumerate}
\end{enumerate}
\end{definition}

\begin{definition}\label{def:augment_pair}
Let $P=P_1\circ\cdots\circ P_p$ and $Q=Q_1\circ\cdots\circ Q_{q}$ be a pair of $y$-alternating paths ending at the same vertex.
$(P,Q)$ is called a \emph{$y$-augmenting pair} if it satisfies all the following conditions.
\begin{enumerate}
  \item $P$ and $Q$ can be written as $P=R\circ P'$ and $Q=R\circ Q'$, respectively, for some walk $R$ such that $P'$ and $Q'$ share no edges.\label{def:augment_pair:1}
  \item $T(P)\not\equiv T(Q)$.\label{def:augment_pair:2}
  \item At least one of $p$ and $q$ is odd, and if both of $p$ and $q$ are odd, $t(P)\not\in V(y)$.\label{def:augment_pair:3}
  \item For any $P_i$ contained in an integral path, none of the $Q_j$'s are contained in $B(P_i)$ in the opposite direction.
  		Moreover, none of the $Q_j$'s are contained in $F(P_i)$ in the same direction if $B(P_{i-2})\circ P_{i-1}\not\equiv B(P_i)\circ P_i^{-1}$.
  		The symmetric condition holds for any $Q_i$ contained in an integral path.\label{def:augment_pair:4}
\end{enumerate}
\end{definition}

Note that for any $y$-alternating path $P$, the condition~\ref{def:augment_pair:4} is always satisfied against $(P,P)$.
Therefore, in testing the condition~\ref{def:augment_pair:4} against $(P, Q)$,
we only need to test pairs $(P_i, Q_j)$ such that at least one of $P_i$ or $Q_j$ is not contained in the common prefix
$R$.

Using a $y$-augmenting path/pair,
we can improve a basic $\F$-packing $y$ in linear time by the following lemmas,
whose proofs are given in Section~\ref{sec:hi:augmentation}.

\begin{lemma}\label{lem:augment_path}
Given a basic $\F$-packing $y$ and a $y$-augmenting path,
a basic $\F$-packing of size at least $|y|+\frac{1}{2}$ can be constructed in linear time.
\end{lemma}

\begin{lemma}\label{lem:augment_pair}
Given a basic $\F$-packing $y$ and a $y$-augmenting pair,
a basic $\F$-packing of size $|y|+\frac{1}{2}$ can be constructed in linear time.
\end{lemma}

\begin{figure}[t]
  \centering
  \includegraphics[scale=0.6]{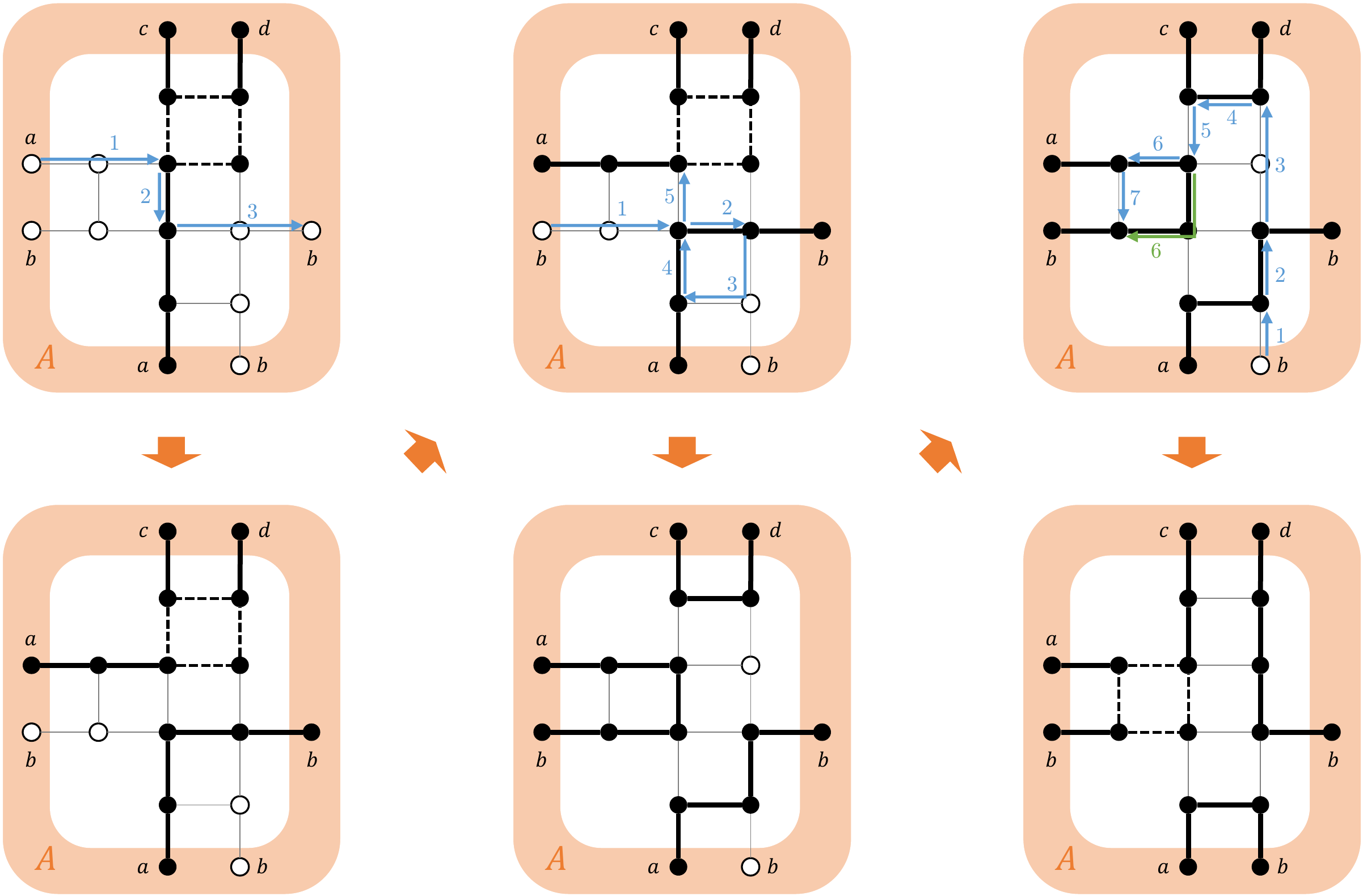}
  \caption{Example of augmentations. 
  Thick solid lines denote edges of weight $1$, dashed lines denote edges of weight $\frac{1}{2}$, and thin solid lines denote edges of weight $0$.
  Blue or green arrows denote segments of augmenting paths/pairs. A number $i$ beside an arrow means that it is the $i$-th segment.}
  \label{fig:example}
\end{figure}

Figure~\ref{fig:example} illustrates an example of augmentations.
We use the following example which corresponds to the internally-disjoint $A'$-paths with $A'=\{a,b,c,d\}$.
Every vertex $u$ has the same domain $D(u)=\{a,b,c,d\}$ and each edge $uv$ has an equality (identity permutation) constraint $\varphi(u)=\varphi(v)$.
An alphabet beside a vertex $u\in A$ shows the value $\varphi_A(u)$.
Any walks starting from $A$ is implicational, and a walk $W$ is conflicting if and only if $\varphi_A(s(W))\neq\varphi_A(t(W))$.
Roughly speaking, for an augmenting path $P$ connecting two distinct vertices in $A$ (i.e.,  when it satisfies the condition~\ref{def:augment_path:0} of Definition~\ref{def:augment_path}),
we take the symmetric difference of $E(y)$ and $E(P)$ (Figure~\ref{fig:example}, left).
For an augmenting path $P$ ending at a vertex on a wheel
(i.e., when it satisfies the condition~\ref{def:augment_path:1} or~\ref{def:augment_path:2} of Definition~\ref{def:augment_path}),
we take the symmetric difference and then decompose the wheel into integral paths (Figure~\ref{fig:example}, middle).
For an augmenting pair $(R\circ P',R\circ Q')$ with a common prefix $R$, we take the symmetric difference of $E(y)$ and $E(R)$, and then introduce a new wheel whose half-integral cycle is $P'\circ Q'^{-1}$ (Figure~\ref{fig:example}, right).
In addition to this basic augmentation strategy, we need several operations in ensuring that the obtained packing is a basic $\F$-packing.
See the proofs in Section~\ref{sec:hi:augmentation} for more detail.

\subsection{Finding Augmenting Path/Pair}\label{sec:hi:search}

\begin{algorithm}[t!]
\caption{Algorithm for computing a $y$-augmenting path/pair}
\label{alg:search}
\begin{algorithmic}[1]
\State Initialize $a(I)\gets 0$ and $b(I)\gets\ell$ for each integral path $I$ of length $\ell$.
\State Initialize $a(S)\gets 0$ for each spoke $S$.
\For{$s\in A\setminus V(y)$}\label{line:search:outer}
	\State $P(s)\gets (s)$ and $X\gets\{s\}$.\label{line:search:pickA}
	\While{$X\neq\emptyset$}\label{line:search:main}
		\State Pick a vertex $u\in X$ and remove $u$ from $X$.\label{line:search:pick}
		\For{$e=uv\in\delta(u)\setminus E(y)$ such that $T(P(u))\circ e$ is implicational}\label{line:search:loop}
			\If{$v$ is visited}
				\If{$T(P(u))\circ e\not\equiv T(P(v))$}
					\Return $(P(u)\circ e, P(v))$\label{line:search:pair}
				\EndIf
			\ElsIf{$v$ is contained in an integral path $I=(v_0,\ldots,v_\ell)$}
				\State Let $i$ be the index such that $v_i=v$.
				\If{$T(P(u))\circ e\not\equiv (v_\ell,\ldots,v_i)$}\label{line:search:cycle:begin1}
					\For{$j\in\{a(I),\ldots,i-1\}$}\label{line:search:cycle:loop1}
						\State $P(v_j)\gets (P(u)\circ e)\circ (v_i,\ldots,v_j)$ and $X\gets X\cup\{v_j\}$.\label{line:search:cycle}
					\EndFor
					\State $a(I)\gets i$.\label{line:search:cycle:end1}
				\EndIf
				\If{$T(P(u))\circ e\not\equiv (v_0,\ldots,v_i)$}\label{line:search:cycle:begin2}
					\For{$j\in\{i+1,\ldots,b(I)\}$}\label{line:search:cycle:loop2}
						\State $P(v_j)\gets (P(u)\circ e)\circ (v_i,\ldots,v_j)$ and $X\gets X\cup\{v_j\}$.
					\EndFor
					\State $b(I)\gets i$.\label{line:search:cycle:end2}
				\EndIf
			\ElsIf{$v$ is contained in a spoke $S=(v_0,\ldots,v_\ell)$}
				\State Let $i$ be the index such that $v_i=v$.
				\If{$T(P(u))\circ e\not\equiv (v_0,\ldots,v_i)$}
					\Return $P(u)\circ e$\label{line:search:path2}
				\EndIf
				\For{$j\in\{a(S),\ldots,i-1\}$}\label{line:search:spoke:begin}
					\State $P(v_j)\gets (P(u)\circ e)\circ (v_i,\ldots,v_j)$ and $X\gets X\cup\{v_j\}$.
				\EndFor
				\State $a(S)\gets i$.\label{line:search:spoke:end}
			\ElsIf{$v$ is contained in a half-integral cycle}
				\State \Return $P(u)\circ e$\label{line:search:path1}
			\ElsIf{$v\in A$}
				\If{$T(P(u))\circ e\not\equiv (v)$}
					\Return $P(u)\circ e$\label{line:search:path0}
				\EndIf
			\Else\enspace ($v\not\in A\cup V(y)$)
				\State $P(v)\gets P(u)\circ e$ and $X\gets X\cup\{v\}$.\label{line:search:odd}
			\EndIf
		\EndFor
	\EndWhile
\EndFor
\State \Return NO
\end{algorithmic}
\end{algorithm}

In this subsection, we propose an algorithm for computing a $y$-augmenting path or pair.
Algorithm~\ref{alg:search} describes a rough sketch of the algorithm.
We will describe later the details of an efficient implementation using the incremental-test oracle.
Note that in this subsection, we prove only the soundness of the algorithm (i.e., the algorithm never returns a path
or pair which is not $y$-augmenting).
The completeness of the algorithm (i.e., the algorithm always finds a $y$-augmenting path or pair if exists) follows
from Lemma~\ref{lem:cover} proved in the next subsection.

\begin{lemma}\label{lem:alg:soundness}
Any path returned by Algorithm~\ref{alg:search} is a $y$-augmenting path,
and any pair returned by the algorithm is a $y$-augmenting pair.
\end{lemma}

In the algorithm, we hold two indices $a(I)$ and $b(I)$ for each integral path $I$ initialized as
$a(I)=0$ and $b(I)=\ell$, respectively, where $\ell$ is the length of $I$; an index $a(S)$ for each spoke $S$ initialized as
$a(S)=0$; a set of active vertices $X$ initialized as $X=\emptyset$; and a $y$-alternating path $P(v)$ for each \emph{visited} vertex~$v$,
where we call a vertex $v$ visited if $v$ has been ever pushed into $X$.
We define \emph{boundaries} as the set of vertices consisting of $v_{a(I)}$ and $v_{b(I)}$ for each integral path $I=(v_0,\ldots,v_\ell)$
and $v_{a(S)}$ for each spoke $S=(v_0,\ldots,v_\ell)$.
We preserve the following invariants during the execution of the algorithm.

\begin{lemma}\label{lem:alg:invariants}
The following invariants hold at any iteration of Algorithm~\ref{alg:search}.
\begin{enumerate}
  \item For any visited $v$, the following holds:\label{lem:alg:invariants:alt}
  \begin{enumerate}
    \item $P(v)$ is a $y$-alternating path,\label{lem:alg:invariants:alt:alt}
    \item $V(P(v))$ contains only visited vertices or boundaries, and\label{lem:alg:invariants:alt:visited}
    \item $P(v)$ has an odd number of segments if and only if $v\not\in A\cup V(y)$.\label{lem:alg:invariants:alt:odd}
  \end{enumerate}
  \item For any visited $u$ and $v$ with $s(P(u))=s(P(v))$, $(P(u), P(v))$ satisfies the condition~\ref{def:augment_pair:1} of $y$-augmenting pairs (Definition~\ref{def:augment_pair}).\label{lem:alg:invariants:pair}
  \item For any integral path $I=(v_0,\ldots,v_\ell)$, the following holds:\label{lem:alg:invariants:cycle}
  \begin{enumerate}
    \item $a(I)\leq b(I)$,\label{lem:alg:invariants:cycle:ab}
    \item $v_i$ is visited if and only if $i<a(I)$ or $b(I)<i$, and\label{lem:alg:invariants:cycle:visited}
	\item for any $P(v)$ with segments $P_1\circ\cdots\circ P_p$ and any segment $P_i$ contained in $I$, $P_i$ is contained in
	$(v_{a(I)},\ldots,v_0)$ or $(v_{b(I)},\ldots,v_\ell)$ in these directions;
	moreover, if $B(P_{i-2})\circ P_{i-1}\not\equiv B(P_i)\circ P_i^{-1}$ holds, then $s(P_i)=v_{a(I)}=v_{b(I)}$
	holds.\label{lem:alg:invariants:cycle:segment}
  \end{enumerate}
  \item For any spoke $S=(v_0,\ldots,v_\ell)$, the following holds:\label{lem:alg:invariants:spoke}
  \begin{enumerate}
    \item $v_i$ is visited if and only if $i<a(S)$, and
	\item for any $P(v)$, all the segments of $P(v)$ contained in $S$ are contained in $(v_{a(S)},\ldots,v_0)$.
  \end{enumerate}
\end{enumerate}
\end{lemma}

We now prove Lemma~\ref{lem:alg:soundness} and~\ref{lem:alg:invariants}.
All the invariants are clearly satisfied at the beginning.
In each iteration, we pick an arbitrary vertex $u$ from $X$ (line~\ref{line:search:pick}).
From the invariant~\ref{lem:alg:invariants:alt}, $P(u)$ is a $y$-alternating path.
We then iterate over the edges $e=uv\in\delta(u)\setminus E(y)$ such that $T(P(u))\circ e$ is implicational.

When $v$ is already visited and $T(P(u))\circ e\not\equiv T(P(v))$ holds, we return a pair $(P(u)\circ e,P(v))$
(line~\ref{line:search:pair}).
\begin{claim}
$(P(v),P(u)\circ e)$ returned at line~\ref{line:search:pair} is a $y$-augmenting pair.
\end{claim}
\begin{proof}
First, we prove $e\not\in E(P(u))$, which implies that $P(u)\circ e$ is a $y$-alternating path.
If $e\in E(P(u))$, $e^{-1}$ is either the last edge of $P(u)$ or the last edge of an odd segment $P_i$ of
$P(u):=P_1\circ\cdots\circ P_p$.
In the former case, $T(P(u))\circ e=T(P(v))\circ e^{-1}\circ e\equiv T(P(v))$ holds, which is a contradiction.
In the latter case, from the invariant~\ref{lem:alg:invariants:pair} against $(P(u),P(v))$, we have $P(v)=P_1\circ\cdots\circ P_{i-1}\circ
W$, where $W$ is the path satisfying $W\circ e^{-1}=P_i$.
Because $P_{i+1}$ and $P_p$ are contained in the same integral path or spoke in the same direction,
we have $B(P_{i-1})\circ P_i\equiv B(P_{i+1})\circ P_{i+1}^{-1}$ from the condition~\ref{def:alternating:cycle_spoke} of $y$-alternating
paths.
Therefore, we have $T(P(u))\circ e=B(P_p)\circ e=B(P_{i+1})\circ P_{i+1}^{-1}\circ e\equiv B(P_{i-1})\circ P_i\circ e=B(P_{i-1})\circ
W\circ e^{-1}\circ e=T(P(v))\circ e^{-1}\circ e\equiv T(P(v))$, which is a contradiction.

Next, we prove that $(P(v),P(u)\circ e)$ satisfies all the conditions of the $y$-augmenting pairs (Definition~\ref{def:augment_pair}).
Suppose that $s(P(u))\neq s(P(v))$ holds.
Then, the vertex $v$ has already been popped from $X$ and, thus, we have $P(u)=P(v)\circ e^{-1}$, which is a contradiction.
Therefore, from the invariant~\ref{lem:alg:invariants:pair}, the condition~\ref{def:augment_pair:1} is satisfied.
The condition~\ref{def:augment_pair:2} is satisfied because $T(P(v))\not\equiv T(P(u))\circ e=T(P(u)\circ e)$.
Meanwhile, the condition~\ref{def:augment_pair:3} is satisfied because the number of segments of $P(u)\circ e$ is odd and because, from the invariant~\ref{lem:alg:invariants:alt},
$v\in V(y)$ implies that the number of segments of $P(v)$ is even.
The condition~\ref{def:augment_pair:4} follows from the invariant~\ref{lem:alg:invariants:cycle}.
\end{proof}

When $v$ is not visited, we consider five cases:
(Case 1) $v$ is contained in an integral path;
(Case 2) $v$ is contained in a spoke;
(Case 3) $v$ is contained in a half-integral cycle (but not in spokes);
(Case 4) $v\in A$ (but not in $V(y)$); or
(Case 5) $v\not\in A\cup V(y)$.
Note that $e$ is not contained in $E(P(u))$ when $v$ is not visited.
Therefore, $P(u)\circ e$ is a $y$-alternating path.
In case 3, we return $P(u)\circ e$, which is a $y$-augmenting path satisfying the condition~\ref{def:augment_path:1} of Definition~\ref{def:augment_path}.
In case 4, if $T(P(u))\circ e\not\equiv (v)$ holds, we return $P(u)\circ e$, which is a $y$-augmenting path satisfying the condition~\ref{def:augment_path:0} of Definition~\ref{def:augment_path}, and otherwise, we do nothing.
In case 5, we set $P(v)\gets P(u)\circ e$ and insert $v$ into $X$ (line~\ref{line:search:odd}).
Because $e\not\in E(y)$ holds, all the invariants are clearly preserved.
Finally, we consider the remaining two cases.

\paragraph{(Case 1)}
Let $I=(v_0,\ldots,v_\ell)$ be the integral path containing $v$ and let $i$ be the index such that $v_i=v$.
Because $v$ is not visited, $a(I)\leq i\leq b(I)$ holds from the invariant~\ref{lem:alg:invariants:cycle}.
If $T(P(u))\circ e\not\equiv (v_\ell,\ldots,v_i)$ holds,
we set $P(v_j)\gets (P(u)\circ e)\circ (v_i,\ldots,v_j)$ and insert $v_j$ into $X$ for each index $j\in\{a(I),\ldots,i-1\}$,
and then update $a(I)\gets i$ (lines~\ref{line:search:cycle:begin1}--\ref{line:search:cycle:end1}).
Similarly, if $T(P(u))\circ e\not\equiv (v_0,\ldots,v_i)$ holds,
we set $P(v_j)\gets (P(u)\circ e)\circ (v_i,\ldots,v_j)$ and insert $v_j$ into $X$ for each index $j\in\{i+1,\ldots,b(I)\}$,
and then update $b(I)\gets i$ (lines~\ref{line:search:cycle:begin2}--\ref{line:search:cycle:end2}).
\begin{claim}
The lines~\ref{line:search:cycle:begin1}--\ref{line:search:cycle:end2} preserve all the invariants.
\end{claim}
\begin{proof}
The invariants~\ref{lem:alg:invariants:alt:visited},~\ref{lem:alg:invariants:alt:odd},~\ref{lem:alg:invariants:pair},
\ref{lem:alg:invariants:cycle:ab},~\ref{lem:alg:invariants:cycle:visited}, and~\ref{lem:alg:invariants:spoke} are clearly preserved.
Let $a(I)$ and $b(I)$ denote the indices before the updates.
For an index $j\in\{a(I),\ldots,i-1\}$, let $P_1\circ\cdots\circ P_p$ be the segments of $P(v_j)$.
If $B(P_{p-2})\circ P_{p-1}\not\equiv B(P_p)\circ P_p^{-1}$ holds, we have $T(P(u))\circ e=B(P_{p-2})\circ P_{p-1}\not\equiv B(P_p)\circ
P_p^{-1}=(v_0,\ldots,v_i)$.
Therefore, both $a(I)$ and $b(I)$ are updated to $i$.
The same argument applies to the case of $j\in\{i+1,\ldots,b(I)\}$.
Thus, the invariant~\ref{lem:alg:invariants:cycle:segment} is preserved.

Finally, we prove the invariant~\ref{lem:alg:invariants:alt:alt} by showing that $P(v_j)$ is a $y$-alternating path for any newly
visited~$v_j$.
The conditions~\ref{def:alternating:distinct}--\ref{def:alternating:odd_even} of $y$-alternating paths (Definition~\ref{def:alternating}) are clearly satisfied.
Consider the case of $j\in\{a(I),\ldots,i-1\}$. The proof for the case of $j\in\{i+1,\ldots,b(I)\}$ is symmetric.
Let $P_1\circ\cdots\circ P_p$ be the segments of $P(v_j)$.
Because $B(P_{p-2})\circ P_{p-1}=T(P(u))\circ e\not\equiv (v_\ell,\ldots,v_i)=F(P_p)$ holds, the condition~\ref{def:alternating:cycle_spoke} is satisfied for $p$.
Let $P_k$ be a segment with $k<p$ contained in the same integral path $I$.
From the invariant~\ref{lem:alg:invariants:cycle:segment} and because $a(I)\leq j<i\leq b(I)$, $P_p=(v_i,\ldots,v_j)$ is contained in
$F(P_k)$, and $B(P_{k-2})\circ P_{k-1}\equiv B(P_k)\circ P_k^{-1}$ holds.
Therefore, the condition~\ref{def:alternating:cycle_spoke} is satisfied for $k$.
\end{proof}

\paragraph{(Case 2)}
Let $S=(v_0,\ldots,v_\ell)$ be the spoke containing $v$ and let $i$ be the index such that $v_i=v$.
Because $v$ is not visited, $a(S)\leq i$ holds from the invariant~\ref{lem:alg:invariants:spoke}.
If $T(P(u))\circ e\not\equiv (v_0,\ldots,v_i)$ holds, we return a path $P(u)\circ e$ (line~\ref{line:search:path2});
otherwise, we set $P(v_j)\gets (P(u)\circ e)\circ (v_i,\ldots,v_j)$ and insert $v_j$ into $X$ for each index $j\in\{a(S),\ldots,i-1\}$,
and then update $a(S)\gets i$ (lines~\ref{line:search:spoke:begin}--\ref{line:search:spoke:end}).
\begin{claim}
$P(u)\circ e$ returned at line~\ref{line:search:path2} is a $y$-augmenting path.
\end{claim}
\begin{proof}
Because $T(P(u))\circ e\not\equiv (v_0,\ldots,v_i)=B(t(P(u)\circ e))$ holds, the condition~\ref{def:augment_path:21} of the $y$-augment paths
(Definition~\ref{def:augment_path}) is satisfied.
From the invariant~\ref{lem:alg:invariants:spoke}, for any segment $P_k$ of $P(u)\circ e$ contained in the spoke $S$, $t(P(u)\circ
e)=v_i$ is contained in $F(P_k)$.
Therefore, the condition~\ref{def:augment_path:22} is satisfied.
\end{proof}

\begin{claim}
Lines~\ref{line:search:spoke:begin}--\ref{line:search:spoke:end} preserve all the invariants.
\end{claim}
\begin{proof}
All the invariants, except for the~\ref{lem:alg:invariants:alt:alt}, are clearly preserved.
We prove the invariant~\ref{lem:alg:invariants:alt:alt} by showing that $P(v_j)$ is a $y$-alternating path for any newly
visited~$v_j$.
The conditions~\ref{def:alternating:distinct}--\ref{def:alternating:odd_even} of the $y$-alternating paths (Definition~\ref{def:alternating}) are clearly satisfied.
Let $P_1\circ\cdots\circ P_p$ be the segments of $P(v_j)$.
Because $B(P_{p-2})\circ P_{p-1}=T(P(u)\circ e)\equiv (v_0,\ldots,v_i)=B(P_p)\circ P_p^{-1}$ holds, the
condition~\ref{def:alternating:cycle_spoke} is satisfied for $p$.
Let $P_k$ be a segment with $k<p$ contained in the same spoke $S$.
From the invariant~\ref{lem:alg:invariants:spoke}, $P_k$ must be contained in $(v_{a(S)},\ldots,v_0)$.
Therefore, $P_p$ is contained in $F(P_k)$.
Thus, the condition~\ref{def:alternating:cycle_spoke} is satisfied for $k$.
\end{proof}

\subsubsection*{Implementation Detail}

We exploit the incremental-test oracle as follows to achieve the linear-time complexity.
For each integral path $I=(v_0,\ldots,v_\ell)$ and for each index $i\in\{0,\ldots,\ell\}$,
we precompute $\A^*((v_0,\ldots,v_i))$ and $\A^*((v_\ell,\ldots,v_i))$.
For each spoke $S=(v_0,\ldots,v_\ell)$ and for each index $i\in\{0,\ldots,\ell\}$,
we precompute $\A^*((v_0,\ldots,v_i))$.
These precomputations can be done in $O(n T)$ time, where $T$ is the running time of the oracle.

For each visited vertex $v$, instead of explicitly holding $P(v)$, we hold (1) $\prev(v)$ that represents the edge picked at
line~\ref{line:search:loop} in the iteration when $P(v)$ is assigned and (2) $\tail(v):=\A^*(T(P(v)))$.
When the algorithm finds a $y$-augmenting path or pair, we restore $P(v)$ using the table of $\prev(v)$ in $O(n)$ time.
For an edge $e=uv\in\delta(u)\setminus
E(y)$, we have $\A^*(T(P(u))\circ e)=\A(\tail(u), e)$, which can be computed in $O(T)$ time.
For any index $j$ picked at line~\ref{line:search:cycle:loop1} or~\ref{line:search:spoke:begin},
we have $\tail(v_j)=\A^*(T(P(v_j)))=\A^*((v_0,\ldots,v_j))$, which has been precomputed.
For any index $j$ picked at line~\ref{line:search:cycle:loop2},
we have $\tail(v_j)=\A^*(T(P(v_j)))=\A^*((v_\ell,\ldots,v_j))$, which has been precomputed.
When $v\not\in A\cup V(y)$, we have $\tail(v)=\A(\tail(u),e)$.
Thus, we can compute $\tail(w)$ in $O(T)$ time for each newly visited vertex~$w$.

Next, we show that each of the equivalence tests in the algorithm can be done in $O(T)$ time.
Let $e=uv$ be the edge picked at line~\ref{line:search:loop}.
First, we consider the equivalence test at line~\ref{line:search:pair}.
If $\prev(v)=e^{-1}$ holds, we have $T(P(u))\circ e\equiv T(P(v))$;
otherwise, $(T(P(u))\circ e,T(P(v)))$ forms a single-branching pair from the invariant~\ref{lem:alg:invariants:pair},
for which we can test $T(P(u))\circ e\not\equiv T(P(v))$ by asking $\T(\A(\tail(u),e),\tail(v))$ in $O(T)$ time.
Next, we consider the equivalence test at line~\ref{line:search:cycle:begin1}.
Because $(T(P(u)\circ e),(v_\ell,\ldots,v_i))$ forms a single-branching pair, we can test
$T(P(u))\circ e\not\equiv (v_\ell,\ldots,v_i)$ by asking $\T(\A(\tail(u),e),\A^*((v_\ell,\ldots,v_i)))$.
This can be done in $O(T)$ time because $\A^*((v_\ell,\ldots,v_i))$ has been precomputed.
The same argument applies to the equivalence tests at lines~\ref{line:search:cycle:begin2} and~\ref{line:search:path2}.
The equivalence test at line~\ref{line:search:path0} can be done in $O(T)$ time by asking $\T(\A(\tail(u),e),\mathcal{I}(v))$.

Now, we have shown that $\tail(v)$ can be computed in $O(T)$ time for each visited vertex $v$,
and that all the equivalence tests can be done in $O(T)$ time for each edge picked at line~\ref{line:search:loop}.
Because each vertex is pushed into $X$ at most once and because each edge is processed at most twice (in both directions), we
obtain the following lemma.

\begin{lemma}\label{lem:search}
Given an $O(T)$-time incremental-test oracle, Algorithm~\ref{alg:search} runs in $O(m T)$ time.
\end{lemma}

\subsection{Constructing Half-Integral $\F$-Cover}\label{sec:hi:cover}

In this subsection,
we prove that if Algorithm~\ref{alg:search} fails to find a $y$-augmenting
path or pair, we can construct a half-integral $\F$-cover of the same size as follows.
Let $a$ and $b$ be the tables used in Algorithm~\ref{alg:search}.
First, we initialize $x(v)\gets 0$ for all $v\in V$.
For each integral path $I=(v_0,\ldots,v_\ell)$,
we set $x(v_{a(I)})\gets\frac{1}{2}$ and $x(v_{b(I)})\gets \frac{1}{2}$ if $a(I) \neq b(I)$,
and set $x(v_{a(I)})\gets 1$ if $a(I)=b(I)$.
We set $x(v_{a(S)})\gets\frac{1}{2}$ for each spoke $S=(v_0,\ldots,v_\ell)$.
From the construction, we have $|x|=|y|$.
We show that the function $x$ is an $\F$-cover.

\begin{lemma}\label{lem:cover}
If Algorithm~\ref{alg:search} fails to find a $y$-augmenting path or pair, the function $x$ constructed as
above is a half-integral $\F$-cover of size $|y|$.
\end{lemma}

We use the following lemma to prove this lemma.
\begin{lemma}\label{lem:cover:s-path}
For any implicational walk $Q$ with $x(V(Q))=0$, the vertex $t(Q)$ is visited and $T(P(t(Q)))\equiv Q$ holds.
\end{lemma}
\begin{proof}
We prove the lemma by induction on the length of $Q$.
For any vertex $s\in A$ with $x(s)=0$, $s$ is visited, and $T(P(s))=(s)$ holds.
Therefore, the statement holds when the length of $Q$ is zero.
Let $Q=Q'\circ uv$ be an implicational walk with $x(V(Q))=0$.
From the induction hypothesis, $u$ is visited, and $T(P(u))\equiv Q'$ holds.
Therefore, we have $T(P(u))\circ uv\equiv Q'\circ uv=Q$.
If $uv\not\in E(y)$ and $v$ is contained in a half-integral cycle but in no spokes, the algorithm returns a $y$-augmenting path at line~\ref{line:search:path1}.
We consider the following four cases: (Case 1) $v\not\in V(y)$;
(Case 2) $uv\not\in E(y)$ and $v$ is contained in an integral path;
(Case 3) $uv\not\in E(y)$ and $v$ is contained in a spoke; or
(Case 4) $uv\in E(y)$.

\paragraph{(Case 1)}
Consider the iteration when $u$ is picked at line~\ref{line:search:pick}.
If $v$ is already visited, $T(P(u))\circ uv\equiv T(P(v))$ holds because, otherwise, the algorithm returns a $y$-augmenting pair.
Therefore, we have $T(P(v))\equiv T(P(u))\circ uv\equiv Q$.
If $v$ is not visited yet and $v\in A$, we have $T(P(u))\circ uv\equiv (v)$.
Therefore, we have $T(P(v))=(v)\equiv T(P(u))\circ uv\equiv Q$.
If $v$ is not visited yet and $v\not\in A$, then $v$ becomes visited, and we have $T(P(v))=T(P(u))\circ uv\equiv Q$.

\paragraph{(Case 2)}
Let $I=(v_0,\ldots,v_\ell)$ be the integral path containing $v$ and let $i$ be the index such that $v_i=v$.
Consider the iteration when $u$ is picked at line~\ref{line:search:pick}.
Let $a(I)$ and $b(I)$ denote the values at the beginning of this iteration (hence, $x(v_{a(I)})$ and $x(v_{b(I)})$ might be zero).
If $i<a(I)$ holds, then $v$ is already visited, and $T(P(v))=(v_0,\ldots,v_i)$ holds.
If $T(P(u))\circ uv \not\equiv (v_0,\ldots,v_i)$ additionally holds, the algorithm returns a $y$-augmenting pair at
line~\ref{line:search:pair}.
Therefore, we have $T(P(v))=(v_0,\ldots,v_i)\equiv T(P(u))\circ uv\equiv Q$.
The same argument applies to the case of $i>b(I)$.

We now consider the remaining case that $a(I)\leq i\leq b(I)$ holds.
Because $(v_0,\ldots,v_i)\not\equiv (v_\ell,\ldots,v_i)$ holds, at least one of
$T(P(u))\circ uv\not\equiv (v_0,\ldots,v_i)$ or $T(P(u))\circ uv\not\equiv (v_\ell,\ldots,v_i)$ holds.
If both hold, $a(I)$ and $b(I)$ are both set to $i$ after this iteration. Therefore, we have $x(v)=1$,
which is a contradiction.
If $T(P(u))\circ uv\equiv (v_0,\ldots,v_i)$ holds, $a(I)$ is set to $i$ after this
iteration.
Because $x(v)=0$, $a(I)$ must be greater than $i$ at the end of the algorithm.
Therefore, we have $T(P(v))=(v_0,\ldots,v_i)\equiv T(P(u))\circ uv\equiv Q$.
The same argument applies to the case of $T(P(u))\circ uv\equiv (v_\ell,\ldots,v_i)$.

\paragraph{(Case 3)}
Let $S=(v_0,\ldots,v_\ell)$ be the spoke containing $v$ and $i$ be the index such that $v_i=v$.
Consider the iteration when $u$ is picked at line~\ref{line:search:pick}.
Let $a(S)$ denote the value at the beginning of this iteration.
If $i<a(S)$ holds, then $v$ is already visited, and $T(P(v))=(v_0,\ldots,v_i)$ holds.
If $T(P(u))\circ uv\not\equiv (v_0,\ldots,v_i)$ additionally holds, the algorithm returns a $y$-augmenting pair at
line~\ref{line:search:pair}.
Therefore, we have $T(P(v))=(v_0,\ldots,v_i)\equiv T(P(u))\circ uv\equiv Q$.

We now consider the remaining case that $a(I)\leq i$ holds.
If $T(P(u))\circ uv\not\equiv (v_0,\ldots,v_i)$ holds, the algorithm returns a $y$-augmenting path at
line~\ref{line:search:path2}.
Therefore, $T(P(u))\circ uv\equiv (v_0,\ldots,v_i)$ holds and $a(I)$ is set to $i$ after this iteration.
Because $x(v)=0$, $a(I)$ must be greater than $i$ at the end of the algorithm.
Therefore, we have $T(P(v))=(v_0,\ldots,v_i)\equiv T(P(u))\circ uv\equiv Q$.

\paragraph{(Case 4)}
Note that in this case, $uv$ must be contained in an integral path or a spoke because the vertices contained in half-integral cycles are never visited.
Let $(v_0,\ldots,v_\ell)$ be the integral path or the spoke containing $uv$ and let $i$ and $j$ be the indices such that
$v_i=v$ and $v_j=u$.
Because $u$ is visited, $T(P(u))$ is either $(v_0,\ldots,v_j)$ or $(v_\ell,\ldots,v_j)$, and w.l.o.g., we can
assume the former case.
If $i=j-1$, then $v$ is also visited, and we have $T(P(v))=(v_0,\ldots,v_{j-1})\equiv (v_0,\ldots,v_j)\circ uv=T(P(u))\circ uv\equiv Q$.
If $i=j+1$ and $v$ is not visited, then we have $x(v)\geq \frac{1}{2}$, which is a contradiction.
Therefore, $v$ is also visited, and we have $T(P(v))=(v_0,\ldots,v_{j+1})=(v_0,\ldots,v_j)\circ uv=T(P(u))\circ uv\equiv Q$.\qedhere

\end{proof}

\begin{proof}[Proof of Lemma~\ref{lem:cover}]
Suppose that there exists a conflicting walk $W\in\F$ with $x(V(W))<1$.
Because $x(v)\in\{0,\frac{1}{2},1\}$ holds, $x(V(W))$ is either $0$ or $\frac{1}{2}$.
If $x(V(W))=0$, by applying Lemma~\ref{lem:cover:s-path} against $W$, we have $W\equiv T(P(t(W)))=(t(W))$, which is a
contradiction.
If $x(V(W))=\frac{1}{2}$, let $v$ be the vertex with $x(v)=\frac{1}{2}$ on $W$.
By splitting $W$ at $v$, we obtain two implicational walks $Q_1$ and $Q_2$ such that $W=Q_1\circ Q_2^{-1}$ and $t(Q_1)=t(Q_2)=v$.
Because $x(v)=\frac{1}{2}$, $v$ is contained in an integral path or a spoke $I=(v_0,\ldots,v_\ell)$.
Let $i$ be the index such that $v_i=v$ and, w.l.o.g., we can assume that $a(I)=i$.
If $I$ is an integral path, $b(I)$ must be greater than $i$ because, otherwise, we have $x(v)=1$.

If the length of $Q_1$ is zero, $i$ is either $0$ or $\ell$.
In the latter case, $I$ is an integral path, and $\ell=i=a(I)<b(I)\leq\ell$, which is a contradiction.
Thus, we have $Q_1=(v_0)$.

If $Q_1$ has a positive length, let $Q_1=R\circ uv$.
By applying Lemma~\ref{lem:cover:s-path} against $R$, $u$ is visited, and $Q_1=R\circ uv\equiv T(P(u))\circ uv$ holds.
If $uv$ is contained in $I$, $u$ must be $v_{i-1}$ because $v_{i+1}$ is not visited.
Therefore, we have $Q_1\equiv T(P(u))\circ uv=(v_0,\ldots,v_{i-1})\circ uv=(v_0,\ldots,v_i)$.
If $uv$ is not contained in $I$, consider the iteration when $u$ is picked at line~\ref{line:search:pick}.
We then have $Q_1\equiv T(P(u))\circ uv\equiv (v_0,\ldots,v_i)$ because, otherwise, $b(I)$ is set to $i$, or the algorithm returns a $y$-augmenting path at line~\ref{line:search:path2}.

We now have proved that $Q_1\equiv (v_0,\ldots,v_i)$ holds.
By applying the same argument to $Q_2$, we have $Q_2\equiv (v_0,\ldots,v_i)$.
Therefore, we have $Q_1\equiv Q_2$, which is a contradiction.
\end{proof}

\subsection{Augmentation}\label{sec:hi:augmentation}
\subsubsection{Simplification of Alternating Path}
Before the proofs of Lemmas~\ref{lem:augment_path} and~\ref{lem:augment_pair},
we introduce a useful procedure to simplify an augmenting path/pair obtained by Algorithm~\ref{alg:search}.
We first define such an operation in a formal manner
and prove the validity just after the definition.

\begin{definition}
  For a basic $\F$-packing $y$ and a $y$-alternating path $P=P_1\circ\cdots\circ P_p$ with $p\geq 2$,
  the \emph{simplification $(y', P')$ of $(y, P)$} is defined as follows (Figure~\ref{fig:alternation}).
  \begin{itemize}
  \item A function $y'$ is constructed from $y$ as follows:
    \begin{itemize}
    \item if $P_2$ is contained in an integral path $I$, replace $I$ with an integral path $P_1\circ F(P_2)^{-1}$;
    \item if $P_2$ is contained in a spoke $S$, replace $S$ with a spoke $P_1\circ F(P_2)^{-1}$.
    \end{itemize}
  \item A walk $P'$ is defined as a concatenation $(B_y(P_2)\circ P_3)\circ P_4\circ P_5\circ\cdots\circ P_p$ of $\max\{p - 2, 1\}$ paths\footnote{When $p=2$, $P' = B_y(P_2)$.}.
  \end{itemize}
\end{definition}


\begin{figure}[t]
  \centering
  \includegraphics[scale=0.8]{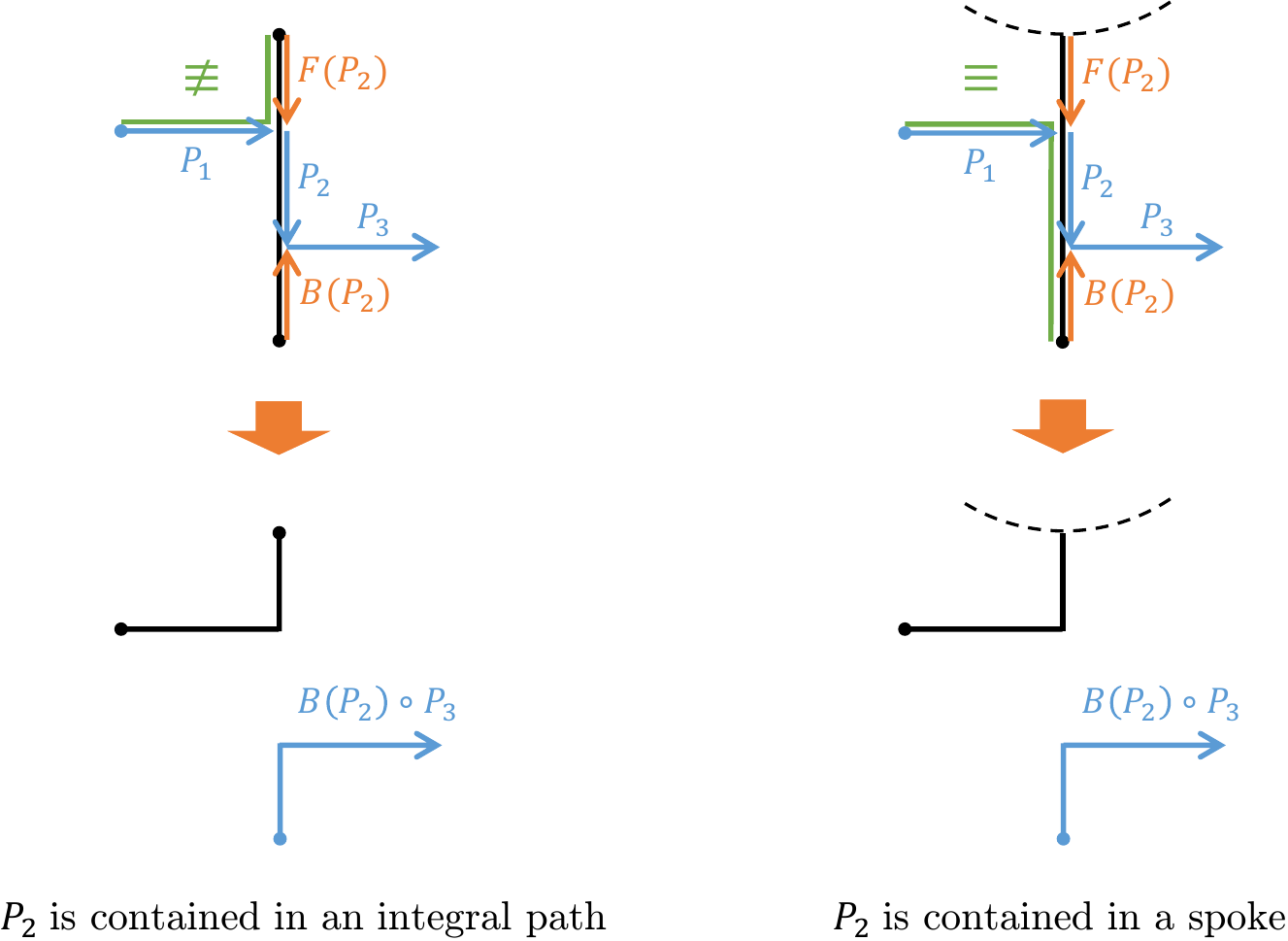}
  \caption{Simplification of a basic $\F$-packing and an alternating path.}
  \label{fig:alternation}
\end{figure}

\begin{lemma}\label{lem:alternating}
For any basic $\F$-packing $y$ and any $y$-alternating path $P=P_1\circ\cdots\circ P_p$ with $p\geq 2$,
the simplification $(y', P')$ of $(y, P)$ satisfies the following.
\begin{enumerate}
  \item $y'$ is a basic $\F$-packing with $|y'| = |y|$.
  \item For any even $i\geq 4$, the following holds.
  		\begin{enumerate}
  		  \item If $P_i$ is contained in an integral path in $y$, it is also contained in an integral path in $y'$. Moreover, $F_{y'}(P_i)\equiv F_y(P_i)$ and $B_{y'}(P_i)\equiv B_y(P_i)$ hold.
  		  \item If $P_i$ is contained in a spoke in $y$, it is also contained in a spoke in $y'$ in the direction toward $A$. Moreover, $F_{y'}(P_i)=F_y(P_i)$ and $B_{y'}(P_i)\equiv B_y(P_i)$ hold.
  		\end{enumerate}
  \item $P' = (B_y(P_2)\circ P_3)\circ P_4\circ P_5\circ\cdots\circ P_p$ is a $y'$-alternating path, where $B_y(P_2)\circ P_3$ is the first segment, and $P_i$ $(i\geq 4)$ is the $(i-2)$-th segment\,\footnote{When $p=2$, $P'$ consists of the single segment $B_y(P_2)$.}.
\end{enumerate}
\end{lemma}

\begin{proof}
First, we prove the first claim.
When $P_2$ is contained in an integral path $I$, from the condition~\ref{def:alternating:cycle} of the $y$-alternating paths (Definition~\ref{def:alternating}),
$P_1\not\equiv F(P_2)$ holds (i.e., $I':=P_1\circ F(P_2)^{-1}\in \F$).
Because $P_1$ is internally disjoint from $F(P_2)$ and $s(P_1)\not\in V(y)$, $I'$ is a simple path.
Thus, we can replace $I$ with the integral path $I'$.
When $P_2$ is contained in a spoke $S_i$ of a wheel, from the condition~\ref{def:alternating:spoke},
$P_1\equiv B(P_2)\circ P_2^{-1}$ holds.
From the definition of the wheel (Definition~\ref{def:wheel}), $S_{i-1}\circ H_{i-1}\circ S_i^{-1}$ and $S_i\circ H_i\circ S_{i+1}^{-1}$ are both in $\F$.
As $B(P_2)\circ P_2^{-1}\circ F(P_2)^{-1}=S_i$ holds, we have $S_i':=P_1\circ F(P_2)^{-1}\equiv S_i$.
When the degree of the wheel is at least three, $S_{i-1}\circ H_{i-1}\circ S_i'^{-1}$ and $S_i'\circ H_i\circ S_{i+1}^{-1}$ are both in $\F$.
When the degree of the wheel is one, $S_1'\circ H_1\circ S_1'^{-1}$ is in $\F$.
Thus, this replacement of the spoke preserves the condition for the wheel.

Next, we prove the second claim.
We can observe the following for any even $i\geq 4$.
\begin{itemize}
  \item If both of $P_2$ and $P_i$ are contained in the same integral path in $y$, from the condition~\ref{def:alternating:cycle} in Definition~\ref{def:alternating},
  		$P_i$ is contained in $F_y(P_2)$, and $P_1\equiv B_y(P_2)\circ P_2^{-1}$ holds.
  		Thus, $P_i$ is contained in the integral path $P_1\circ F_y(P_2)^{-1}$ in $y'$.
  		If they have the same direction, $P_2$ is contained in $B_y(P_i)$, and we can write $B_{y}(P_i)=(B_y(P_2)\circ P_2^{-1})\circ W$ for some subpath $W$.
  		Then, it holds that $F_{y'}(P_i)=F_y(P_i)$ and $B_{y'}(P_i)=P_1\circ W\equiv (B_y(P_2)\circ P_2^{-1})\circ W=B_y(P_i)$.
  		If they have the opposite direction, $P_2$ is contained in $F_y(P_i)$, and we can write $F_y(P_i)=(B_y(P_2)\circ P_2^{-1})\circ W$ for some subpath $W$.
  		Then, it holds that $F_{y'}(P_i)=P_1\circ W\equiv (B_y(P_2)\circ P_2^{-1})\circ W=F_y(P_i)$ and $B_{y'}(P_i)=B_y(P_i)$.
  \item If both $P_2$ and $P_i$ are contained in the same spoke in $y$, from the condition~\ref{def:alternating:spoke} in Definition~\ref{def:alternating},
  		$P_i$ is contained in $F_y(P_2)$, and $P_1\equiv B_y(P_2)\circ P_2^{-1}$ holds.
  		Thus, $P_i$ is contained in the spoke $P_1\circ F_y(P_2)^{-1}$ in $y'$.
  		Because both of them are directed toward $A$, we can write $B_y(P_i)=(B_y(P_2)\circ P_2^{-1})\circ W$ for some subpath $W$.
  		Then, it holds that $F_{y'}(P_i)=F_y(P_i)$ and $B_{y'}(P_i)=P_1\circ W\equiv(B_y(P_2)\circ P_2^{-1})\circ W=B_y(P_i)$.
  \item Otherwise, the integral path or the spoke containing $P_i$ does not change. Thus, we have $F_{y'}(P_i)=F_y(P_i)$ and $B_{y'}(P_i)=B_y(P_i)$.
\end{itemize}

Finally, we prove the third claim.
Because $B_y(P_2)$ does not contain any $P_i$, the conditions~\ref{def:alternating:distinct}--\ref{def:alternating:odd_even} in Definition~\ref{def:alternating} are satisfied.
From the second claim and the property that $A\circ W\equiv B\circ W$ holds for any implicational walks $A\circ W$ and $B$ satisfying $A\equiv B$, none of the three equivalence
relations appeared in the condition~\ref{def:alternating:cycle_spoke} change.
As we have seen in the proof of the second claim, $E(B_{y'}(P_i))\setminus E(B_y(P_i))\subseteq E(P_1)$ and
$E(F_{y'}(P_i))\setminus E(F_y(P_i))\subseteq E(P_1)$ hold.
Therefore, for any $i\geq 4$, none of the $P_j$'s with $j>i$ are newly contained in $B_{y'}(P_i)$ or $F_{y'}(P_i)$.
Thus, the condition~\ref{def:alternating:cycle_spoke} is satisfied.
\end{proof}

We obtain the following corollaries from Lemma~\ref{lem:alternating} by repeatedly applying the simplifying operation.
\begin{corollary}\label{cor:alternating_sub}
Given a basic $\F$-packing $y$, a $y$-alternating path $P=P_1\circ\cdots\circ P_p$, and an even integer $p'\leq p$,
a basic $\F$-packing $y'$ of the same size and a $y'$-alternating path $P'$ satisfying the following conditions can be constructed in linear time.
\begin{enumerate}
  \item $P'$ can be written as $P'=(B'\circ P_{p'+1})\circ P_{p'+2}\circ\cdots\circ P_p$ for some implicational path $B'\equiv B_y(P_{p'})$.
  \item For any even $i\geq p'+2$, the following holds.
  \begin{enumerate}
    \item If $P_i$ is contained in an integral path in $y$, it is also contained in an integral path in $y'$. Moreover, $F_{y'}(P_i)\equiv F_y(P_i)$ and $B_{y'}(P_i)\equiv B_y(P_i)$ hold.
    \item If $P_i$ is contained in a spoke in $y$, it is also contained in a spoke in $y'$ in the direction toward $A$. Moreover, $F_{y'}(P_i)=F_y(P_i)$ and $B_{y'}(P_i)\equiv B_y(P_i)$ hold.
  \end{enumerate}
\end{enumerate}
\end{corollary}

\begin{corollary}\label{cor:alternating_full}
Given a basic $\F$-packing $y$ and a $y$-alternating path $P=P_1\circ\cdots\circ P_p$,
a basic $\F$-packing $y'$ of the same size and a single-segment $y'$-alternating path $P'$ satisfying $P'\equiv T_y(P)$ can be constructed in linear time.
\end{corollary}

\subsubsection{Augmentation by Augmenting Path}
The subsection aims to prove Lemma \ref{lem:augment_path}.

\begin{figure}[t]
  \centering
  \includegraphics[scale=0.8]{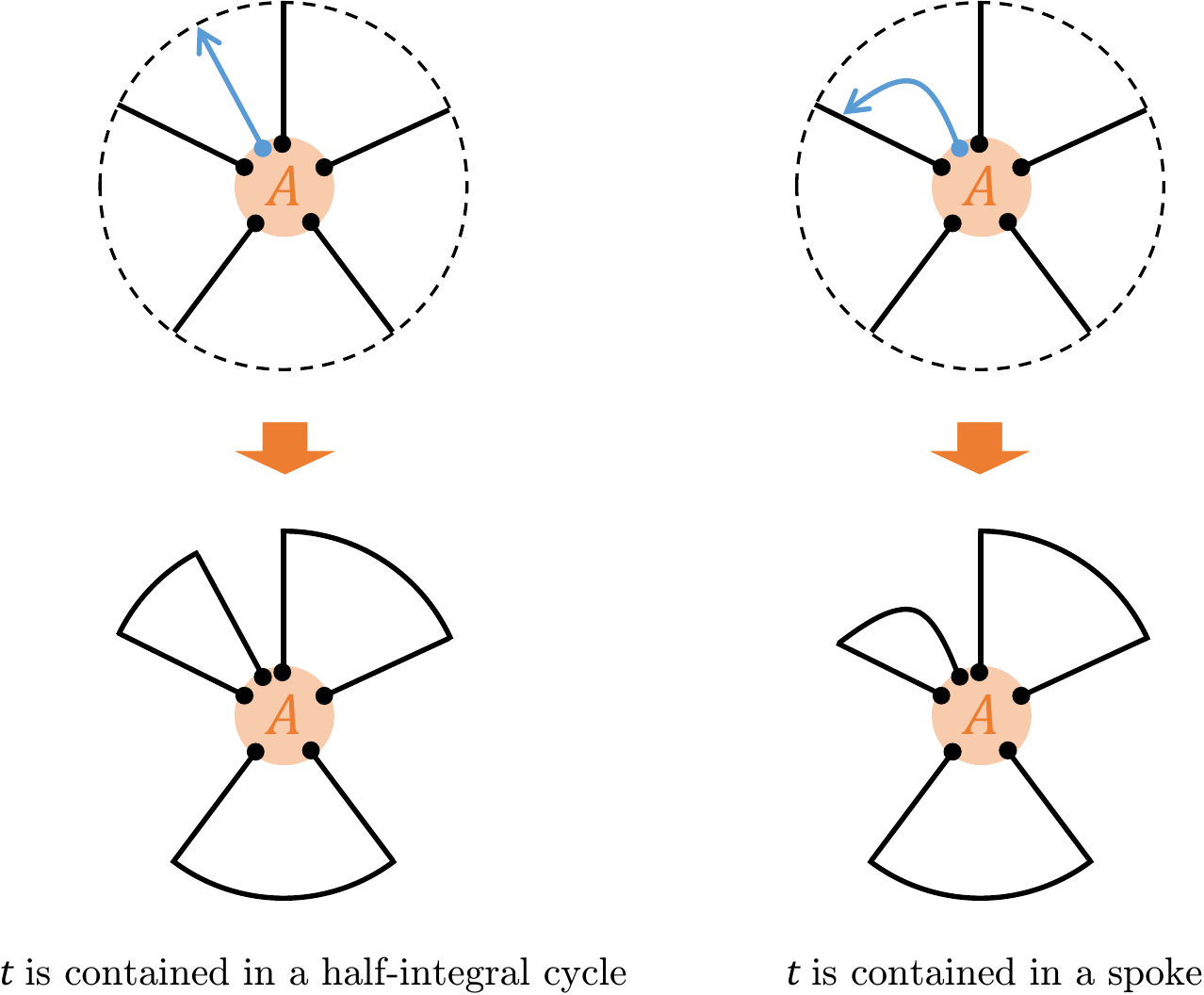}
  \caption{Augmentation by a single-segment augmenting path.}
  \label{fig:augment_path}
\end{figure}

\begin{proof}[Proof of Lemma~\ref{lem:augment_path}]
Let $P=P_1\circ\cdots\circ P_p$ be the given $y$-augmenting path and $t=t(P)$.
If $p>1$, by Corollary~\ref{cor:alternating_full},
we obtain in linear time a basic $\F$-packing $y'$ of the same size and a single-segment $y'$-alternating path $P'$ such that $t(P')=t$ and $P'\equiv T_y(P)$.
We now show that $P'$ is a $y'$-augmenting path.
If $P$ satisfies the condition~\ref{def:augment_path:0} or~\ref{def:augment_path:1} of the $y$-augmenting paths (Definition~\ref{def:augment_path}),
$P'$ is a $y'$-augmenting path because $t(P')=t$ is still contained in $A\setminus V(y')$ or $V(y')\setminus V_1(y')$.
If $P$ satisfies the condition~\ref{def:augment_path:2}, the spoke $S$ that contains $t$ in $y$ might not exist in $y'$.
However, from the condition~\ref{def:augment_path:22}, $t$ is still contained in a spoke $S'$ in $y'$, which may not be identical to $S$.
By the same argument as in the proof of the second statement of Lemma~\ref{lem:alternating}, we have $B_{y'}(t)\equiv B_y(t)$.
Therefore, we have $T_{y'}(P')=P'\equiv T_y(P)\not\equiv B_y(t)\equiv B_{y'}(t)$.
Thus, $P'$ is a $y'$-augmenting path.

We can now concentrate on the case when $p=1$ (Figure~\ref{fig:augment_path}).
If the condition~\ref{def:augment_path:0} is satisfied, $P$ is a conflicting path containing no vertices in $V(y)$.
Therefore, we can obtain a basic $\F$-packing of size $|y|+1$ by inserting the integral path $P$.
Otherwise, let $(\{S_1,\ldots,S_d\}, H_1\circ\cdots\circ H_d)$ be the wheel containing $t$.
W.l.o.g., we can assume that $t$ is contained in $H_d$ or $S_d$.

If $t$ is contained in $H_d$, let $F$ be the prefix subpath of $H_d$ to $t$ and $B$ be the suffix subpath of $H_d$ from $t$ (i.e., $F\circ B=H_d$).
We have $S_d\circ F\not\equiv S_1\circ B^{-1}$ because $S_d\circ H_d\circ S_1^{-1}\in\F$ Thus, at least one of $S_d\circ F\circ P^{-1}$ and $S_1\circ B^{-1}\circ P^{-1}$ is in $\F$, and w.l.o.g., we can assume the former case.
Then, we can obtain a basic $\F$-packing of size $|y|+\frac{1}{2}$ by decomposing the wheel into $(d-1)/2$ integral paths $\{S_1\circ H_1\circ S_2^{-1},S_3\circ H_3\circ S_4^{-1},\ldots,S_{d-2}\circ H_{d-2}\circ S_{d-1}^{-1}\}$
and by inserting an integral path $S_d\circ F\circ P^{-1}$.

If $t$ is contained in $S_d$, we have $P\not\equiv B(t)$. Then, we can obtain a basic $\F$-packing of size $|y|+\frac{1}{2}$ by decomposing the wheel into $(d-1)/2$ integral paths
$\{S_1\circ H_1\circ S_2^{-1},S_3\circ H_3\circ S_4^{-1},\ldots,S_{d-2}\circ H_{d-2}\circ S_{d-1}^{-1}\}$
and by inserting an integral path $P\circ B(t)^{-1}$.
\end{proof}

\subsubsection{Augmentation by Augmenting Pair}
This subsection aims to prove Lemma~\ref{lem:augment_pair}.

Intuitively, we want to augment $y$ as follows.
We first simplify the common prefix $R$ of $(P, Q)$ by applying Corollary~\ref{cor:alternating_sub}, and then obtain a new wheel whose half-integral cycle is $P'\circ Q'^{-1}$.
However, it is not that easy.
This approach does not work when $P'$ or $Q'$ intersects with spokes or intersects with the same integral path multiple times.
Therefore, we augment $y$ by gradually simplifying the augmenting pair.
First, we prove the lemma against a special case.

\begin{lemma}\label{lem:augment_pair:special}
Given a basic $\F$-packing $y$ and a $y$-augmenting pair $(P=P_1\circ\cdots\circ P_p, Q=Q_1\circ\cdots\circ Q_q)$ satisfying all the following conditions,
a basic $\F$-packing of size $|y|+\frac{1}{2}$ can be constructed in linear time.
\begin{enumerate}
  \item None of the $P_i$'s are contained in the spokes.\label{def:augment_pair:special:no_spokes}
  \item None of the integral paths contain multiple $P_i$'s.\label{def:augment_pair:special:unique_cycles}
  \item One of the following two conditions is satisfied:\label{def:augment_pair:special:short_prefix}
  \begin{enumerate}
    \item $p\geq q=1$ or\label{def:augment_pair:special:a}
    \item $p\geq q=2$, $P_1=Q_1$, and $P_2$ and $Q_2$ are contained in an integral path in the opposite direction.\label{def:augment_pair:special:b}
  \end{enumerate}
\end{enumerate}
\end{lemma}

\begin{proof}
When $(P, Q)$ satisfies the condition~\ref{def:augment_pair:special:a}, we further divide the case into the following two cases:
$p=1$ or $p\geq 2$.

If $p=1$, $(P, Q)$ is a single-branching pair such that $P\not\equiv Q$.
Therefore, we can obtain a basic $\F$-packing of size $|y|+\frac{1}{2}$ by introducing a new wheel $P\circ Q^{-1}$ of degree one.

\begin{figure}[t]
  \centering
  \includegraphics[scale=0.8]{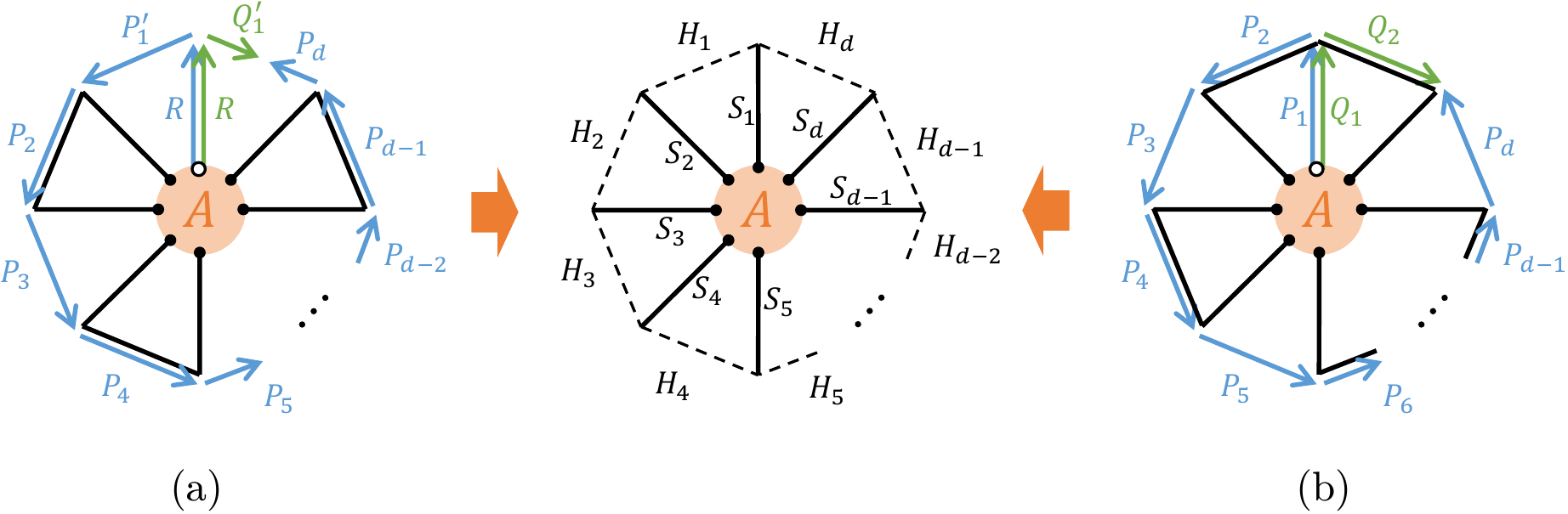}
  \caption{Applying Lemma~\ref{lem:augment_pair:special}.}
  \label{fig:pair_special}
\end{figure}

In $p\geq 2$, let $R$ be the common prefix and let us write $P_1=R\circ P_1'$ and $Q_1=R\circ Q_1'$.
We set $d:=p+1$ and $P_d:=(t(P))$ if $p$ is even;
otherwise, we set $d:=p$.
We define paths $\{H_1,\ldots,H_d\}$ and $\{S_1,\ldots,S_d\}$ as follows (Figure~\ref{fig:pair_special} (a)).
\begin{itemize}
  \item $H_1:=P_1'$.
  \item $H_i:=P_i$ for $i\in\{2,\ldots,d-1\}$.
  \item $H_d:=P_d\circ Q_1'^{-1}$.
  \item $S_1:=R$.
  \item $S_i:=F(P_i)$ for even $i\in\{2,4,\ldots,d-1\}$.
  \item $S_i:=B(P_{i-1})$ for odd $i\in\{3,5,\ldots,d\}$.
\end{itemize}
We now show that these paths form a wheel.
The first condition of the wheel (Definition~\ref{def:wheel}) is trivially satisfied.
As none of the integral paths in $y$ contain multiple $P_i$'s, the second and third conditions are satisfied.
We can see that the fourth condition is satisfied as follows.
\begin{itemize}
  \item $S_1\circ H_1\circ S_2^{-1}=R\circ P_1'\circ F(P_2)^{-1}=P_1\circ F(P_2)^{-1}\in\F$.
  \item $S_i\circ H_i\circ S_{i+1}^{-1}=F(P_i)\circ P_i\circ B(P_i)^{-1}\in\F$ for even $i\in\{2,4,\ldots,d-1\}$.
  \item $S_i\circ H_i\circ S_{i+1}^{-1}=B(P_{i-1})\circ P_i\circ F(P_{i+1})^{-1}\in\F$ for odd $i\in\{3,5,\ldots,d-2\}$.
  \item $S_d\circ H_d\circ S_1^{-1}=B(P_{d-1})\circ P_d\circ Q_1'^{-1}\circ R^{-1}=T(P)\circ T(Q)^{-1}\in\F$.
\end{itemize}
Thus, we can obtain a basic $\F$-packing of size $|y|+\frac{1}{2}$ by removing the $(d-1)/2$ integral paths intersecting $P$ and by inserting the wheel of degree $d$.

Finally, we consider the case when $(P, Q)$ satisfies the condition~\ref{def:augment_pair:special:b}.
Note that $p$ must be odd from the condition~\ref{def:augment_pair:3} of $y$-augmenting pairs (Definition~\ref{def:augment_pair}).
Let $d:=p$.
We define paths $\{H_1,\ldots,H_d\}$ and $\{S_1,\ldots,S_d\}$ as follows (Figure~\ref{fig:pair_special} (b)).
\begin{itemize}
  \item $H_i:=P_{i+1}$ for $i\in \{1,\ldots,d-1\}$.
  \item $H_d:=Q_2^{-1}$.
  \item $S_1:=P_1$.
  \item $S_i:=B(P_i)$ for even $i\in\{2,4,\ldots,d-1\}$.
  \item $S_i:=F(P_{i+1})$ for odd $i\in\{3,5,\ldots,d-2\}$.
  \item $S_d:=B(Q_2)$.
\end{itemize}
We now show that these paths form a wheel.
The first condition of the wheel is trivially satisfied.
The second and third conditions are satisfied because none of the integral paths in $y$ contain multiple $P_i$'s.
We can see that the fourth condition is satisfied as follows.
\begin{itemize}
  \item $S_1\circ H_1\circ S_2^{-1}=P_1\circ P_2\circ B(P_2)^{-1}=P_1\circ (B(P_2)\circ P_2^{-1})^{-1}=Q_1\circ F(Q_2)^{-1}\in\F$.
  \item $S_i\circ H_i\circ S_{i+1}^{-1}=B(P_i)\circ P_{i+1}\circ F(P_{i+2})^{-1}\in\F$ for even $i\in\{2,4,\ldots,d-3\}$.
  \item $S_i\circ H_i\circ S_{i+1}^{-1}=F(P_{i+1})\circ P_{i+1}\circ B(P_{i+1})^{-1}\in\F$ for odd $i\in\{3,5,\ldots,d-2\}$.
  \item $S_{d-1}\circ H_{d-1}\circ S_d^{-1}=B(P_{d-1})\circ P_d\circ B(Q_2)^{-1}=T(P)\circ T(Q)^{-1}\in\F$.
  \item $S_d\circ H_d\circ S_1^{-1}=B(Q_2)\circ Q_2^{-1}\circ P_1^{-1}=F(P_2)\circ P_1^{-1}\in \F$.
\end{itemize}
Thus, we can obtain a basic $\F$-packing of size $|y|+\frac{1}{2}$ by removing the $(d-1)/2$ integral paths intersecting $P$ and by inserting the wheel of degree $d$.
\end{proof}


Next, we provide two lemmas for weakening the assumptions in the abovementioned lemma.
\begin{lemma}\label{lem:augment_pair:detour}
Given a basic $\F$-packing $y$ and a $y$-augmenting pair $(P=P_1\circ\cdots\circ P_p, Q=Q_1\circ\cdots\circ Q_q)$ satisfying all the following conditions,
a basic $\F$-packing of size $|y|+\frac{1}{2}$ can be constructed in linear time.
\begin{enumerate}
  \item None of the $P_i$'s are contained in the spokes.
\renewcommand{\theenumi}{\arabic{enumi}'}
  \item Any two segments of $P$ contained in the same integral path have the same direction.\label{def:augment_pair:detour:direction}
\renewcommand{\theenumi}{\arabic{enumi}}
  \item One of the following two conditions is satisfied:
  \begin{enumerate}
    \item $p\geq q=1$ or
    \item $p\geq q=2$, $P_1=Q_1$, and $P_2$ and $Q_2$ are contained in an integral path in the opposite direction.
  \end{enumerate}
\end{enumerate}
\end{lemma}
\begin{proof}
Note that the conditions~\ref{def:augment_pair:special:no_spokes} and~\ref{def:augment_pair:special:short_prefix} are
the same as those in Lemma~\ref{lem:augment_pair:special}. Besides,
if the condition~\ref{def:augment_pair:special:b} is satisfied, from the condition~\ref{def:alternating:cycle} of
$y$-alternating paths (Definition~\ref{def:alternating}), the integral path containing $P_2$ and $Q_2$ never
contain any other segments.
We call a segment of $P$ \emph{obstructive} if it is contained in an integral path containing multiple segments of $P$.
If there exist no obstructive segments, the
condition~\ref{def:augment_pair:special:unique_cycles} of Lemma~\ref{lem:augment_pair:special} is satisfied.
Thus, we can obtain a basic $\F$-packing of size $|y|+\frac{1}{2}$ by applying
Lemma~\ref{lem:augment_pair:special}.
We repeat the following process while obstructive segments exist.

Let $P_a$ be the first obstructive segment and 
$P_b$ ($b>a$) be the next segment contained in the integral path containing $P_a$.
From the condition~\ref{def:alternating:cycle} of the $y$-alternating paths (Definition~\ref{def:alternating}),
$P_b$ is contained in $F(P_a)$, and $B(P_{a-2})\circ P_{a-1}\equiv B(P_a)\circ P_a^{-1}$ holds.
We construct a basic $\F$-packing $\bar{y}$ of the same size and a
$\bar{y}$-augmenting pair $(\bar{P},Q)$ such that the precondition of this lemma is satisfied and the number
of obstructive segments strictly decreases as follows (see Figure~\ref{fig:pair_detour}).

\begin{figure}[t]
  \centering
  \includegraphics[scale=0.8]{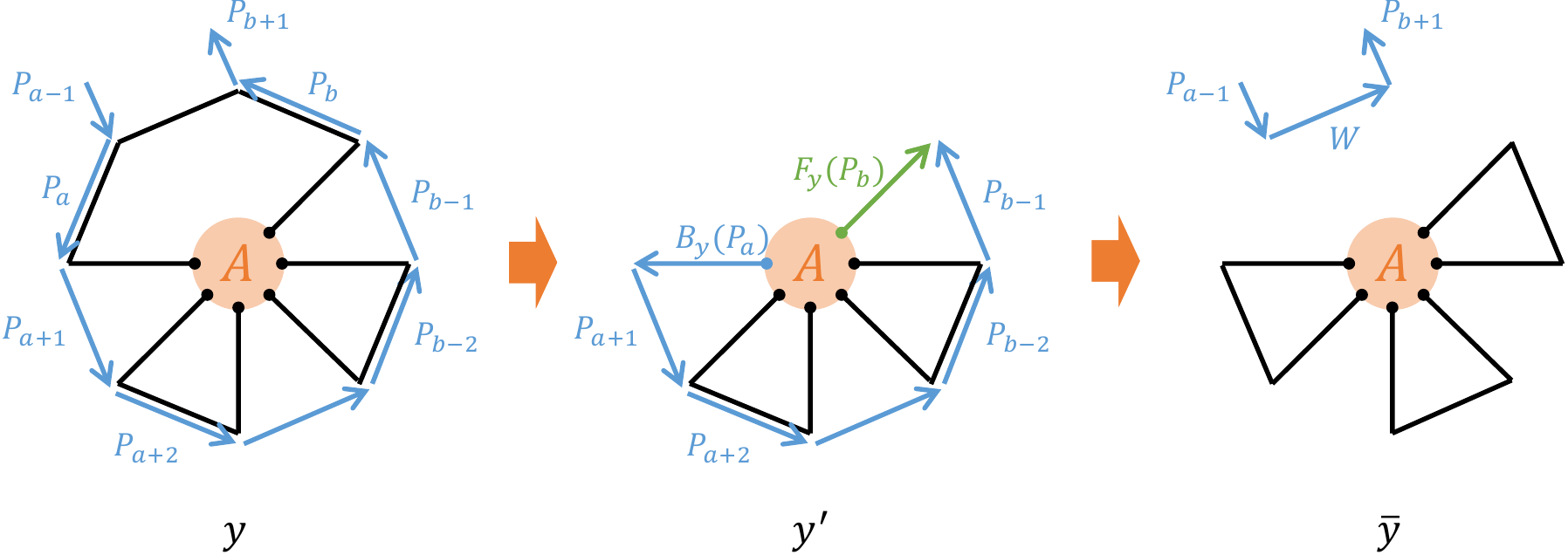}
  \caption{Applying Lemma~\ref{lem:augment_pair:detour}.}
  \label{fig:pair_detour}
\end{figure}

First, we construct a basic $\F$-packing $y'$ of size $|y|-1$ by removing the integral path containing $P_a$ and $P_b$
from $y$.
Observe that $P':=(B_y(P_a)\circ P_{a+1})\circ P_{a+2}\circ\cdots\circ P_{b-1}$ is a $y'$-alternating path.
Then, by applying Corollary~\ref{cor:alternating_full} against $y'$ and $P'$, we obtain a basic $\F$-packing $y''$ of
size $|y|-1$ and a single-segment $y''$-alternating path $P''$ satisfying $P''\equiv T_{y'}(P')=B_{y'}(P_{b-2})\circ
P_{b-1}=B_y(P_{b-2})\circ P_{b-1}$.
From the condition~\ref{def:alternating:cycle} of the $y$-alternating path $P$,
$B_y(P_{b-2})\circ P_{b-1}\not\equiv F_y(P_b)$ holds. From the choice of $P_b$, $F_y(P_b)$ is internally disjoint
from $P''$.
Thus, we can obtain a basic $\F$-packing $\bar{y}$ of size $|y|$ by introducing a new integral path
$P''\circ F_y(P_b)$.
Let $W$ be the path from $s(P_a)$ to $t(P_b)$ along the integral path and let
$\bar{P}:=P_1\circ\cdots\circ P_{a-2}\circ (P_{a-1}\circ W\circ P_{b+1})\circ P_{b+2}\circ\cdots\circ P_p$.

Finally, we prove that $(\bar{P}, Q)$ is a $\bar{y}$-augmenting pair satisfying the preconditions of this lemma.
From the construction of $\bar{y}$ and $\bar{P}$, the
conditions~\ref{def:alternating:distinct}--\ref{def:alternating:odd_even} of the $\bar{y}$-alternating paths
(Definition~\ref{def:alternating}) are clearly satisfied.
The condition~\ref{def:alternating:cycle} is satisfied because $B_{\bar{y}}(P_{a-2})\circ (P_{a-1}\circ W\circ P_{b+1})\equiv B_y(P_{a-2})\circ P_{a-1}\circ W\circ P_{b+1}\equiv B_y(P_a)\circ P_a^{-1}\circ W\circ P_{b+1}\equiv B_y(P_b)\circ P_{b+1}$ holds.
Thus, $\bar{P}$ is a $\bar{y}$-alternating path.
Because $T_{\bar{y}}(\bar{P})\equiv T_y(P)\neq T_y(Q)\equiv T_{\bar{y}}(Q)$ holds and $(\bar{P}, Q)$
satisfies the precondition~\ref{def:augment_pair:special:short_prefix} of the lemma, $(\bar{P}, Q)$ is a $\bar{y}$-augmenting pair.
Because no segments of $\bar{P}$ are newly contained in the spokes in $\bar{y}$, the condition~\ref{def:augment_pair:special:no_spokes} of the lemma is satisfied.
From the choice of $P_a$, for any even $i\in\{a+2,\ldots,b-2\}$, $B_{y'}(P_i)$ contains no segments from $\{P_2,\ldots,P_{a-2}\}$.
Therefore, no two segments of $\bar{P}$ are newly contained in a same integral path in $\bar{y}$.
Thus, the condition~\ref{def:augment_pair:detour:direction} is satisfied.
When $(P,Q)$ satisfies the condition~\ref{def:augment_pair:special:a}, $(\bar{P},Q)$ also satisfies the condition~\ref{def:augment_pair:special:a}.
When $(P,Q)$ satisfies the condition~\ref{def:augment_pair:special:b}, $(\bar{P},Q)$ also satisfies the condition~\ref{def:augment_pair:special:b}
because the integral path containing $P_2$ and $Q_2$ remains in $\bar{y}$.
Thus, all the conditions in the lemma are satisfied.

We can find the pair $(P_a, P_b)$ by gradually increasing an index $i$, which is not reset during the repetition, and
by searching for $P_j$ contained in $F(P_i)$ by traversing the integral path.
Therefore, each edge is traversed at most once through the whole process. Thus, the total running time is linear in
the graph size.
\end{proof}

\begin{lemma}\label{lem:augment_pair:shortcut}
Given a basic $\F$-packing $y$ and a $y$-augmenting pair $(P=P_1\circ\cdots\circ P_p,Q=Q_1\circ\cdots\circ Q_q)$,
either of a $y$-augmenting path or a $y$-augmenting pair $(\bar{P}, \bar{Q})$ satisfying all the following conditions
can be constructed in linear time.
\begin{enumerate}
  \item All the segments of $\bar{Q}$, except for the last one, are contained in
  $\bar{P}$.\label{lem:augment_pair:shortcut:1}
  \item $\bar{P}$ can be written as $\bar{P}=P\circ Q_q^{-1}\circ Q_{q-1}^{-1}\circ\cdots\circ Q_{q'}^{-1}$ for some
  $q'$.\label{lem:augment_pair:shortcut:2}
  \item The common prefix of $(\bar{P}, \bar{Q})$ contains the common prefix of $(P,
  Q)$.\label{lem:augment_pair:shortcut:3}
  \item The following two conditions are satisfied for the new segments $S:=\{Q_q^{-1},\ldots,Q_{q'}^{-1}\}$ of
  $\bar{P}$:\label{lem:augment_pair:shortcut:4}
  \begin{enumerate}
    \item no segments in $S$ are contained in the spokes, and
    \item any two segments in $S$ contained in the same integral path have the same
    direction.
  \end{enumerate}
\end{enumerate}
\end{lemma}

\begin{proof}
Initially, the conditions~\ref{lem:augment_pair:shortcut:2}--\ref{lem:augment_pair:shortcut:4} are trivially
satisfied.
We repeat the following process, which preserves these conditions.
The condition~\ref{lem:augment_pair:shortcut:1} is satisfied
when $q$ becomes one or $Q_{q-1}$ gets contained in $P$.

\paragraph{(Case 1)}
If both of $p$ and $q$ are odd,
we update $P'\gets P_1\circ\cdots\circ P_{p-1}\circ (P_p\circ Q_q^{-1})$ and $Q'\gets Q_1\circ\cdots\circ Q_{q-1}$.
Because $P_p$ and $Q_q$ share no edges and because
$T(P')\circ T(Q')^{-1}=(B(P_{p-1})\circ (P_p\circ Q_q^{-1}))\circ B(Q_{q-1})^{-1}=(B(P_{p-1})\circ P_p)\circ(B(Q_{q-1})\circ Q_q)^{-1}=T(P)\circ T(Q)^{-1}$ holds,
$(P',Q')$ is a $y$-augmenting pair.

\paragraph{(Case 2)}
If $p$ is even and $q$ is odd,
we update $P'\gets P_1\circ\cdots\circ P_{p-1}\circ P_p\circ Q_q^{-1}$ and $Q'\gets Q_1\circ\cdots\circ Q_{q-1}$.
Because $P_p$ and $Q_q$ share no edges and because
$T(P')\circ T(Q')^{-1}=(B(P_p)\circ Q_q^{-1})\circ B(Q_{q-1})^{-1}=B(P_p)\circ (B(Q_{q-1})\circ Q_q)^{-1}=T(P)\circ T(Q)^{-1}$ holds,
$(P',Q')$ is a $y$-augmenting pair.

\begin{figure}[t]
  \centering
  \includegraphics[scale=0.7]{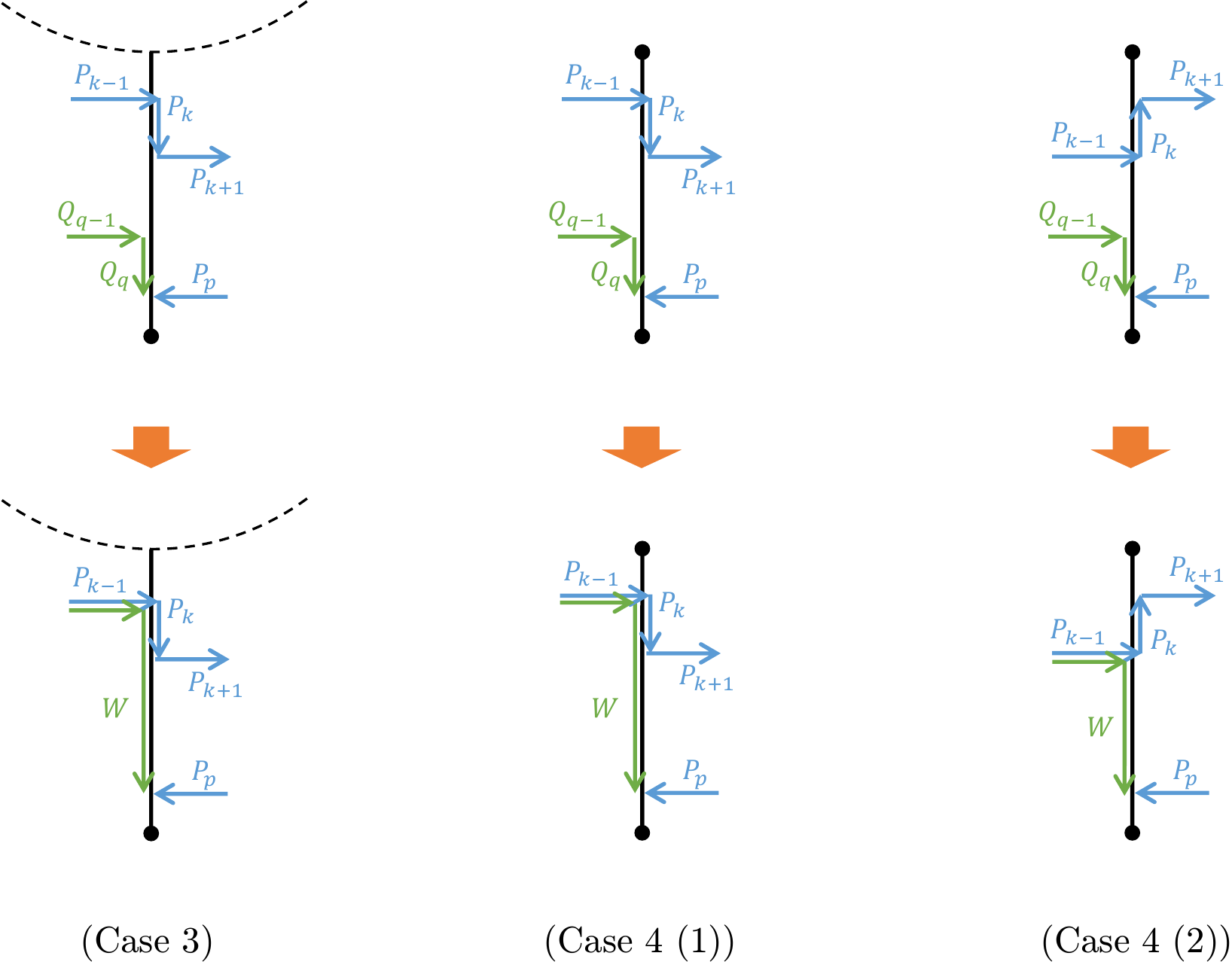}
  \caption{Applying Lemma~\ref{lem:augment_pair:shortcut}.}
  \label{fig:pair_shortcut}
\end{figure}

\paragraph{(Case 3)}
If $p$ is odd, $q$ is even, and $Q_q$ is contained in a spoke, then we search for the nearest segment of $P$ contained in
$F(Q_q)$ by traversing the spoke.
Note that none of the other segments of $Q$ are contained in $F(Q_q)$ from the condition~\ref{def:alternating:spoke} of
the $y$-alternating paths (Definition~\ref{def:alternating}). Moreover, each segment of $P$ is fully contained in $F(Q_q)$ or
internally disjoint from $F(Q_q)$ because $Q_q$ shares no edges with $P$. Thus, each edge is traversed at most once
through the entire process, and the total running time of this part is linear in the graph size.
If none of the $P_i$'s are contained in $F(Q_q)$,
then $P$ is a $y$-augmenting path because $t(P)$ is contained in the spoke, and $T(P)\not\equiv T(Q)=B(t(P))$ holds.
Otherwise, let $P_k$ be the nearest segment of $P$ contained in $F(Q_q)$; let $W$ be the path from $s(P_k)$ to $t(P)$
along the spoke; and let $\bar{Q}:=P_1\circ\cdots\circ P_{k-1}\circ W$ (Figure~\ref{fig:pair_shortcut} (left)).
Because $T(\bar{Q})=B(W)=B(Q_q)=T(Q)$ holds, $(P,\bar{Q})$ is a $y$-augmenting pair.
Note that this finishes the process.

\paragraph{(Case 4)}
If $p$ is odd, $q$ is even, and $Q_q$ is contained in an integral path, then we search for the
nearest segment of $P$ or $Q$ contained in $F(Q_q)$ by traversing the integral path.
Because each edge is traversed at most twice (in two directions) through the entire process, the total running time of this part is linear in the graph size.
Note that from the condition~\ref{def:alternating:cycle} of the $y$-alternating path $Q$
(Definition~\ref{def:alternating}), $F(Q_q)$ cannot contain a segment $Q_k$ such that $Q_k$ has the same direction as
$Q_q$ or has the opposite direction as $Q_q$ and $B(Q_{k-2})\circ Q_{k-1}\not\equiv B(Q_k)\circ Q_k^{-1}$ holds.

If the nearest segment is $P_k$ such that $P_k$ has the same direction as $Q_q$ or
has the opposite direction as $Q_q$ and $B(P_{k-2})\circ P_{k-1}\not\equiv B(P_k)\circ P_k^{-1}$ holds,
let $W$ be the path from $s(P_k)$ to $t(P)$ along the integral path and $\bar{Q}:=P_1\circ\cdots\circ P_{k-1}\circ W$ (see Figure~\ref{fig:pair_shortcut} (right)).
Because $T(\bar{Q})=B(W)=B(Q_q)=T(Q)$ holds, $(P,\bar{Q})$ is a $y$-augmenting pair, and we then finish the process.

Otherwise (i.e., (a) no segments are contained in $F(Q_q)$, (b) the nearest segment is $Q_k$ that has the opposite
direction as $Q_q$ and $B(Q_{k-2})\circ Q_{k-1}\equiv B(Q_k)\circ Q_k^{-1}$ holds, or
(c) the nearest segment is $P_k$ that has the opposite direction as $Q_q$ and
$B(P_{k-2})\circ P_{k-1}\equiv B(P_k)\circ P_k^{-1}$ holds),
we update $P'\gets P_1\circ\cdots\circ P_p\circ Q_q^{-1}$ and $Q'\gets Q_1\circ\cdots\circ Q_{q-1}$.
Note that in this case, from the condition~\ref{def:alternating:cycle} of the $y$-alternating paths
(Definition~\ref{def:alternating}) and the condition~\ref{def:augment_pair:4} of $y$-augmenting pairs
(Definition~\ref{def:augment_pair}), conditions (b) and (c) hold not only against the nearest segment $Q_k$ or $P_k$,
but also against any segments contained in $F(Q_q)$.

First, we prove that $P'$ is a $y$-alternating path.
The first four conditions of the $y$-alternating paths (Definition~\ref{def:alternating}) are clearly satisfied.
The condition~\ref{def:alternating:cycle_spoke} is satisfied for $i=p+1$
because $B(P_{p-1})\circ P_p=T(P)\not\equiv T(Q)=B(Q_q)=F(Q_q^{-1})$ holds.
For checking the condition~\ref{def:alternating:cycle_spoke} against the other $P_i$'s, it suffices to show that
for any segment $P_i$ contained in the same integral path as $Q_q$,
it holds that $B(P_{i-2})\circ P_{i-1}\equiv B(P_i)\circ P_i^{-1}$, and $Q_q$ is contained in $F(P_i)$.
Let $P_i$ be a segment contained in $B(Q_q)$.
From the condition~\ref{def:augment_pair:4} of the $y$-augmenting pairs (Definition~\ref{def:augment_pair}), $P_i$ has the same direction as $Q_q$.
Thus, $Q_q$ is contained in $F(P_i)$. Therefore, from the condition~\ref{def:augment_pair:4} again,
$B(P_{i-2})\circ P_{i-1}\equiv B(P_i)\circ P_i^{-1}$ holds.
Let $P_i$ be a segment contained in $F(Q_q)$.
As we have discussed earlier, $B(P_{i-2})\circ P_{i-1}\equiv B(P_i)\circ P_i^{-1}$ then holds, and $P_i$ has the
opposite direction as $Q_q$, implying that $Q_q$ is contained in $F(P_i)$.

Next, we prove that $(P',Q')$ is a $y$-augmenting pair.
Conditions~\ref{def:augment_pair:1} and~\ref{def:augment_pair:3} of the $y$-augmenting pairs (Definition~\ref{def:augment_pair}) are clearly satisfied.
Because $T(P')=B(Q_q^{-1})=F(Q_q)\not\equiv B(Q_{q-2})\circ Q_{q-1}=T(Q')$ holds, the condition~\ref{def:augment_pair:2}
is satisfied.
For checking the condition~\ref{def:augment_pair:4}, it suffices to show that
(1) none of the $Q_j$'s with $j<q$ are contained in $B(Q_q^{-1})=F(Q_q)$ in the same direction as $Q_q$;
(2) if $B(P_{p-1})\circ P_p\not\equiv B(Q_q^{-1})\circ (Q_q^{-1})^{-1}=F(Q_q)\circ Q_q$,
none of the $Q_j$'s with $j<q$ are contained in $F(Q_q^{-1})=B(Q_q)$ in the opposite direction as $Q_q$;
(3) for any $Q_j$ contained in an integral path, $Q_q$ is not contained in $B(Q_j)$ in the same direction as $Q_j$; and
(4) if $B(Q_{j-2})\circ Q_{j-1}\not\equiv B(Q_j)\circ Q_j^{-1}$ holds, $Q_q$ is not contained in $F(Q_j)$
in the opposite direction as $Q_j$.
All these conditions directly follow from the condition~\ref{def:alternating:cycle} of the $y$-alternating path
$Q$ (Definition~\ref{def:alternating}).

Finally, we prove that $(P', Q')$ satisfies the
conditions~\ref{lem:augment_pair:shortcut:2}--\ref{lem:augment_pair:shortcut:4} of this lemma.
Conditions~\ref{lem:augment_pair:shortcut:2} and~\ref{lem:augment_pair:shortcut:3} are clearly satisfied.
From the condition~\ref{def:alternating:cycle} of the initial $y$-alternating path $Q$, $B(Q_q)$ contains none of the new segments $S$.
Thus, if a segment in $S$ is contained in the same integral path as $Q_q$, it must be contained in $F(Q_q)$.
Therefore, it has the same direction as $Q_q^{-1}$.
\end{proof}

Finally, we prove Lemma~\ref{lem:augment_pair} by combining
Lemmas~\ref{lem:augment_pair:detour} and~\ref{lem:augment_pair:shortcut}.

\begin{proof}[Proof of Lemma~\ref{lem:augment_pair}]
First, we apply Lemma~\ref{lem:augment_pair:shortcut} against $(Q, P)$.
If we obtain a $y$-augmenting path, we obtain a basic $\F$-packing of size
$|y|+\frac{1}{2}$ by applying Lemma~\ref{lem:augment_path}.
Otherwise, we obtain an updated $y$-augmenting pair $(\bar{Q},\bar{P})$ such that
all the segments of $\bar{P}$, except for the last one, are contained in the common prefix.

Next, we apply Lemma~\ref{lem:augment_pair:shortcut} against $(\bar{P}, \bar{Q})$.
We finish if we obtain a $y$-augmenting path; otherwise, we obtain an updated $y$-augmenting pair $(\hat{P},\hat{Q})$.
Let $\hat{P}:=\hat{P}_1\circ\cdots\circ\hat{P}_{\hat{p}}$ and $\hat{Q}:=\hat{Q}_1\circ\cdots\circ\hat{Q}_{\hat{q}}$.
Then, $\hat{Q}_{\hat{q}-1}$ is contained in $\hat{P}_{\hat{q}-1}$.
Note that this implies that $\hat{P}_i=\hat{Q}_i$ holds for any $i\in\{1,\ldots,\hat{q}-2\}$. Moreover, $\hat{P}_{\hat{q}-1}=\hat{Q}_{\hat{q}-1}$ also holds if $\hat{q}$ is even.
We have $\hat{q}\geq \bar{p}-1$
because the common prefix of $(\hat{P},\hat{Q})$ contains the common prefix of $(\bar{P},\bar{Q})$
and because $\hat{P}_{\bar{p}-1}$ is contained in the common prefix of $(\bar{P},\bar{Q})$.
Let us assume that $\hat{q}=\bar{p}-1$ holds, and $\hat{q}$ is odd.
In this case, $\hat{P}_{\hat{q}}$ is contained in $\hat{Q}_{\hat{q}}$. Therefore, $\hat{P}_{\hat{q}}=\hat{Q}_{\hat{q}}$ holds.
Then, from the condition~\ref{def:augment_pair:3} of the $y$-augmenting pair $(\hat{P},\hat{Q})$ (Definition~\ref{def:augment_pair}),
$\hat{p}$ must be even, and $\hat{P}_{\hat{p}}$ must be contained in the same integral path as $\hat{P}_{\bar{p}}$.
However, this violates the condition~\ref{def:alternating:cycle} of the $y$-alternating path $\hat{P}$ because $B(\hat{P}_{\bar{p}-2})\circ\hat{P}_{\bar{p}-1}=T(\hat{Q})\not\equiv T(\hat{P})=B(\hat{P}_{\hat{p}})=B(\hat{P}_{\bar{p}})\circ \hat{P}_{\bar{p}}^{-1}$ holds.
Thus, $\hat{q}\geq \bar{p}$ holds if $\hat{q}$ is odd.
We consider two cases.

If $\hat{q}$ is odd or $\hat{P}_{\hat{q}}$ and $\hat{Q}_{\hat{q}}$ have the same direction,
let $r$ be the maximum even integer at most $\hat{q}$.
We apply Corollary~\ref{cor:alternating_sub} independently against $\hat{P}$ and $\hat{Q}$ and obtain
a basic $\F$-packing $y'$ of the same size and $y'$-alternating paths $P'$ and $Q'$.
Here, $\hat{P}_i=\hat{Q}_i$ holds for any $i\in\{1,\ldots,r-1\}$, and $F(\hat{P}_r)=F(\hat{Q}_r)$ holds.
Hence, the basic $\F$-packing obtained by these two applications are the same. We then, thus, use the same symbol $y'$.
Because $T_{y'}(P')\equiv T_y(\hat{P})\not\equiv T_y(\hat{Q})\equiv T_{y'}(Q')$ holds, $(P',Q')$ is a
$y'$-augmenting pair.
We now show that $(P',Q')$ satisfies the conditions in Lemma~\ref{lem:augment_pair:detour}.
Because $Q'$ is single-segment, the third condition is satisfied.
We have $r+2=\hat{q}+1>\bar{p}$ when $\hat{q}$ is odd, and we have $r+2=\hat{q}+2>\bar{p}$ when $\hat{q}$ is even.
Therefore, from the condition~\ref{lem:augment_pair:shortcut:4} of Lemma~\ref{lem:augment_pair:shortcut}, the first and the second conditions hold.
Thus, we obtain a basic $\F$-packing of size $|y|+\frac{1}{2}$ by applying Lemma~\ref{lem:augment_pair:detour} against $(P',Q')$.

If $\hat{q}$ is even and $\hat{P}_{\hat{q}}$ and $\hat{Q}_{\hat{q}}$ have the opposite direction,
let $r:=\hat{q}-2$.
We apply Corollary~\ref{cor:alternating_sub} against independently $\hat{P}$ and $\hat{Q}$ and obtain
a basic $\F$-packing $y'$ of the same size and $y'$-alternating paths $P'$ and $Q'$.
Here, $\hat{P}_i=\hat{Q}_i$ holds for any $i\in\{1,\ldots,r\}$.
Hence, the basic $\F$-packings obtained by these two applications are the same. We, thus, use the same symbol $y'$.
Because $T_{y'}(P')\equiv T_y(\hat{P})\not\equiv T_y(\hat{Q})\equiv T_{y'}(Q')$ holds, $(P',Q')$ is a
$y'$-augmenting pair.
We now show that $(P',Q')$ satisfies the conditions in Lemma~\ref{lem:augment_pair:detour}.
Because we set $r:=\hat{q}-2$, the number of segments of $Q'$ is two.
Because $\hat{P}_{\hat{q}-1}=\hat{Q}_{\hat{q}-1}$ holds, the first segments of $P'$ and $Q'$ are the same.
Because $\hat{P}_{\hat{q}}$ and $\hat{Q}_{\hat{q}}$ have the opposite direction, the second segments of $P'$ and $Q'$ have the opposite direction.
Thus, the third condition is satisfied.
Note that in this case, $\hat{P}_{\hat{q}}$ and $\hat{Q}_{\hat{q}}$ are contained in an integral path and,
from the condition~\ref{def:alternating:cycle} of the $y$-alternating paths (Definition~\ref{def:alternating}),
the integral path containing $\hat{P}_{\hat{q}}$ and $\hat{Q}_{\hat{q}}$ never contain any other segments.
Moreover, because $r+4=\hat{q}+2>\bar{p}$, the first and second conditions follow from the condition~\ref{lem:augment_pair:shortcut:4} of Lemma~\ref{lem:augment_pair:shortcut}.
Thus, we obtain a basic $\F$-packing of size $|y|+\frac{1}{2}$ by applying Lemma~\ref{lem:augment_pair:detour} against $(P',Q')$.
\end{proof}


\section{Farthest Cover}\label{sec:farthest}
This section provides an algorithm for computing a farthest minimum half-integral $\F$-cover and proves the following theorem.
\begin{theorem}\label{thm:farthest}
Let $C$ be a set of 0/1/all constraints on variables $V$ and $\varphi_A$ be a partial assignment for a subset $A\subseteq V$.
Given the primal graph of $C$, the set $A$, an incremental-test oracle for $(C,\varphi_A)$, and an integer $k$,
we can compute a pair of a farthest minimum half-integral $\F_{C,\varphi_A}$-cover $x$ and a maximum half-integral $\F_{C,\varphi_A}$-packing $y$ with $|x|=|y|\leq\frac{k}{2}$
or correctly conclude that the size of the minimum half-integral $\F_{C,\varphi_A}$-cover is at least $\frac{k+1}{2}$ in $O(kmT)$ time,
where $m$ is the number of constraints, and $T$ is the running time of the incremental-test oracle.
\end{theorem}

We use the following structure of a minimum half-integral $\F$-cover.
\begin{lemma}\label{lem:farthest:dual}
The following holds for any maximum basic $\F$-packing $y$ and any minimum half-integral $\F$-cover $x$.
\begin{enumerate}
  \item $x(V(I))=1$ for any integral path $I$ of $y$.
  \item $x(V(S))=\frac{1}{2}$ for any spoke $S$ of $y$.
\end{enumerate}
\end{lemma}
\begin{proof}
For a wheel $W$, we denote the degree of $W$ by $d(W)$ and the set of vertices contained in $W$ by $V(W)$.
Note that $V(W)$ is not a multiset.
First, we prove that $x(V(W))$ is always at least $\frac{d(W)}{2}$ and if the equality holds, then $x(V(S))=\frac{1}{2}$ holds for any spoke $S$ of $W$.
Let $W$ be a wheel with a half-integral cycle $H_1\circ\cdots\circ H_d$ and spokes $\{S_1,\ldots,S_d\}$.
We define $H:=V(W)\setminus(V(S_1)\cup\ldots\cup V(S_d))$.
We then have
\begin{align*}
x(V(W))&=x(H)+\sum_{i=1}^d x(V(S_i))\\
&=\frac{1}{2}x(H)+\frac{1}{2}\sum_{i=1}^d x(V(S_i\circ H_i\circ S_{i+1}^{-1}))\\
&\geq \frac{1}{2}\sum_{i=1}^d 1\\
&=\frac{d}{2}.
\end{align*}
If $x(V(W))=\frac{d}{2}$ holds, we have $x(H)=0$ and $x(V(S_i))+x(V(S_{i+1}))=1$ for any $i\in\{1,\ldots,d\}$.
Because $d$ is odd, this implies that $x(V(S_i))=\frac{1}{2}$ holds for any $i\in\{1,\ldots,d\}$.

We now prove the lemma.
Because all the integral paths and wheels do not share any vertices, we have
\begin{align*}
|x|&\geq\sum_{I:\text{integral path}}x(V(I))+\sum_{W:\text{wheel}}x(V(W))\\
&\geq \sum_{I:\text{integral path}}1+\sum_{W:\text{wheel}}\frac{d(W)}{2}\\
&=|y|\\
&=|x|.
\end{align*}
Therefore, $x(V(I))=1$ and $x(V(W))=\frac{d(W)}{2}$ must hold for any integral path $I$ and any wheel $W$.
\end{proof}

Fix an arbitrary maximum basic $\F$-packing.
From the above lemma, for every minimum half-integral $\F$-cover $x$, we can construct unique indices $a_x$ and $b_x$ satisfying the following.
\begin{enumerate}
  \item For every integral path $I=(v_0,\ldots,v_\ell)$, $x(v_{a_x(I)})\geq\frac{1}{2}$, $x(v_{b_x(I)})\geq\frac{1}{2}$ and $a_x(I)\leq b_x(I)$ hold.
  \item For every spoke $S=(v_0,\ldots,v_\ell)$, $x(v_{a_x(S)})=\frac{1}{2}$ holds.
\end{enumerate}

Therefore, we obtain the following corollary.
\begin{corollary}\label{cor:farthest:dom}
A minimum half-integral $\F$-cover $x'$ dominates a minimum half-integral $\F$-cover $x$ if and only if the following conditions hold.
\begin{enumerate}
  \item For every integral path $I$, $a_x(I)\leq a_{x'}(I)\leq b_{x'}(I)\leq b_x(I)$ holds.
  \item For every spoke $S$, $a_x(S)\leq a_{x'}(S)$ holds.
\end{enumerate}
\end{corollary}

We use the following operation in our algorithm.
Let $P$ be an implicational walk ending at a vertex $t$.
First, we create a new vertex $t'$ and introduce a constraint $\varphi(t)=\varphi(t')$ (along with a new edge $tt'$).
We then insert $t'$ into $A$ and set $\varphi_A(t')=\imp(P)$.
We call this operation as \emph{contracting $P$}.

\begin{lemma}\label{lem:farthest:contract}
Let $P$ be an implicational walk ending at a vertex $t$ and let $\F'$ be the set of conflicting walks after contracting $P$.
Then, any half-integral $\F'$-cover $x$ is also a half-integral $\F$-cover,
and any half-integral $\F$-cover $x$ satisfying $x(V(P)\setminus \{t\})=0$ is also a half-integral $\F'$-cover.
\end{lemma}
\begin{proof}
The first claim is trivial because $\F\subseteq\F'$.
We now prove the second claim.
Suppose that $x$ is not a half-integral $\F'$-cover.
Then, there exists a walk $W\in\F'$ with $x(V(W))<1$.
Because $x$ is a half-integral $\F$-cover, $W$ must contain the edge $t't$.
Therefore, we can write $W=t't\circ Q^{-1}$ for some implicational walk $Q$ with $\imp(Q)\neq\imp(P)$.
We then have $x(V(P\circ Q^{-1}))=x(V(W))<1$ and $P\circ Q^{-1}\in\F$, which contradicts the fact that $x$ is a half-integral $\F$-cover.
\end{proof}

\begin{algorithm}[t!]
\caption{Algorithm for computing a farthest minimum half-integral $\F$-cover}
\label{alg:farthest}
\begin{algorithmic}[1]
\State Compute a maximum basic $\F$-packing $y$.
\For{each integral path $I=(v_0,\ldots,v_\ell)$ of $y$}
	\While{$a(I)<b(I)$}\label{line:farthest:cycle:begin1}
		\State Contract $(v_0,\ldots,v_{a(I)+1})$.
		\If{there exists a $y$-augmenting path or pair}
			\State Rewind the search and contraction of this iteration; \textbf{break}
		\EndIf
	\EndWhile\label{line:farthest:cycle:end1}
	\While{$a(I)<b(I)$}\label{line:farthest:cycle:begin2}
		\State Contract $(v_\ell,\ldots,v_{b(I)-1})$.
		\If{there exists a $y$-augmenting path or pair}
			\State Rewind the search and contraction of this iteration; \textbf{break}
		\EndIf
	\EndWhile\label{line:farthest:cycle:end2}
\EndFor
\For{each spoke $S=(v_0,\ldots,v_\ell)$ of $y$}
	\While{$a(S)<\ell$}\label{line:farthest:spoke:begin}
		\State Contract $(v_0,\ldots,v_{a(I)+1})$.
		\If{there exists a $y$-augmenting path or pair}
			\State Rewind the search and contraction of this iteration; \textbf{break}
		\EndIf
	\EndWhile\label{line:farthest:spoke:end}
\EndFor
\State \Return the minimum half-integral $\F$-cover constructed from the current tables.
\end{algorithmic}
\end{algorithm}

We now describe the algorithm for computing a farthest minimum half-integral $\F$-cover (see Algorithm~\ref{alg:farthest}).
We iteratively apply the contraction in the algorithm.
We denote the current set of the conflicting walks by $\F$ and the original set by $\F_\text{orig}$.
First, we compute a maximum basic $\F$-packing $y$ using the algorithm in Section~\ref{sec:hi}.
We keep and reuse the tables ($a$, $b$, $\prev$, and $\tail$) used in the last execution of Algorithm~\ref{alg:search}, which returned NO.
We process the integral paths and the spokes of $y$ one by one, whose detail will be described later, by preserving the following invariants.
\begin{lemma}
The following invariants hold at any step of Algorithm~\ref{alg:farthest}.
\begin{enumerate}
  \item $y$ is a maximum basic $\F$-packing.\label{lem:farthest:invariants:y}
  \item Let $x$ be the minimum half-integral $\F$-cover constructed from the current tables as described in Section~\ref{sec:hi:cover} (i.e., $a_x=a$ and $b_x=b$ hold).
  		Then, any minimum half-integral $\F_\text{orig}$-cover dominating $x$ is also a minimum half-integral $\F$-cover dominating $x$.\label{lem:farthest:invariants:cover}
  \item For any processed integral path $I$, there exists no minimum half-integral $\F$-cover $x'$ satisfying $a(I)<a_{x'}(I)\leq b_{x'}(I)\leq b(I)$ or $a(I)\leq a_{x'}(I)\leq b_{x'}(I)<b(I)$.\label{lem:farthest:invariants:cycle}
  \item For any processed spoke $S$, there exists no minimum half-integral $\F$-cover $x'$ satisfying $a(S)<a_{x'}(S)$.\label{lem:farthest:invariants:spoke}
\end{enumerate}
\end{lemma}

When all the integral paths and the spokes are processed, we return the minimum half-integral $\F$-cover $x$ constructed from the current tables.
We can easily prove the correctness of the algorithm from these invariants.
\begin{lemma}
When all the integral paths and the spokes are processed, the minimum half-integral $\F$-cover $x$ of $G$ constructed from the current
tables is a farthest minimum half-integral $\F_\text{orig}$-cover.
\end{lemma}
\begin{proof}
From the invariants~\ref{lem:farthest:invariants:cycle} and~\ref{lem:farthest:invariants:spoke} and Corollary~\ref{cor:farthest:dom}, $x$ is a farthest minimum half-integral $\F$-cover.
Because $\F_\text{orig}\subseteq \F$ and from the invariant~\ref{lem:farthest:invariants:y}, $x$ is also a minimum half-integral $\F_\text{orig}$-cover.
Suppose that there exists a minimum half-integral $\F_\text{orig}$-cover $x'$ dominating~$x$.
Then, from the invariant~\ref{lem:farthest:invariants:cover}, $x'$ is a minimum half-integral $\F$-cover dominating $x$, which contradicts the fact that $x$ is a farthest minimum half-integral $\F$-cover.
\end{proof}

We now describe how to process the integral paths and spokes.
For each integral path $I=(v_0,\ldots,v_\ell)$, we first repeat the following while $a(I)<b(I)$ holds
(lines~\ref{line:farthest:cycle:begin1}--\ref{line:farthest:cycle:end1}).
Let $i:=a(I)+1$.
We first contract the implicational path $(v_0,\ldots,v_i)$.
When using the incremental-test oracle, this operation can be done in a constant time by setting $\mathcal{I}(v_i')=\mathcal{A^*}((v_0,\ldots,v_i))$, which has been precomputed.
We then search for a $y$-augmenting path or pair using Algorithm~\ref{alg:search}.
Instead of searching for a $y$-augmenting path/pair from scratch by initializing the tables,
we restart the search from line~\ref{line:search:pickA} of Algorithm~\ref{alg:search} by setting $s\gets v_i'$ and reusing the current tables.
If the restarted search returns NO, we keep the tables updated by the restarted search and continue the repetition.
Because $v_i'v_i\not\equiv (v_\ell,\ldots,v_i)$ holds, we have $a(I)\geq i$ after the search.
If the restarted search finds a $y$-augmenting path or pair, we rewind all the changes in this step
(i.e., we restore the tables to the state before the search and remove the edge $v_i'v_i$ inserted by the contraction),
and then exit the repetition.
Note that we do not rewind the changes in the previous steps where the restarted searches returned NO.

\begin{claim}
The restarted search can correctly compute a $y$-augmenting path or pair if exists.
\end{claim}
\begin{proof}
We can virtually think as follows.
Because the order of $A\setminus V(y)$ at line~\ref{line:search:outer} of Algorithm~\ref{alg:search} is arbitrary, we can choose completely the
same order as the one used in the last failed search that has constructed the current tables.
Note that, from the invariant~\ref{lem:alg:invariants:cycle:visited} in Lemma~\ref{lem:alg:invariants},
$v_i$ was not visited in the last execution. Therefore, the insertion of the edge $v_i'v_i$ does not affect the search at all.
In the end, the execution reaches to the final iteration of the while loop with $s=v_i'$, and all the tables are completely the same as the current tables.
Thus, instead of running Algorithm~\ref{alg:search} from scratch, we can use the restarted search.
\end{proof}

\begin{claim}
Lines~\ref{line:farthest:cycle:begin1}--\ref{line:farthest:cycle:end1} preserve all the invariants.
\end{claim}
\begin{proof}
Because we keep the changes only when the search fails, the invariant~\ref{lem:farthest:invariants:y} is satisfied.
The invariant~\ref{lem:farthest:invariants:cover} follows from Lemma~\ref{lem:farthest:contract} against $P:=(v_0,\ldots,v_i)$.
The invariants~\ref{lem:farthest:invariants:cycle} and~\ref{lem:farthest:invariants:spoke} are trivial.
Note that the current integral path $I$ is still under processing.
\end{proof}

\begin{claim}
When the repetition of lines~\ref{line:farthest:cycle:begin1}--\ref{line:farthest:cycle:end1} ends, there exists no minimum half-integral $\F$-cover $x'$
satisfying $a(I)<a_{x'}(I)\leq b_{x'}(I)\leq b(I)$.
\end{claim}
\begin{proof}
The repetition ends when $a(I)$ becomes equal to $b(I)$ or when a $y$-augmenting path or pair is found.
The former case is trivial.
In the latter case, let $\F'$ be the set of the conflicting walks after contracting $(v_0,\ldots,v_{a(I)+1})$.
Suppose that there exists a minimum half-integral $\F$-cover $x'$ satisfying the condition in the lemma.
Then, from Lemma~\ref{lem:farthest:contract}, $x'$ is also a half-integral $\F'$-cover.
In contrast, the size of the maximum half-integral $\F'$-packing $y'$ is strictly larger than~$|y|$ because a $y$-augmenting path or pair is found.
Thus, we have $|y'|>|y|=|x'|\geq |y'|$, which is a contradiction.
\end{proof}

After repeating lines~\ref{line:farthest:cycle:begin1}--\ref{line:farthest:cycle:end1}, we repeat the following while $a(I)<b(I)$ holds
(lines~\ref{line:farthest:cycle:begin2}--\ref{line:farthest:cycle:end2}).
Let $i:=b(I)-1$.
We first contract the implicational path $(v_\ell,\ldots,v_i)$, and then we restart the search by setting $s\gets v_i'$.
If the restarted search fails, we continue the repetition,
and if the restarted search succeeds, we rewind all the changes in this step and exit the repetition.
By the same argument as in the case of lines~\ref{line:farthest:cycle:begin1}--\ref{line:farthest:cycle:end1},
all the invariants are preserved and, when the repetition ends, there exists no minimum half-integral $\F$-cover $x'$ satisfying
$a(I)\leq a_{x'}(I)\leq b_{x'}(I)<b(I)$.
Thus, the invariant~\ref{lem:farthest:invariants:cycle} is satisfied for $I$ when we have finished processing $I$.

Next, we repeat the following for each spoke $S=(v_0,\ldots,v_\ell)$ while $a(S)<\ell$ holds
(lines~\ref{line:farthest:spoke:begin}--\ref{line:farthest:spoke:end}).
Let $i:=a(S)+1$.
We first contract the implicational path $(v_0,\ldots,v_i)$, and then restart the search by setting $s\gets v_i'$.
We continue the repetition if the restarted search fails,
and, otherwise, we rewind all the changes in this step and exit the repetition.
By the same argument as in the case of integral paths,
all the invariants are preserved and, when the repetition ends, there exists no minimum half-integral $\F$-cover $x'$ satisfying
$a(S)<a_{x'}(S)$.
Thus, the invariant~\ref{lem:farthest:invariants:spoke} is satisfied for $S$ when we have finished processing $S$.

Finally, we analyze the running time of Algorithm~\ref{alg:farthest}.
The number of inserted edges is $O(n)$, and the number of while loops (lines~\ref{line:farthest:cycle:begin1}--\ref{line:farthest:cycle:end1},
\ref{line:farthest:cycle:begin2}--\ref{line:farthest:cycle:end2}, or~\ref{line:farthest:spoke:begin}--\ref{line:farthest:spoke:end}) is
$2k$.
For each while loop, a series of the restarted searches can be regarded as a single execution of Algorithm~\ref{alg:search}, which takes
$O(m T)$ time.
Therefore, the total running time is $O(kmT)$.
Thus, we obtain Theorem~\ref{thm:farthest}.


\section{Linear-Time FPT Algorithms}\label{sec:linear}

\subsection{Algorithm for \INTCOVER}\label{sec:fpt:algorithm}

This section proposes an $O(d^{2k}km)$-time algorithm for \INTCOVER.
Our algorithm is based on the branch-and-bound algorithm in~\cite{Iwata:2016ja}.
We exploit the farthest minimum half-integral $\F$-cover and \emph{parallel unit-propagation}
to obtain a linear-time FPT algorithm.

Let $I=(C,\varphi_A)$ be a pair of 0/1/all constraints $C$ on a variable set $V$ and a partial-assignment $\varphi_A$ on a subset $A\subseteq V$.
We denote $\F_{C,\varphi_A}$ by $\F_I$.
For a variable $u\in V\setminus A$ and an element $a\in D(u)$, we denote by $I[u\gets a]$ a pair $(C,\varphi_{A\cup\{u\}})$ such that $\varphi_{A\cup\{u\}}(u)=a$.
We denote by $I-u$ a pair $(C[V\setminus\{u\}],\varphi_{A\setminus\{u\}})$ for a variable $u\in V$, where $\varphi_{A\setminus\{u\}}$ is the restriction of $\varphi_A$ to $A\setminus\{u\}$.
We call this operation \emph{deleting $u$}.
Let $N_u:=\{v\in A\mid uv\in E,\ \varphi_A\text{ does not satisfy }C_{uv}\}$ for a variable $u\in A$.
We denote by $I/u$ a pair $(C[V\setminus(\{u\}\cup N_u)],\varphi_{A'})$ such that $A'=(A\cup\{v\mid C_{uv}(\varphi_A(u))\neq\all\})\setminus(\{u\}\cup N_u)$ and $\varphi_{A'}(v)=C_{uv}(\varphi_A(u))$ for $v\in A'\setminus A$.
We call this operation \emph{fixing $u$}.
We can obtain an incremental-test oracle for $I/u$ by setting $\mathcal{I}(v)=\mathcal{A}(\mathcal{I}(u), uv)$ for each $v\in A'\setminus A$.
We can observe the following.
\begin{lemma}\label{lem:fpt:delete}
The following holds for any pair $I=(C,\varphi_A)$ and $u\in V$.
\begin{enumerate}
  \item $I$ admits a deletion set of size $k$ containing $u$ if and only if $I-u$ admits a deletion set of size $k-1$.
  \item For any half-integral $\F_{I-u}$-cover $x'$, the following function $x$ is a half-integral $\F_I$-cover: $x(u)=1$ and $x(v)=x'(v)$ for $v\in V\setminus\{u\}$.
\end{enumerate}
\end{lemma}
\begin{proof}
First, we prove the first claim.
For a deletion set $X$ for $I$ containing $u$, $X\setminus\{u\}$ is a deletion set for $I-u$.
For a deletion set $X'$ for $I-u$, $X'\cup\{u\}$ is a deletion set for $I$.

Next, we prove the second claim.
Let $W$ be a walk in $\F_I$.
We have $x(V(W))\geq x'(V(W))\geq 1$ if $W\in\F_{I-u}$.
Otherwise, $W$ visits $u$; therefore, we have $x(V(W))\geq x(u)=1$.
\end{proof}
\begin{lemma}\label{lem:fpt:fix}
The following holds for any pair $I=(C,\varphi_A)$ and $u\in A$.
\begin{enumerate}
  \item $I$ admits a deletion set of size $k$ not containing $u$ if and only if $I/u$ admits a deletion set of size $k-|N_u|$.
  \item For any half-integral $\F_{I/u}$-cover $x'$, the following function $x$ is a half-integral $\F_I$-cover: $x(u)=0$, $x(v)=1$ for $v\in N_u$, and $x(v)=x'(v)$ for $v\in V\setminus(\{u\}\cup N_u)$.
\end{enumerate}
\end{lemma}
\begin{proof}
First, we prove the first claim.
Let $X$ be a deletion set for $I$ not containing $u$.
Because $X$ must contain all of $N_u$, $X\setminus N_u$ is a deletion set for $I/u$ of size $|X|-|N_u|$.
For a deletion set $X'$ for $I/u$, $X'\cup N_u$ is a deletion set for $I$.

Next, we prove the second claim.
Let $W$ be a walk in $\F_I$.
If $W\in\F_{I/u}$, we have $x(V(W))\geq x'(V(W))\geq 1$.
If $W$ visits a vertex in $N_u$, we have $x(V(W))\geq 1$.
Otherwise, we can write $W=uv\circ W'$.
Then, we have $W'\in\F_{I/u}$, and thus, we have $x(V(W))\geq x'(V(W'))\geq 1$.
\end{proof}

For a minimum half-integral $\F_I$-cover $x$, we define an operation called a \emph{persistency reduction} as follows.
We first delete every vertex in $x^{-1}(1)$ in an arbitrary order, and then fix every vertex in $R(x)$ in an arbitrary order\footnote{We pick an arbitrary vertex $u\in R(x)\cap A$ and fix $u$. This changes $R(x)$ and $A$, and we repeat the process until $R(x)$ becomes the empty set.}.
Because $x$ is a half-integral $\F_I$-cover, for any implicational walk $W$ with $x(V(W)\setminus\{t(W)\})=0$ and $x(t(W))\leq\frac{1}{2}$, the value $\imp(W)$ depends only on $t(W)$.
Therefore, the ordering does not affect the result and $N_u=\emptyset$ for every fixing.
We denote the obtained pair $(C[V\setminus (x^{-1}(1)\cup R(x))], \varphi_{A'})$ by $I/x$.

\begin{lemma}\label{lem:fpt:reduction1}
$I$ admits a deletion set of size $k$ if and only if $I/x$ admits a deletion set of size $k-|x^{-1}(1)|$.
\end{lemma}
\begin{proof}
From Theorem~\ref{thm:persistency}, there exists a minimum deletion set $X$ for $I$ such that $x^{-1}(1)\subseteq X\subseteq V\setminus R(x)$.
Therefore, from Lemma~\ref{lem:fpt:delete} and~\ref{lem:fpt:fix}, the claim holds.
\end{proof}

\begin{lemma}\label{lem:fpt:reduction2}
For a farthest minimum half-integral $\F_I$-cover $x$,
the restriction of $x$ to $V\setminus (x^{-1}(1)\cup R(x))$ is the unique minimum half-integral $\F_{I/x}$-cover.
\end{lemma}
\begin{proof}
Let $x'$ be the restriction of $x$. 
Then, $x'$ is a half-integral $\F_{I/x}$-cover because $x'(u)=\frac{1}{2}$ for every $u\in A'$.
Suppose that there exists a half-integral $\F_{I/x}$-cover $z'$ with $|z'|\leq |x'|$ and $z'\neq x'$.
From Lemmas~\ref{lem:fpt:delete} and~\ref{lem:fpt:fix}, the following function $z$ is a half-integral $\F$-cover:
$z(u)=z'(u)$ for $u\in V\setminus (x^{-1}(1)\cup R(x))$, $z(u)=1$ for $u\in x^{-1}(1)$, $z(u)=0$ for $u\in R(x)$.
Then, we have $|z|=|z'|+|x^{-1}(1)|\leq |x'|+|x^{-1}(1)|=|x|$.
Therefore, $z$ is a minimum half-integral $\F$-cover dominating $x$, which is a contradiction.
\end{proof}

Let $I=(C,\varphi_\emptyset)$ be a pair with $A=\emptyset$.
A set of pairs $B\subseteq\{(u,a)\mid u\in V,\ a\in D(u)\}$ is called a \emph{branching set} for $I$ if it has the following property:
any deletion set for $I$ is a deletion set for at least one of $I[u\gets a]$ with $(u,a)\in B$.
Note that any deletion set for $I[u\gets a]$ is a deletion set for~$I$.
The running time of our algorithm depends on the choice of branching sets.
In general, we can use the following standard choice:
pick a vertex $u\in V$ and set $B:=\{(u,a)\mid a\in D(u)\}$.
In the next section, we choose different branching sets for problem-specific improvements.
\begin{lemma}\label{lem:fpt:branching_general}
$B:=\{(u,a)\mid a\in D(u)\}$ for some $u\in V$ is a branching set.
\end{lemma}
\begin{proof}
Let $X$ be a deletion set for $I$, and let $\varphi$ be a satisfying assignment for $C[V\setminus X]$.
If $u\in X$, $X$ is a deletion set for every $I[u\gets a]$.
Otherwise, $X$ is a deletion set for $I[u\gets a]$ with $a=\varphi(u)$.
\end{proof}

\begin{algorithm}[t!]
\caption{A linear-time FPT algorithm for \INTCOVER}
\label{alg:bb}
\begin{algorithmic}[1]
\Procedure{Solve}{$I=(C,\varphi_A),k$}
\State Compute a farthest minimum half-integral $\F_I$-cover $x$.
\If{$|x|>k$}
	\Return $\mathbf{false}$
\EndIf
\State $I\gets I/x$, $k\gets k-|x^{-1}(1)|$.
\If{$V=\emptyset$}\label{line:bb:outer}
	\Return $\mathbf{true}$
\EndIf
\If{$A\neq\emptyset$}
	\State Pick a vertex $u\in A$.
	\State \Return $\Call{Solve}{I-u,k-1}\vee\Call{Solve}{I/u,k-|N_u|}$
\Else
	\State Choose a branching set $B$.
	\If{there exists $(u,a)\in B$ with $\F_{I[u\gets a]}=\emptyset$}
		\State $I\gets I-R$ and {\bf goto} line~\ref{line:bb:outer}, where $R$ is the set of the implicated variables.
	\Else
		\State \Return $\bigvee_{(u,a)\in B}\Call{Solve}{I[u\gets a], k}$
	\EndIf
\EndIf
\EndProcedure
\end{algorithmic}
\end{algorithm}

We now provide a linear-time FPT algorithm for \INTCOVER (Algorithm~\ref{alg:bb}).
We denote the size of the minimum half-integral $\F_I$-cover for a pair $I=(C,\varphi_A)$ by $c(I)$.
$\Call{Solve}{I,k}$ is a procedure that returns $\mathbf{true}$ if and only if $I$ admits a deletion set of size at most $k$.
We prove the following.
\begin{theorem}
Let $I=(C,\varphi_A)$ be a pair of a set $C$ of 0/1/all constraints on a variable set $V$ and a partial assignment $\varphi_A$ for a subset $A\subseteq V$.
We are given the primal graph of $C$, the set $A$, an $O(T)$-time incremental-test oracle for $(C,\varphi_A)$, and an integer $k$.
Under the following assumptions, Algorithm~\ref{alg:bb} correctly answers whether $(C,\varphi_A)$ admits a deletion set of size at most $k$ or not in $O(\max(2,b)^{2(k-c(I))} kmT)$ time,
where $b$ is the integer in the assumption and $m$ is the number of constraints.
\begin{enumerate}
  \item For any $V'\subseteq V$, we can choose a branching set for $(C[V'],\varphi_\emptyset)$ of size at most $b$.
  \item For any $V'\subseteq V$ and any $(u,a)\in B$ for a possible branching set $B$ for $(C[V'],\varphi_\emptyset)$, we have an $O(T)$-time incremental-test oracle for $(C[V'],\varphi_{\{u\}})$ with $\varphi_{\{u\}}(u)=a$.
\end{enumerate}
\end{theorem}

Note that for \INTCOVER, we assume that each constraint is given as a table of size $O(d)$.
Therefore, the naive implementation of the incremental-test oracle runs in a constant time.
When using the naive implementation of the incremental-test oracle, the second assumption trivially holds.
Thus, by using the standard choice of branching sets, the algorithm runs in $O(d^{2(k-c(I))}km)=O(d^{2k}km)$ time.

First, we compute a farthest minimum half-integral $\F_I$-cover $x$.
Because the size of $x$ is a lower bound on the size of the minimum deletion set, if $|x|>k$ holds, there exists no deletion set of size at most $k$.
Otherwise, we have $k-c(I)\geq 0$.
From Lemma~\ref{lem:fpt:reduction1}, we can apply the persistency reduction and decrease $k$ by $|x^{-1}(1)|$.
This does not change the difference $k-c(I)$ because $c(I)$ also decreases by $|x^{-1}(1)|$.
We rename the reduced instance to $I=(C,\varphi_A)$ for simplicity of the notation, and let $G=(V,E)$ be the primal graph of $C$.
We rename the restriction of $x$ to $V$ to $x$.
From Lemma~\ref{lem:fpt:reduction2}, $x$ is the unique minimum half-integral $\F_I$-cover that maps every vertex in $A$ to $\frac{1}{2}$ and every other vertex to $0$.
This part can be done in $O(kmT)$ time and we can construct an incremental-test oracle for the new $I$ from the incremental-test oracle for the old $I$.

If $A\neq\emptyset$, we pick a vertex $u\in A$ and return $\Call{Solve}{I-u,k-1}\vee \Call{Solve}{I/u,k-|N_u|}$.
\begin{claim}
$I$ admits a deletion set of size at most $k$ if and only if $I-u$ admits a deletion set of size at most $k-1$ or $I/u$ admits a deletion set of size at most $k-|N_u|$.
Moreover, both of $k-1-c(I-u)<k-c(I)$ and $k-|N_u|-c(I/u)<k-c(I)$ hold.
\end{claim}
\begin{proof}
The first claim follows from Lemma~\ref{lem:fpt:delete} and~\ref{lem:fpt:fix}.
We now prove the second claim.

Suppose that $I-u$ admits a half-integral $\F_{I-u}$-cover $z'$ of size at most $|x|-1=c(I)-1$.
Then, from Lemma~\ref{lem:fpt:delete}, we can construct a half-integral $\F_I$-cover $z$ of size at most $|x|$ with $z(u)=1$.
Therefore, $z$ is a minimum half-integral $\F_I$-cover with $x\neq z$, which contradicts the uniqueness of~$x$.

Suppose that $I/u$ admits a half-integral $\F_{I/u}$-cover $z'$ of size at most $|x|-|N_u|=c(I)-|N_u|$.
Then, from Lemma~\ref{lem:fpt:fix}, we can construct a half-integral $\F_I$-cover $z$ of size at most $|x|$ with $z(u)=0$.
Therefore, $z$ is a minimum half-integral $\F_I$-cover with $x\neq z$, which contradicts the uniqueness of~$x$.
\end{proof}

If $A=\emptyset$, we choose a branching set $B$ of size at most $b$.
For each $(u,a)\in B$, we check whether $\F_{I[u\gets a]}=\emptyset$ or not by the unit-propagation
(i.e., by applying Algorithm~\ref{alg:search} against the empty packing).
If there exists such $(u,a)$, let $R$ be the set of the implicated variables (i.e., $R(\mathbf{0})$ for the empty cover $\mathbf{0}$ that always returns $0$).
Let $m'$ be the number of edges incident to a vertex in $R$, and suppose that $(u,a)$ is the pair minimizing $m'$.
The unit-propagation for $(u,a)$ takes only $O(m'T)$ time.
Hence, by running the unit-propagation for every $(u,a)\in B$ in parallel, which is simulated in a single processor in a round-robin fashion,
we can find such $(u,a)$ in $O(bm'T)$ time.
We then set $I\gets I-R$, where $I-R$ is the pair obtained by deleting every vertex $v\in R$, and go back to line~\ref{line:bb:outer}.
\begin{claim}
$I$ admits a deletion set of size at most $k$ if and only if $I-R$ admits a deletion set of size at most $k$.
\end{claim}
\begin{proof}
Any deletion set for $I$ is also a deletion set for $I-R$.
Let $X$ be a deletion set for $I-R$, and let $\varphi_{V\setminus(R\cup X)}$ be a satisfying assignment for $C[V\setminus (R\cup X)]$.
From the construction of $R$, $C[R]$ admits a satisfying assignment $\varphi_R$, and every constraint $C_{vw}$ with $v\in R$ and $w\not\in R$ is a two-fan $(\varphi(v)=a)\vee(\varphi(w)=b)$ with $a=\varphi_R(v)$.
Therefore, an assignment $\varphi_{V\setminus X}$ such that $\varphi_{V\setminus X}(v)=\varphi_{R}(v)$ for $v\in R$ and $\varphi_{V\setminus X}(v)=\varphi_{V\setminus (R\cup X)}(v)$ for $v\in V\setminus (R\cup X)$ satisfies $C[V\setminus X]$.
Thus, $X$ is also a deletion set for~$I$.
\end{proof}

If $\F_{I[u\gets a]}\neq\emptyset$ for all $(u,a)\in B$, we return $\bigvee_{(u,a)\in B}\Call{Solve}{I[u\gets a]}$.
\begin{claim}
$I$ admits a deletion set of size at most $k$ if and only if at least one of $I[u\gets a]$ admits a deletion set of size $k$.
Moreover, $k-c(I[u\gets a])<k-c(I)$ for every $(u,a)\in B$.
\end{claim}
\begin{proof}
The first claim follows from the definition of the branching set.
Because $\F_{I[u\gets a]}\neq\emptyset$, we have $k-c(I[u\gets a])\leq k-\frac{1}{2}<k=k-c(I)$.
\end{proof}

We now have proved the correctness of the algorithm.
Finally, we analyze the running time.
Let $\mathcal{T}(\Delta)$ be the running time of $\Call{Solve}{I,k}$ when $\Delta:=k-c(I)$.
We can compute a farthest minimum half-integral $\F_I$-cover in $O(k m T)$ time.
If $A\neq\emptyset$, we branch into two cases and $\Delta$ decreases by at least $\frac{1}{2}$ for each case.
Therefore, we have $\mathcal{T}(\Delta)\leq 2\mathcal{T}(\Delta-\frac{1}{2})+O(kmT)$.
If $A=\emptyset$, we search for $(u,a)\in B$ with $\F_{I[u\gets a]}=\emptyset$ by the parallel unit-propagation.
If there exists such $(u,a)$, the parallel unit-propagation takes $O(bm'T)$ time, and $m$ decreases by $m'$.
Therefore, in $O(bmT)$ time, we reach to the state that either $V=\emptyset$ or there exists no such $(u,a)$.
In the latter case, we branch into at most $b$ cases, and $\Delta$ decreases by at least $\frac{1}{2}$ for each case.
Therefore, we have $\mathcal{T}(\Delta)\leq b\mathcal{T}(\Delta-\frac{1}{2})+O(kmT+bmT)$.
Thus, we have $\mathcal{T}(\Delta)=O(\max(2,b)^{2\Delta}kmT)$.

\subsection{Applications to Other Problems}\label{sec:fpt:formulation}

Finally, we show that various NP-hard problems can be expressed as a special case of \INTCOVER.
Note that we use $A=\emptyset$ for every problem other than \textsc{Node Multiway Cut}.
We obtain a linear-time FPT algorithm for each problem by giving an incremental-test oracle and a specialized choice of branching sets.

\parameterizedproblem{\textsc{Node Unique Label Cover}}
{A finite alphabet $\Sigma$, a graph $G=(V,E)$, a permutation $\pi_e$ of $\Sigma$ for every edge $e\in \hat{E}$ such that $\pi_{uv} = \pi_{vu}^{-1}$ given as a table of size $|\Sigma|$, and an integer $k$.}
{$k,|\Sigma|$}
{Is there a pair of set $X \subseteq V$ of at most $k$ vertices and assignment $\varphi : V \setminus X \to \Sigma$ such that $\pi_{uv}(\varphi(u))=\varphi(v)$ for every $uv \in E[V\setminus X]$?}

\parameterizedproblem{\textsc{Two-fan Deletion}}
{A set of variables $V$, a set of two-fan constraints on $V$ of the form $(\varphi(u)=a)\vee (\varphi(v)=b)$ given as a pair $(a,b)$, and an integer $k$.}
{$k$}
{Is there a pair of set $X \subseteq V$ of at most $k$ variables and assignment $\varphi$ satisfying every constraint $C_{uv} \in C[V \setminus X]$?}

These two problems are special cases of \INTCOVER such that the set of constraints is limited to permutation or two-fan constraints.
Note that the size of the domain is not a parameter for \textsc{Two-fan Deletion}.
The naive implementation of the incremental-test oracle for these problems runs in a constant time.
Thus, we can solve \textsc{Node Unique Label Cover} in $O(|\Sigma|^{2k}km)$ time.
We can use the following choice of a branching set for \textsc{Two-fan Deletion}:
pick a two-fan constraint $(\varphi(u)=a)\vee(\varphi(v)=b)$ and set $B:=\{(u,a),(v,b)\}$.
Thus, we can solve \textsc{Two-fan Deletion} in $O(2^{2k}km)=O(4^kkm)$ time.
\begin{lemma}
$B:=\{(u,a),(v,b)\}$ for a two-fan constraint $(\varphi(u)=a)\vee(\varphi(v)=b)$ is a branching set.
\end{lemma}
\begin{proof}
Let $X$ be a deletion set for $I$ and let $\varphi$ be an assignment for $V \setminus X$ satisfying $C[V\setminus X]$.
If $u\in X$, $X$ is a deletion set for $I[u\gets a]$, and if $v\in X$, $X$ is a deletion set for $I[v\gets b]$.
Otherwise, at least one of $\varphi(u)=a$ or $\varphi(v)=b$ holds.
$X$ is a deletion set for $I[u\gets a]$ in the former case, while $X$ is a deletion set for $I[v\gets b]$ in the latter case.
\end{proof}

The next two problems generalize \textsc{Pseudoforest Deletion} \cite{bodlaender2017faster}
in different directions,
where a \emph{pseudoforest} is a graph in which
the number of edges is at most the number of vertices for every connected component.

A graph is called \emph{monochromatically orientable} if there exists an edge orientation such that, for every vertex, all the incoming edges are monochromatic.
It is known that a graph is a pseudoforest if and only if there exists an edge orientation such that, for every vertex, the number of incoming edges is at most one.
Therefore, when every edge has a distinct color, a graph is monochromatically orientable if and only if it is a pseudoforest.
Thus, the following problem is a generalization of \textsc{Pseudoforest Deletion}.

\parameterizedproblem{\textsc{Monochromatically Orientable Deletion}}
{An edge-colored graph $G=(V,E)$ and an integer $k$.}
{$k$}
{Is there a set $X\subseteq V$ of at most $k$ vertices such that $G-X$ is monochromatically orientable?}

We can solve \textsc{Monochromatically Orientable Deletion} in $O(4^kkm)$ time by the following reduction to \textsc{Two-fan Deletion}.
Let $L$ be the set of colors.
For each vertex $v\in V$, we create a variable $v$ with domain $D(v)=L$, which represents the color of the incoming edges.
We create a two-fan constraint $(\varphi(u)=c)\vee (\varphi(v)=c)$ for each edge $uv\in E$ of color $c\in L$.
\begin{lemma}\label{lem:fpt:mono-ori}
A graph $G=(V,E)$ is monochromatically orientable if and only if the corresponding set $C$ of two-fan constraints is satisfiable.
\end{lemma}
\begin{proof}
From a monochromatic orientation of $G$, we can construct a satisfying assignment $\varphi$ as follows:
we set $\varphi(v):=c$ for each vertex $v\in V$, where $c$ is the unique color of the incoming edges or an arbitrary color if $v$ has no incoming edges.
From a satisfying assignment $\varphi$, we can construct a monochromatic orientation
by orienting each edge $uv\in E$ of color $c$ so that
the edge is directed toward $u$ if $\varphi(u)=c$, and toward $v$ otherwise.
\end{proof}

Another generalization of \textsc{Pseudoforest Deletion} is presented as follows.
For a graph $G=(V,E)$ and an edge $e=uv\in E$, \emph{contracting $e$} is an operation deleting the edge $e$ and merging $u$ and $v$ into a new vertex $e$.
Note that this operation may create parallel edges (edges $uw$ and $vw$ become parallel edges) and self-loops.
If $G$ has parallel edges connecting $u$ and $v$, the operation only removes one of them, and the rest becomes self-loops.

Let $S$ be a subset of edges.
A subset $X\subseteq V$ is called a \emph{subset feedback vertex set} if $G-X$ has no simple cycle passing through an edge of $S$, or equivalently, the graph obtained from $G-X$ by contracting every edge $e\in E[V \setminus X]\setminus S$ is a forest.
Similarly, we call $X$ a \emph{subset pseudoforest deletion set} if the graph obtained from $G-X$ by contracting every edge $e\in E[V \setminus X]\setminus S$ is a pseudoforest.

\parameterizedproblem{\textsc{Subset Pseudoforest Deletion}}
{A graph $G=(V,E)$, a set $S\subseteq E$, and an integer $k$.}
{$k$}
{Is there a set $X\subseteq V$ of at most $k$ vertices such that the graph obtained from $G-X$ by contracting every edge $e\in E[V\setminus X]\setminus S$ is a pseudoforest?}

This problem can be expressed as \INTCOVER as follows.
Every vertex $v$ has the same domain $D(v)=S$.
We introduce a two-fan constraint $(\varphi(u)=e)\vee (\varphi(v)=e)$ for every edge $e=uv\in S$.
We also introduce an equality (identity permutation) constraint $\varphi(u)=\varphi(v)$ for every edge $uv\not\in S$.
Let $C$ be the obtained set of constraints.
If $C[V']$ has no two-fan constraints for some $V'\subseteq V$, $(C[V'],\varphi_\emptyset)$ has a deletion set of size zero.
Therefore, we can use the same choice of branching sets as for \textsc{Two-fan Deletion}.
Thus, the algorithm runs in $O(4^k km)$ time.
The correctness of the expression follows from the following lemma.

\begin{lemma}
Let $G=(V,E)$ be a graph with a subset $S\subseteq E$ and let $C$ be the corresponding set of 0/1/all constraints.
Then, the graph obtained by contracting every edge $e\in E\setminus S$ is a pseudoforest if and only if $C$ is satisfiable.
\end{lemma}
\begin{proof}
We modify $C$ for each contraction of an edge $e=uv$ by removing the constraint $C_{uv}$ and replacing every occurrence of $u$ and $v$ by $e$.
When we obtain the graph $G'$ by contracting every edge $e\in E\setminus S$,
we also obtain the set $C'$ of constraints corresponding to $G'$.
Because every contracted edge $e=uv$ has the equality constraint, $C$ is satisfiable if and only if $C'$ is satisfiable.
From Lemma~\ref{lem:fpt:mono-ori}, $C'$ is satisfiable if and only if $G'$ is a pseudoforest.
Therefore, $C$ is satisfiable if and only if $G'$ is a pseudoforest.
\end{proof}

As mentioned in Section~\ref{sec:comparison}, \textsc{Node Multiway Cut} is also a special case of our problem.

\parameterizedproblem{\textsc{Node Multiway Cut}}
{A graph $G=(V,E)$, a set of terminals $T \subseteq V$, and an integer $k$.}
{$k$}
{Is there a set $X \subseteq V \setminus T$ of size at most $k$ such that every terminal in $T$ lies in a different connected component of $G - X$?}

This problem can be expressed as \INTCOVER as follows.
First, we split each terminal $s\in T$ as follows to make $s$ undeletable:
for each edge $sv\in\delta(s)$, we create a new vertex $s_v$ and replace the edge $sv$ with $s_vv$.
Let $G'=(V',E')$ be the obtained graph.
We introduce a variable $v$ with $D(v):=T$ for each vertex $v\in V'$
and an equality (identity permutation) constraint $\varphi(u)=\varphi(v)$ for each edge $uv\in E'$.
Finally, we set $A=\{s_v\mid s\in T,~sv\in\delta(s)\}$ and $\varphi_A(s_v):=s$.
We can observe that a minimum deletion set $X$ avoiding every $s_v$ always exists because each vertex $s_v$ has degree one,
and such $X$ is actually a minimum multiway cut.

Let $I=(C,\varphi_A)$ be the obtained instance.
We can construct a multiway cut of size at most $2|x|$ by rounding up a half-integral $\F_I$-cover $x$.
Therefore, we have $k-c(I)\leq \frac{1}{2}k$.
Because the domain size is $|T|=O(n)$, the naive implementation of the incremental-test oracle runs in a constant time.
Because any $I=(C[V'],\varphi_\emptyset)$ has a deletion set of size zero, we do not need branching sets.
Thus, we can solve \textsc{Node Multiway Cut} in $O(2^{2\cdot\frac{1}{2}k}km)=O(2^kkm)$ time.

\medskip
We finally show \textsc{Group Feedback Vertex Set} and its applications.

\parameterizedproblem{\textsc{Group Feedback Vertex Set}}
{A group $\Gamma=(D,\cdot)$ given as an $O(T_\Gamma)$-time oracle performing the group operation~$(\cdot)$, a $\Gamma$-labeled graph $G=(V,E)$ with labeling $\lambda:\hat{E} \to D$ with $\lambda(uv)\cdot\lambda(vu) = 1_\Gamma$ for every $uv \in \hat{E}$, where $1_\Gamma$ is the unity of $\Gamma$, and an integer $k$.}
{$k$}
{Is there a set $X \subseteq V$ of at most $k$ vertices such that $G-X$ has a consistent labeling? That is, is there a labeling $\varphi : V \setminus X \to D$ such that $\varphi(u)\cdot\lambda(uv) = \varphi(v)$ for every $uv \in E[V\setminus X]$?}

This problem can be expressed as \INTCOVER because a function $\pi_e(a):=a\cdot\lambda(e)$ is a permutation.
Note that $G-X$ has a consistent labeling if and only if it admits no non-zero cycles
(i.e., it admits no cycle $(v_0,\ldots,v_\ell)$ with $\lambda(v_0v_1)\cdot\lambda(v_1v_2)\cdots\lambda(v_{\ell-1}v_\ell)\neq 1_\Gamma$).
In contrast to \textsc{Node Unique Label Cover}, the domain size is not a parameter, and each permutation is given not as a table of size $|D|$, but as an $O(T_\Gamma)$-time oracle answering $a\cdot b$ for given $a,b\in D$.
Therefore, the naive implementation of the incremental-test oracle runs in $O(T_\Gamma)$ time.
We can use the following choice of a branching set:
pick a vertex $s$ and set $B:=\{(s,1_\Gamma)\}$.
Thus, we can solve \textsc{Group Feedback Vertex Set} in $O(2^{2k}km)=O(4^kkm)$ time.
\begin{lemma}
$B:=\{(s,1_\Gamma)\}$ is a branching set for \textsc{Group Feedback Vertex Set}.
\end{lemma}
\begin{proof}
Let $X$ be a deletion set for $I$, and let $\varphi$ be an assignment for $V \setminus X$ satisfying $C[V\setminus X]$.
If $s\in X$, $X$ is a deletion set for $I[s\gets 1_\Gamma]$;
otherwise, $\varphi'$ such that $\varphi'(v):=\varphi(v)\cdot\varphi(s)^{-1}$ is a satisfying assignment for $C[V\setminus X]$ with $\varphi'(s)=1_\Gamma$.
\end{proof}

\parameterizedproblem{\textsc{Subset Feedback Vertex Set}}
{A graph $G=(V,E)$, a set $S \subseteq E$, and an integer $k$.}
{$k$}
{Is there a set $X \subseteq V$ of at most $k$ vertices such that no cycle passes through an edge of $S$ in $G-X$?}

This problem can be expressed as \textsc{Group Feedback Vertex Set} as follows.
We use a group $\Gamma=(2^S,\oplus)$, where $\oplus$ is the XOR operator ($X\oplus Y=(X\setminus Y)\cup (Y\setminus X)$).
We set $\lambda(e)=\{e\}$ for each edge $e\in S$ and $\lambda(e)=\emptyset$ for each edge $e\in E\setminus S$.
Then, a cycle is non-zero if and only if it contains an edge in $S$.

The group operation takes $O(|S|)=O(m)$ time. Therefore, the naive implementation of the incremental-test oracle takes $O(m)$ time.
We now provide a constant-time incremental-test oracle for $(C[V'],\varphi_{\{s\}})$ with $\varphi_{\{s\}}(s)=\emptyset$.
We use the following implementation.
\[
U:=S\cup\{\epsilon\}.\quad
\mathcal{I}(s)=\epsilon.\quad
\mathcal{A}(a,e)=\begin{cases}
  	e&(e\in S)\\
  	a&(e\not\in S).
  \end{cases}\quad
\mathcal{T}(a,b)=\begin{cases}
  	\mathbf{true}&(a\neq b)\\
  	\mathbf{false}&(a=b).
  \end{cases}
\]

For a walk $W$, the function $\mathcal{A}^*(W)$ returns the last edge of $W$ contained in $S$ or $\epsilon$ if $W$ contains no edges in $S$.
Then, for a single-branching pair $(P,Q)$, we have $P\not\equiv Q\iff$ the simple cycle contained in $P\circ Q^{-1}$ contains an edge in $S$
$\iff\mathcal{A}^*(P)\neq\mathcal{A}^*(Q)$.
Therefore, the abovementioned implementation is correct.
Thus, we can solve \textsc{Subset Feedback Vertex Set} in $O(4^k km)$ time.

\parameterizedproblem{\textsc{Non-monochromatic Cycle Transversal}}
{An edge-colored graph $G=(V,E)$ and an integer $k$.}
{$k$}
{Is there a set $X\subseteq V$ of at most $k$ vertices such that $G-X$ contains no non-monochromatic cycles?}

This problem can be expressed as \textsc{Group Feedback Vertex Set} as follows.
Let $L$ be the set of colors, and let $c(e)\in L$ denote the color of an edge $e$.
We use the group $\Gamma=(2^{V\times L},\oplus)$.
We set $\lambda(e)=\{(u,c(e)),(v,c(e))\}$ for each edge $e=uv\in E$.
Then, a cycle is non-zero if and only if it is non-monochromatic.

The naive implementation of the incremental-test oracle takes $O(n|L|)=O(nm)$ time.
We now provide a constant-time incremental-test oracle for $(C[V'],\varphi_{\{s\}})$ with $\varphi_A(s)=\emptyset$.
We use the following implementation.
\[
U:=(V\times L)\cup\{\epsilon,*\}.\quad
\mathcal{I}(s)=\epsilon.\quad
\mathcal{T}(a,b)=\begin{cases}
\mathbf{true}&(*\neq a\neq b\neq *)\\
\mathbf{false}&\text{otherwise}.
\end{cases}
\]
\[
\mathcal{A}(\epsilon,sv)=(s,c(sv)).\quad
\mathcal{A}((w,c),uv)=\begin{cases}
*&(w=v\wedge c=c(uv))\\
(w,c)&(w\neq v\wedge c=c(uv))\\
(u,c(uv))&(c\neq c(uv)).
\end{cases}
\]

Let $W=(v_0,\ldots,v_\ell)$ be a walk with $\ell>0$, and let $c=c(v_{\ell-1}v_\ell)$.
A suffix $(v_i,\ldots,v_\ell)$ is called the \emph{longest monochromatic suffix} of $W$ if $c(v_jv_{j+1})=c$ for every $j\geq i$ and $c(v_{i-1}v_i)\neq c$ or $i=0$ holds.
We can observe that the longest monochromatic suffix of $W$ starts from $v_i$ and has the color~$c$ if $\mathcal{A}^*(W)=(v_i,c)$,
and the longest monochromatic suffix of $W$ forms a monochromatic cycle if $\mathcal{A}^*(W)=*$.
Then, for a single-branching pair $(P,Q)$, we have $P\not\equiv Q\iff$
the simple cycle contained in $P\circ Q^{-1}$ is non-monochromatic
$\iff$ none of $P$ and $Q$ induces a monochromatic cycle and the longest monochromatic suffixes of $P$ and $Q$ start from different vertices or have different colors
$\iff$ $*\neq \mathcal{A}^*(P)\neq\mathcal{A}^*(Q)\neq *$.
Therefore, the abovementioned implementation is correct.
Thus, we can solve \textsc{Non-monochromatic Cycle Transversal} in $O(4^k km)$ time.

\bibliographystyle{abbrvurl}
\bibliography{tex/paper}

\newpage
\appendix


\section{Proof of Persistency}\label{sec:persistency}
First, we review the results of \cite{Iwata:2016ja}.
Let $D$ be a set containing a special element $\perp$, and let $D_I:=D\setminus\{\perp\}$.
Consider a poset on $D$ such that $\perp\ <a$ for every $a\in D_I$ and all the other pairs $(a,b)$ with $a\neq b$ are not comparable.
Let $a\sqcap b$ be a binary operator that returns the minimum of $a$ and $b$ if they are comparable and returns $\perp$ otherwise.
Similarly, let $a\sqcup b$ be a binary operator that returns the maximum of $a$ and $b$ if they are comparable and returns $\perp$ otherwise.
We define $\bm{a}\cdot\bm{b}:=(a_1\cdot b_1,\ldots,a_n\cdot b_n)$
for a binary operator $(\cdot)$ and tuples of elements $\bm{a}=(a_1,\ldots,a_n)$, $\bm{b}=(b_1,\ldots,b_n)\in D^n$.
Each index $i$ herein might have a distinct domain $D_i$ containing $\perp$, and $D^n$ is an abbreviation for $D_1\times D_2\times\ldots\times D_n$.
A function $f:D^n\to\mathbb{R}$ is called \emph{$k$-submodular} if the following inequality holds for any pair of inputs $\bm{a},\bm{b}\in D^n$:
\[
f(\bm{a})+f(\bm{b})\geq f(\bm{a}\sqcap\bm{b})+f(\bm{a}\sqcup\bm{b}).
\]
We denote the restriction of a function $f:D^n\to\mathbb{R}$ to the domain $D_I^n$ by $f_I$.

\begin{lemma}[\cite{Iwata:2016ja}]\label{lem:review:persistency}
For a $k$-submodular function $f:D^n\to\mathbb{R}$ and a minimizer $\bm{b}\in D^n$ of $f$,
there exists a minimizer $\bm{a}\in D_I^n$ of $f_I$ such that $a_i = b_i$ for every $i$ with $b_i \neq {\perp}$.
\end{lemma}
\begin{lemma}[\cite{Iwata:2016ja}]\label{lem:review:ksub}
The following three functions are $k$-submodular.
\begin{itemize}
  \item For a permutation $\pi$ on $D_I$, a function $p_\pi:D\times D\to\{0,\infty\}$ defined as
	\[
		p_\pi(x,y)=\begin{cases}
		0&\text{if }(x=y=\,\perp)\vee(x,y\in D_I\wedge \pi(x)=y),\\
		\infty&\text{otherwise}.
		\end{cases}
	\]
  \item For elements $a,b\in D_I$, a function $t_{a,b}:D\times D\to\{0,\infty\}$ defined as
  	\[
		t_{a,b}(x,y)=\begin{cases}
		0&\text{if }(x=y=\,\perp)\vee (x=a)\vee (y=b),\\
		\infty&\text{otherwise}.
		\end{cases}
  	\]
  \item A function $e:D^r\to\{0,\frac{1}{2},1\}$ defined as
  	\[
  		e(\bm{x})=\begin{cases}
  		0&\text{if }x_1=x_2=\cdots=x_r,\\
  		1&\text{if } \exists i,j \text{ such that } \perp\,\neq x_i\neq x_j\neq\, \perp,\\
  		\frac{1}{2}&\text{otherwise}.\\
  		\end{cases}
  	\]
\end{itemize}
\end{lemma}

\begin{proof}[Proof of Theorem~\ref{thm:persistency}]
Let $f:D^{\hat{E}}\to\mathbb{R}$ be a function obtained by taking the sum of the following functions.
Each variable $uv\in\hat{E}$ here has a domain $D(u)\cup\{\perp\}$.
We add the function $p_\pi$ on $(uv,vu)$ for each permutation constraint $\pi(\varphi(u))=\varphi(v)$ of $C$.
We add the function $t_{a,b}$ on $(uv,vu)$ for each two-fan constraint $(\varphi(u)=a)\vee(\varphi(v)=b)$ of $C$.
Meanwhile, we add the $d$-ary function $e$ on $\delta(v)$ for each variable $v\in V \setminus A$, where $d=|\delta(v)|$ is the degree of $v$.
Similarly for each variable $v \in A$, we add the $(d+1)$-ary function $e$ on $\delta(v)$ with one argument fixed to $\varphi(v)$ (hence it acts as a $d$-ary function), where $d = |\delta(v)|$.
Then, from Lemma~\ref{lem:review:ksub}, $f$ is $k$-submodular.

We first observe that we can convert an input $\bm{b} \in D^{\hat{E}}$ of $f$ with $f(\bm{b}) < \infty$ to a half-integral $\F_{C,\varphi_A}$-cover $x$ with $x(V)=f(\bm{b})$ and
convert a half-integral $\F_{C,\varphi_A}$-cover $x$ to an input $\bm{b}\in D^{\hat{E}}$ with $f(\bm{b})\leq x(V)$.
For each vertex $v$, we denote by $f_v(\bm{b})$ the value of the function $e$ for $v$, that is,
for $v\in V\setminus A$, $f_v(\bm{b}):=e(\bm{b}|_{\delta(v)})$, where $\bm{b}|_{\delta(v)}$ is the restriction of $\bm{b}$ to $\delta(v)$,
and for $v\in A$, $f_v(\bm{b}):=e(\bm{b}|_{\delta(v)},\varphi_A(v))$.
We have $f(\bm{b})=\sum_{v\in V}f_v(\bm{b})$ when $f(\bm{b})<\infty$.

Let $\bm{b} \in D^{\hat{E}}$ be an input of $f$ with $f(\bm{b}) < \infty$.
We construct $x$ by setting $x(v)=f_v(\bm{b})$ for each $v\in V$.
To see that $x$ is a (half-integral) $\F_{C,\varphi_A}$-cover, take a $\varphi_A$-conflicting walk $W = (v_0,\ldots,v_\ell)$ with $v_0,v_\ell \in A$.
Now, consider a sequence of variables for $f$ along the walk $W$:
\[
  v_0v_1, v_1v_0, v_1v_2, v_2v_1,\ldots,v_{\ell-2}v_{\ell-1},v_{\ell-1}v_{\ell-2}, v_{\ell-1}v_\ell,v_\ell v_{\ell-1}.
\]
Let $b_{v_0v_{-1}}=\varphi_A(v_0)$ and $b_{v_\ell v_{\ell+1}}=\varphi_A(v_\ell)$.
As the walk $P$ is $\varphi_A$-conflicting, we must have an index $0 \leq i \leq \ell$ such that $\perp\, \neq b_{v_iv_{i-1}} \neq b_{v_iv_{i+1}} \neq\, \perp$
or two indices $0\leq i<j\leq \ell$ such that $b_{v_iv_{i-1}}\neq\,\perp$, $b_{v_iv_{i+1}}=\,\perp$, $b_{v_{j-1}v_j}=\,\perp$, and $b_{v_jv_{j+1}}\neq\,\perp$.
In both cases, $x(V(P))\geq 1$.

Next, let $x$ be a half-integral $\F_{C,\varphi_A}$-cover.
Let $\varphi_{R(x)}$ be the satisfying assignment for $C[R(x)]$.
We then define an input $\bm{b}$ for $f$ as follows:
\[
  b_{uv} = \begin{cases}
  \varphi_{R(x)}(u) & \text{if }u \in R(x),\\
  C_{vu}(\varphi_{R(x)}(v)) & \text{if }u \not \in R(x), v \in R(x), \text{ and } C_{vu}(\varphi_{R(x)}(v))\neq\all, \\
  \perp & \text{otherwise.}
  \end{cases}
\]
We have $f(\bm{b})<\infty$ from the construction.
For every vertex $u\in V$ with $x(u)=0$, we have $b_{uv}=\varphi_{R(x)}(u)$ for every $uv\in \delta(u)$ if $u\in R(x)$.
Otherwise, we have $b_{uv}=\,\perp$ for every $uv\in\delta(u)$.
We have $f_v(\bm{b})=0$ in both cases.
For any vertex $u\in V$ with $x(u)=\frac{1}{2}$, there exist no edges $uv_1,uv_2\in\delta(u)$ such that $\perp\,\neq b_{uv_1}\neq b_{uv_2}\neq\,\perp$ because, otherwise,
there exists a $\varphi_A$-conflicting walk $W$ with $x(V(W))=x(u)=\frac{1}{2}<1$, which is a contradiction.
Thus, we have $f(\bm{b})\leq x(V)$.

Now, we prove the claim.
Let $x$ be a minimum half-integral $\F_{C,\varphi_A}$-cover.
We then construct a minimizer $\bm{b} \in D^{\hat{E}}$ of $f$ with $f(\bm{b})=x(V)$.
From Lemma~\ref{lem:review:persistency}, there is a minimizer $\bm{a} \in D_I^{\hat{E}}$ of $f_I$ such that $a_{uv} = b_{uv}$ for every $uv \in \hat{E}$ with $b_{uv} \neq \bot$.
From the construction above, this means the existence of the desired deletion set $X:=\{v\in V\mid f_v(\bm{a})=1\}$.
\end{proof}

\section{Axiomatic Model}\label{sec:axiomatic}
We introduce an equivalent formulation of conflicting/implicational walks.
The merit of this formulation is that checking whether a set of walks satisfies the conditions below is often easier than finding an explicit expression as 0/1/all constraints.

\begin{definition}\label{def:nice}
A pair $(\F, \F^*)$ of (possibly infinite) sets of walks in the same graph is called \emph{nice} if it satisfies the following conditions.
\begin{enumerate}
  \item $\F\subseteq \F^*$.
  \item $\F^*$ is closed under taking a prefix (i.e., for any walk $W\in\F^*$ and any prefix-subwalk $P$ of $W$, $P\in\F^*$ holds).
  \item For two walks $P,Q\in\F^*$ ending at the same vertex, we write $P\equiv Q$ if and only if $P\circ Q^{-1}\not\in \F$.
  		The relation $(\equiv)$ then becomes an equivalence relation, that is,
  		(1) $P\equiv P$,
  		(2) $P\equiv Q\iff Q\equiv P$, and
  		(3) $P\equiv Q\wedge Q\equiv R\implies P\equiv R$ hold for every walks $P,Q,R\in\F^*$ ending at the same vertex.
  \item For any equivalent walks $P,Q\in\F^*$ ending at $u$ and any edge $uv\in E$, $P\circ uv\in\F^*$ if and only if $Q\circ uv\in\F^*$.
\end{enumerate}
\end{definition}

For two walks $P,Q\in\F^*$ ending at the same vertex, we write $P\not\equiv Q$ if and only if $P\circ Q^{-1}\in \F$.
The notations $P\equiv Q$ or $P\not\equiv Q$ implicitly imply that $P,Q\in\F^*$ and $t(P)=t(Q)$.
Note that, from the definition, $P\in \F \iff P^{-1}\in \F$ always holds. However, $P\in\F^*\iff P^{-1}\in\F^*$ may not hold.
When $\F^*$ is the set of all walks starting from a set of vertices $A$,
$(\F,\F^*)$ is nice if and only if $\F$ is the set of non-returning $A$-walks (cf. \cite[pp.~109--111]{pap2006constructive}).

\begin{lemma}
Let $\F^*$ be the set of all implicational walks and let $\F$ be the set of all conflicting walks.
Then $(\F,\F^*)$ is nice.
\end{lemma}
\begin{proof}
Conditions 1, 2, and 4 are trivial from the definition, and we only prove Condition 3.
From Lemma~\ref{lem:def:imp}, for any two walks $P,Q\in\F^*$ ending at the same vertex, $P\equiv Q$ if and only if $\imp(P)=\imp(Q)$ holds.
Therefore, $(\equiv)$ is an equivalence relation.
\end{proof}

\begin{lemma}
For any nice pair $(\F,\F^*)$ for a graph $G=(V,E)$, there exist domains for $V$, 0/1/all constraints for $E$, and a partial assignment $\varphi_A$ for some $A \subseteq V$ such that
$\F$ is exactly the set of all $\varphi_A$-conflicting walks and $\F^*$ is exactly the set of all $\varphi_A$-implicational walks.
\end{lemma}
\begin{proof}
For each vertex $v\in V$, let $D(v)$ be the classes of equivalent walks in $\F^*$ ending at $v$.
We denote the class containing $W$ by $[W]\in D(v)$.
Let $A:=\{v\mid (v)\in\F^*\}$ and $\varphi_A(v):=[(v)]$.
For each edge $uv\in \hat{E}$, we define a subdomain $D(uv)\subseteq D(u)$ and a function $f_{uv}:D(uv)\to D(v)$ as follows.
\begin{enumerate}
  \item For any walk $W\in\F^*$ ending at $u$, $[W]\in D(uv)$ if and only if $W\circ uv\in\F^*$.
  \item For any walk $W\circ uv\in\F^*$, $f_{uv}([W])=[W\circ uv]$ holds.
\end{enumerate}

\begin{claim}
$D(uv)$ and $f_{uv}$ are well-defined.
\end{claim}
\begin{proof}
Let $W'\in [W]$. We have $W'\circ uv\in\F^*$ if $W\circ uv\in\F^*$ (by Condition 4).

Suppose that $W'\circ uv\not\equiv W\circ uv$, which implies $W'\not\equiv W\circ uv\circ vu$.
Meanwhile, by Condition 3-(1),
$W \circ uv \equiv W \circ uv$ holds and, hence, $W \circ uv \circ vu \equiv W$.
Thus, we have $W'\not\equiv W\circ uv\circ vu\equiv W$, which is a contradiction.
\end{proof}

\begin{claim}\label{claim:def:permutation}
If $f_{uv}([W])\in D(vu)$ holds for some class $[W]\in D(uv)$, we have $D(u)=D(uv)$, $D(v)=D(vu)$, and $f_{uv}=f_{vu}^{-1}$.
\end{claim}
\begin{proof}
We have $f_{vu}(f_{uv}([W]))=[W\circ uv\circ vu]=[W]$ because $W\circ uv\circ vu\equiv W$ holds.
Suppose that there exists a class $[W']\in D(u)\setminus D(uv)$ (i.e., $W' \circ uv \not\in \F^*$).

Since $[W'] \neq [W] \in D(uv)$,
we have $W'\not\equiv W \equiv W\circ uv\circ vu$,
implying $W'\circ uv\in\F^*$, which is a contradiction.
Therefore, we have $D(u)=D(uv)$.
By the same argument, we have $D(v)=D(vu)$.
Thus, the claim holds.
\end{proof}

\begin{claim}\label{claim:def:two-fan}
If $f_{uv}([W])\not\in D(vu)$ holds for some class $[W]\in D(uv)$, there exist classes $a\in D(u)$ and $b\in D(v)$ such that
$D(uv)=D(u)\setminus\{a\}$, $D(vu)=D(v)\setminus\{b\}$, $f_{uv}(a')=b$ for every class $a'\in D(uv)$, and $f_{vu}(b')=a$ for every class $b'\in D(vu)$.
\end{claim}
\begin{proof}
Let $b:=f_{uv}([W])\not\in D(vu)$.
Suppose that there exists a class $[W']\in D(uv)$ such that $f_{uv}([W'])\neq f_{uv}([W])$.
Because $W\circ uv\not\equiv W'\circ uv$ holds, we have $W\circ uv\circ vu\in \F^*$, which is a contradiction.
Therefore, we have $f_{uv}(a')=b$ for every class $a'\in D(uv)$.
Let $[Q]\in D(v)\setminus\{b\}$.
Because $Q\not\equiv W\circ uv$ holds, we have $Q\circ vu\in \F^*$.
Therefore, we have $D(vu)=D(v)\setminus\{b\}$.

Pick an arbitrary class $b'\in D(vu)$.
If $a:=f_{vu}(b')\in D(uv)$, from Claim~\ref{claim:def:permutation}, we have $D(vu)=D(v)$, which is a contradiction.
Therefore, by the same argument, we have $D(uv)=D(u)\setminus\{a\}$ and $f_{vu}(b')=a$ for every class $b'\in D(vu)$.
\end{proof}

Now, we prove the lemma.
We introduce a constraint for each edge $uv\in E$ as follows:
Pick an arbitrary class $c\in D(uv)$.
If $f_{uv}(c)\in D(vu)$ holds, from Claim~\ref{claim:def:permutation}, the edge $uv$ can be expressed as $f_{uv}(\varphi(u))=\varphi(v)$ for the permutation $f_{uv}$.
If $f_{uv}(c)\not\in D(vu)$ holds, from Claim~\ref{claim:def:two-fan}, the edge $uv$ can be expressed as $(\varphi(u)=a)\vee (\varphi(v)=b)$ for some classes $a\in D(u)$ and $b\in D(v)$.
\end{proof}

\end{document}